\documentclass[acmsmall,screen,review=false]{acmart}

\AtBeginDocument{%
  \providecommand\BibTeX{{%
    \normalfont B\kern-0.5em{\scshape i\kern-0.25em b}\kern-0.8em\TeX}}}

\setcopyright{acmcopyright}
\copyrightyear{0000}
\acmYear{0000}
\acmDOI{00.0000/0000000.0000000}

\acmJournal{TOCL}
\acmVolume{00}
\acmNumber{0}
\acmArticle{000}
\acmMonth{0}

\usepackage{booktabs}   
\usepackage{subcaption} 

\usepackage{setspace}

\usepackage{multicol,multicap}

\usepackage{amsmath,amssymb,amsthm,graphicx}

\usepackage{mathrsfs}
\usepackage{stmaryrd}

\usepackage{array}
\usepackage{booktabs}
\usepackage{multirow}

\usepackage{extarrows}
\UseRawInputEncoding

\newcommand{\Cc}{\mathbb{C}}
\newcommand{\Dd}{\mathcal{D}}
\newcommand{\Ee}{\mathbb{E}}
\newcommand{\Ff}{\mathcal{F}}

\newcommand{\Pp}{\mathcal{P}}

\newcommand{\Tt}{\mathcal{T}}

\newcommand{\Zz}{\mathbb{Z}}

\newcommand{\ket}[1]{|#1\rangle}
\newcommand{\bra}[1]{\langle#1|}
\newcommand{\inprod}[2]{\langle#1|#2\rangle}
\newcommand{\mixprod}[3]{\bra{#1}#2\ket{#3}}

\newcommand{\outprod}[2]{\ket{#1}\bra{#2}}
\newcommand{\leng}[1]{||#1||}
\newcommand{\Boolean}{\mathbf{Bool}}
\newcommand{\integer}{\mathbf{Int}}
\newcommand{\SKIP}{\mathbf{skip}}

\newcommand{\ifStat}[1]{\mathbf{if}\ #1\ \mathbf{fi}}
\newcommand{\whileStat}[2]{\mathbf{while}\ #1\ \mathbf{do}\ #2\ \mathbf{od}}
\newcommand{\WHILE}{\mathbf{while}}

\newcommand{\procedure}{\mathbf{Proc}}
\newcommand{\recDec}[2]{\procedure\ #1 \colon #2 }
\newcommand{\CALL}{\mathbf{call}}
\newcommand{\RELEASE}[1]{\mathbf{Rel}\ #1}
\newcommand{\lst}[1]{\bar{#1}}
\newcommand{\sem}[1]{\llbracket #1 \rrbracket}
\newcommand{\pcor}[3]{ \{ #1 \}\, #2\, \{ #3 \} }
\newcommand{\apcor}[4]{ #1\{ #2 #1\}\, #3\, #1\{ #4 #1\} }
\newcommand{\tcor}[3]{ \langle #1 \rangle\, #2\, \langle #3 \rangle }
\newcommand{\atcor}[4]{ #1\langle #2 #1\rangle\, #3\, #1\langle #4 #1\rangle }

\newcommand{\LOCAL}[1]{\mathbf{Loc}\ #1}
\newcommand{\IF}{\mathbf{if}}

\newcommand{\bottom}{\mathbf{bot}}
\newcommand{\qOperation}

\newcommand{\trueOrder}{{\mathbb I}_{\mathit \sqsubseteq}}
\newcommand{\trueEquality}{{\mathbb I}_{\mathit =}}

\newcommand\Hh{\mathcal{H}}

\newcommand{\hide}[1]{ }

\newcommand{\calB}{\mathcal{B}}
\newcommand{\calC}{\mathcal{C}}
\newcommand{\calD}{\mathcal{D}}
\newcommand{\calE}{\mathcal{E}}
\newcommand{\calF}{\mathcal{F}}
\newcommand{\calG}{\mathcal{G}}

\newcommand{\calO}{\mathcal{O}}
\newcommand{\calR}{\mathcal{R}}

\theoremstyle{definition} \newtheorem{example}{Example}[section]

\theoremstyle{plain}
\newtheorem{lemma}{Lemma}[section]
\newtheorem{proposition}{Proposition}[section]
\newtheorem{theorem}{Theorem}[section]

\newtheorem{claim}{Claim}
\newtheorem*{claim*}{Claim}

\theoremstyle{definition}
\newtheorem{definition}{Definition}[section]
\newtheorem{remark}{Remark}[section]

\makeatletter

\newcommand{\Rmnum}[1]{\expandafter\@slowromancap\romannumeral #1@}
\makeatother

\newcommand{\assnequal}{:=}
\newcommand{\starequal}{\; {\ast}{=}\; }

\newcommand{\qPL}{\mathit{qPL}}
\newcommand{\EqPL}{\mathit{EqPL}}
\newcommand{\qRP}{\mathit{RqPL}}

\newcommand{\EqRP}{\mathit{eRqPL}}
\newcommand{\Var}{\mathit{Var}}

\newcommand{\qSearch}{\mathit{qSearch}}
\newcommand{\qSearchDag}{\mathit{qSearch\_dag}}

\newcommand{\WP}{\mathit{WP}}
\newcommand{\WLP}{\mathit{WLP}}

\newcommand{\fwp}{\mathit{wp}}
\newcommand{\fwlp}{\mathit{wlp}}
\newcommand{\fxp}{\mathit{xp}}

\newcommand{\reVdash}[1]{\vdash_{\mathit{#1}}}
\newcommand{\qPD}{\mathit{qPD}}
\newcommand{\qBase}{\mathit{qBS}}
\newcommand{\qBE}{\mathit{qBE}}

\newcommand{\qPP}{\mathit{qPP}}
\newcommand{\bfH}{\mathbf{H}}

\newcommand{\qTP}{\mathit{qTP}}

\newcommand{\RFS}{\mathcal{RFS}}
\newcommand{\RQFS}{\mathit{RQFS}}

\newcommand{\itY}{\mathit{Y}}
\newcommand{\calX}{\mathcal{X}}
\newcommand{\bbX}{\mathbb{X}}
\newcommand{\bbY}{\mathbb{Y}}
\newcommand{\Prem}{\mathit{Prem}}
\newcommand{\voutprod}[3]{\ket{#1}_{#2}\bra{#3}}
\newcommand{\vinprod}[3]{\bra{#1}{#2}\ket{#3}}

\newcommand{\vpcor}[3]{ \begin{array}{c}
                          \{ #1 \} \\
                          #2 \\
                          \{ #3 \}
                        \end{array}
 }

\newcommand{\avpcor}[4]{ \begin{array}{c}
                          #1 \{ #2 #1 \} \\
                          #3 \\
                          #1 \{ #4 #1 \}
                        \end{array}
 }

\newcommand{\avtcor}[4]{ \begin{array}{c}
                          #1 \langle #2 #1 \rangle \\
                          #3 \\
                          #1 \langle #4 #1 \rangle
                        \end{array}
 }

 \newcommand{\proc}{\mathit{proc}}
 \newcommand{\tr}{\mathit{tr}}
 \newcommand{\QPL}{\mathrm{QPL}}
 \newcommand{\type}{\it type}

 \newcommand{\Alice}{\mathit{Alice}}
 \newcommand{\Bob}{\mathit{Bob}}

\newcommand{\PDOP}{\mathrm{PDO}}
 \newcommand{\QOP}{\mathrm{QOP}}

\newcommand{\QPRED}{\mathrm{QPRD}}
 \newcommand{\PQPT}{\mathrm{PQPT}}

\newcommand{\Params}[1]{\mathit{Prmt}(#1)}
  \newcommand{\RQMC}{\mathrm{RQMC}}
  \newcommand{\toy}{\mathit toy}
\newcommand{\MainProgram}{\mathit main}

  \newcommand{\pRule}{\rm{Rp}}
  \newcommand{\tRule}{\rm{Rt}}

\begin{document}

\title{Reasoning about Recursive Quantum Programs}

\subtitle{A new assertion logic for verifying quantum programs with probabilistic control}   


\author{Zhaowei Xu}
\affiliation{%
  \institution{Institute of Software, Chinese Academy of Sciences}
  \city{Haidian}
  \state{Beijing}
  \country{China}}
\affiliation{
  \institution{Laboratoire de Recherche en Informatique, Universit\'{e} Paris-Saclay}
  \city{Orsay}
  \country{France}}
\email{zhaowei@lri.fr}

\author{Mingsheng Ying}
\affiliation{%
  \institution{Institute of Software, Chinese Academy of Sciences}
  \city{Haidian}
  \state{Beijing}
  \country{China}}
\affiliation{%
  \institution{University of Technology Sydney}
  \country{Australia}}
\affiliation{%
  \institution{Tsinghua University}
  \city{Haidian}
  \state{Beijing}
  \country{China}}
\email{mingshengying@gmail.com}

\author{Beno\^{\i}t Valiron}
\affiliation{%
  \institution{\'{E}cole CentraleSup\'{e}lec}
  \city{Orsay}
  \country{France}}
\affiliation{
  \institution{Laboratoire de Recherche en Informatique, Universit\'{e} Paris-Saclay}
  \city{Orsay}
  \country{France}}
\email{benoit.valiron@lri.fr}

\renewcommand{\shortauthors}{Xu, et al.}

\begin{abstract}
Most modern (classical) programming languages support recursion.
Recursion has also been successfully applied to the design of several quantum algorithms
and introduced in a couple of quantum programming languages.
So, it can be expected that recursion will become one of the fundamental paradigms of quantum programming.
Several program logics have been developed for verification of quantum $\WHILE$-programs.
However, there are as yet no general methods for reasoning about (mutual) recursive procedures
and ancilla quantum data structure in quantum computing (with measurement).
We fill the gap in this paper by proposing a parameterized quantum assertion logic
and, based on which, designing a quantum Hoare logic
for verifying parameterized recursive quantum programs with ancilla data and probabilistic control.
The quantum Hoare logic can be used to prove partial, total, and even probabilistic correctness
(by reducing to total correctness) of those quantum programs.
In particular, two counterexamples for illustrating
incompleteness of non-parameterized assertions in verifying recursive procedures,
and, one counterexample for showing the failure of reasoning with exact probabilities based on partial correctness,
are constructed.
The effectiveness of our logic is shown by three main examples ---
recursive quantum Markov chain (with probabilistic control), fixed-point Grover's search, and recursive quantum Fourier sampling.
\end{abstract}



\maketitle

\section{Introduction}

\subsection{Background and motivation}

Quantum computation nowadays has become a hot topic in computer science.
One of the fundamental incentives of this research direction is
to have successfully designed several quantum algorithms, particularly,
Shor's algorithm \cite{shor1994algorithms} and Grover's search algorithm \cite{grover1996fast},
which by employing the intriguing and unnatural effects of quantum mechanics,
e.g. superposition and entanglement,
can obtain significant computational advantages.
The design of quantum algorithms is mainly based on the slogan ``quantum data and classical control",
that is, that the data could be superposed and even entangled,
which can be manipulated by basic quantum operations --- unitary evolution and measurement,
but the high-level control is still classical
(either deterministic or probabilistic, e.g. case, loops, etc).
(For recent discussions about quantum control,
i.e. superposition of quantum programs or superposition of quantum processes,
see, e.g., \cite{BadescuP15}, \cite{chiribella2012perfect,chiribella2013quantum},
and Chaps. 6 and 7 of  \cite{ying2016foundations}.)

\paragraph{Quantum programming with recursion.}
Classical recursion, as a high-level control structure,
has been applied to quantum algorithm design,
and brought substantial advantages to quantum computation.
Instead of speaking about quantum algorithms with while-loop control
(as examples of tail recursion, say,
Shor's algorithm \cite{shor1994algorithms} and Grover's search algorithm \cite{grover1996fast}),
we shall exemplify quantum algorithms with general recursion.
The first representative is Grover's fixed-point search algorithm \cite{grover2005fixed,yoder2014fixed},
which, by applying (mutual) recursion,
provides his famous/popular original search algorithm \cite{grover1996fast} with an advantage
--- converging monotonically to the target state.
As another typical example,
recursive quantum Fourier sampling \cite{bernstein1997quantum} requires exponentially fewer queries than the classical one,
and has extensive applications in research on quantum complexity theory
(cf. the Introduction of \cite{mckague2012interactive}).
This quantum algorithm is described by a recursive procedure with pointer passing and ancilla qubits.
As a third example, recursive quantum Markov chain is
a quantum extension of Etessami and Yannakakis's Recursive Markov chains \cite{etessami2009recursive},
and can be used to simulate a multi-player game with probabilistic control \cite{feng2013reachability}
(to be the running example of this paper).

Implementation of quantum algorithms, i.e. quantum programming,
has been extensively investigated for the past two decades (also following the slogan ``quantum data and classical control"),
including both high-level (imperative) and low-level (functional) programming languages
and their semantics \cite{altenkirch2005functional,omer2003structured,sabry2003modeling,sanders2000quantum,Selinger2004},
as surveyed in \cite{Selinger2004, gay2006quantum, ying2016foundations}.
In particular,
recursive procedures with pointer passing have already been introduced by Selinger
in his high-level quantum programming language $\QPL$ \cite{Selinger2004}.
The literature \cite{ying2016foundations} defined
a quantum $\WHILE$-language with recursion and local variables (to describe ancilla quantum data).
In the last few years,
a number of mature low-level quantum programming languages have been developed,
e.g., Quipper \cite{green2013quipper}, Scaffold \cite{abhari2012scaffold},
LIQUi$|\rangle$ \cite{wecker2014liqui}, Q$\#$ \cite{svore2018q}, and QWIRE \cite{paykin2017qwire}.

\paragraph{Verification of quantum programs.}
In view of counter-intuitiveness of quantum effects,
quantum programming is an inherently error-prone process.
To ensure correctness and safety of quantum systems,
formal verification and formal program analysis (static approaches) is a right choice
\cite{baltag2011quantum,brunet2004dynamic,ying2017invariants, li2018termination, liu2019formal},
relative to running-time (dynamic) approaches like testing and debugging,
due to a series of frustrating facts:
quantum states fail to be directly observed before measurement;
the current state could be potentially destroyed after measurement;
and the measurement result is randomly distributed.
Among those formal methods on quantum programming,
we prefer program-logic-based approaches, e.g. quantum program logic \cite{baltag2011quantum,brunet2004dynamic},
compared with model-based methods, e.g. quantum model checking \cite{gay2008qmc,feng2013model,feng2015qpmc,ying2014model},
since the former is developed in a syntax-oriented style
making it easily extensible to various program features, e.g. recursion, non-determinism, parallelism, etc.

Hoare logic has been the fundamental method for formal program verification \cite{Apt}.
The basic idea of this logic is based on the intermediate assertion method \cite{floyd1967assigning,hoare1969axiomatic},
which was originated with Alan Turing \cite{apt2019fifty} (called Turing-Floyd-Hoare Principle).
Hoare's approach makes (interactive) theorem proving for verifying high-level algorithmic description language
proceed at the same abstraction level as the language itself.
Thus verification using Hoare logic is more human-friendly than low-level (machine-friendly) verification.
Several Hoare-like logics for reasoning about quantum programs have been developed \cite{chadha2006reasoning,feng2007proof,kakutani2009logic,Ying11,Unruh19a,Unruh19b,BartheHYYZ20}.
Among them, D'Hondt and Panangaden \cite{panangaden2006}
proposed a notion of quantum weakest precondition for a general quantum operation.
The attractiveness of this approach is that quantum predicates,
used as preconditions and postconditions,
are modelled by Hermitian operators and thus have a natural interpretation as physical observables.
Based on this, a quantum Hoare-like logic
for reasoning about quantum $\WHILE$-language was designed
and its (relative) completeness was established, by Ying in \cite{Ying11}.
In these quantum Hoare logics, we find that
\begin{enumerate}
  \item Quantum predicate can merely describe a fixed property on quantum states,
  and the mechanism for guaranteeing termination is not designed in a syntactical style.
  \item These programming languages don't support general recursion,
  and there is no compositional inference rule for the structure of ancilla quantum variables.
  \item These logics are mainly for verifying deterministic properties of quantum programs,
  and at most can do reasoning with approximate probabilities.
  In other words, there is no axiomatic basis for reasoning about quantum programs with exact probabilities.
\end{enumerate}

\subsection{Contributions of the Paper}

The aim of this paper is to propose an assertion logic for
verifying parameterized recursive quantum programs with ancilla qubits or quints
(the abbreviation of quantum int).
By using formulas of this assertion logic as pre- and post-conditions,
a quantum Hoare logic is designed for proving partial correctness, total correctness,
and even probabilistic correctness (including probabilistic termination) of those quantum programs.
The work extends D'Hondt and Panangaden's quantum predicates and quantum weakest preconditions \cite{panangaden2006},
and Ying's quantum Hoare logic \cite{Ying11}.
Concretely speaking, five main contributions are highlighted:
\begin{itemize}
\item We extends the syntax of quantum $\WHILE$-language defined in \cite{Ying11} with general recursive procedures.
    For the formal semantics of the extended quantum programming language,
    we define {\it a nondeterministic operational semantics} by introducing the concept of labeled transition relation,
    define {\it a quantum-operation-directed denotational semantics} (independent of program states),
    and {\it relate them to each other} (cf. Thm. \ref{thm_sem}).
    Note that the denotation of a recursive procedure can be defined as
    the least fixed point of a function over quantum operations,
    making the representation of this denotation have a closed form.
    The formal semantics can be used to describe behaviors of quantum programs with probabilistic control
    (cf. Exams. \ref{exam_opSem_RQMC} and \ref{exam_deSem_RQMC}).

\item By {\it constructing two counterexamples}
    (cf. Exm. \ref{exam_counter_par} for partial correctness, and Exam. \ref{exam_counter_tot} for total correctness),
    we illustrate {\it the failure of Hoare's approach (i.e. intermediate assertion method) with quantum predicates as pre- and post-conditions in verifying recursive procedures}.
    To extend the applicability of Hoare's approach from while loops to general recursion
    (cf. Rem. \ref{remark_tail_rec}),
    we have to {\it introduce a parameterized quantum assertion logic}
    extending quantum predicates by incorporating parameters.
    In the setting of this assertion logic,
    we redefine the notion of L\"{o}wner order and upper (resp. lower) limit of quantum predicates,
    quantum program correctness and expressiveness of intermediate assertions
    developed in \cite{panangaden2006,Ying11} (cf. Thms. \ref{thm_sem_dual} and \ref{thm_exp}).

\item We introduce {\it inference rules for proving both partial and total correctness of recursive procedures}
      with formulas of the new assertion logic as pre- and post-conditions.
     These rules usually should be used in combination with the {\it Substitution Rule,
     dealing with substitution for parameters in pre- and post-conditions of a Hoare's triple}
     (cf. Exms. \ref{exam_counter_par} and \ref{exam_counter_tot} for counterexamples).
     The rules of proving total correctness can be adapted to {\it proving probabilistic correctness,
     i.e. reasoning with both approximate and exact probabilities
     (cf. Thms. \ref{thm_reas_approx_prob} and \ref{thm_reas_exact_prob}),
     based on the result of (general or compact) soundness and completeness}.
     However, {\it the rules of proving partial correctness can't be used universally for reasoning with exact probabilities,
     even if the issue of nontermination is involved (cf. Rem. \ref{rem_exam_counter_prob} for a counterexample)}.
    Note that the rule for proving total correctness is purely syntactic,
    and, as a special case, we obtain a syntax-directed inference rule for proving total correctness of while loops.
    We also discuss the issues of synthesizing intermediate assertions and necessity of parameters
    in verifying recursive procedures and while loops.
    These discussions reveal that recursion is generally more complex than while loops in the setting of program logics.

\item Verification of some more sophisticated recursive quantum programs
    like the example of recursive quantum Fourier sampling cannot be done by merely using the above techniques;
    they further need the facilities of variable localization and parameter passing.
    So, our fourth main contribution is to develop
    {\it inference rules for proving both partial and total correctness of variable localization and recursive quantum procedures with pointer passing.}
    To this end,
    we propose {\it two different but equivalent inference rules for proving the two correctness of variable localization},
    and the {\it Adaptation Rule for dealing with the substitution of program variables in pointer passing}.
    Note that previous proof rules for recursive procedures should be used jointly with Adaptation Rule in verifying parameterized recursive procedures.
    The proof rule for variable localization endows the logic
    with the ability of verifying a (general) quantum operation in a compositional way.
    Various aforementioned results,
    like Thms. \ref{thm_sem}, \ref{thm_exp}, \ref{thm_reas_approx_prob}, \ref{thm_reas_exact_prob}
    (and soundness and completeness results), can be extended to covering these facilities.

    \item The fifth and last main contribution is to present various examples and related work.
    Among these examples, we adopt recursive quantum Markov chain as the running example of the paper,
    since it can fully show the ability of our logic in dealing with quantum programs with probabilistic control.
    Specifically speaking, this example illustrates operational semantics, denotational semantics,
    and reasoning with exact probability (including probabilistic correctness and probabilistic termination).
    The examples of Grover's fixed-point search and recursive quantum Fourier sampling
    are used to illustrate rules for proving partial and total correctness
    of recursive quantum programs with deterministic control (containing auxiliary facilities).
    We compare our quantum Hoare logic with other (deterministic, probabilistic or quantum) Hoare-like logics,
    and discuss local and global reasoning in the setting of quantum computing.
\end{itemize}

\subsection{Organization of the Paper}

We present preliminaries on quantum program verification in Sec. \ref{sec_pre};
syntax and semantics of recursive quantum programs are defined in Sec. \ref{sec_rqp};
quantum assertion logic is presented in Sec. \ref{sec_ass};
starting proof systems are shown in Sec. \ref{sec_prf};
expanded proof systems including auxiliary facilities are shown in Sec. \ref{sec_loc};
Case studies of Grover's fixed-point search and recursive quantum Fourier sampling are presented in Sec. \ref{sec_castu};
Comparison with the related work is given in Sec. \ref{sec_relwk};
Sec. \ref{sec_con} concludes the paper with a discussion of the future work.

The running example of recursive quantum Markov chain is throughout the paper from Subsec. \ref{subsec_syntax}
through Subsecs. \ref{subsec_op_sem} and \ref{subsec_de_sem} to Subsec. \ref{subsec_prob_corr}.
The proof (or proof schetch) of various results,
including proof for the two counterexamples --- Exms. \ref{exam_counter_par} and \ref{exam_counter_tot},
proof of soundness and completeness results,
and a full verification of Grover's fixed-point search and recursive quantum Fourier sampling,
are put into the appendix.

\section{Preliminaries}\label{sec_pre}

For convenience of the reader,
we briefly review the basics of quantum theory,
and fix the symbols and notations used in the subsequent sections.

\subsection{Quantum states}

\paragraph{Definition of linear operators}
The state space of a quantum system is a Hilbert space $\Hh$.
For any positive integer $n$,
an $n$-dimensional Hilbert space is essentially the space $\Cc^n$ of complex vectors.
We use Dirac's notation, $\ket{\psi}$, to denote a complex vector in $\Cc^n$.
The inner product (resp. outer product) of two vectors $\ket{\psi}$ and $\ket{\phi}$,
denoted $\inprod{\psi}{\phi}$ (resp. $\outprod{\psi}{\phi}$),
is the product of $\bra{\psi} \triangleq (\ket{\psi})^{\dag}$ (i.e. the conjugate transpose of $\ket{\psi}$)
and $\ket{\phi}$ (resp. the product of $\ket{\psi}$ and $\bra{\phi}$).
The norm of a vector $\ket{\psi}$ is denoted by $\leng{\psi} \triangleq \sqrt{\inprod{\psi}{\psi}}$.
We say that a set of vectors $\{ \ket{\psi_i} \}_i$ is an orthonormal basis of $\Hh$,
if $\inprod{\psi_i}{\psi_j} = \delta_{i,j}$ for all $i,j$ (where $\delta_{i,j} = 1$ if $i = j$, and $= 0$ otherwise).
Then every vector of $\Hh$ can be represented as a linear combination of any orthonormal basis of $\Hh$.
We define (linear) operators over $\Hh$ as a linear mapping.
In the space $\Cc^n$, an operator $A$ is represented by an $n\times n$ matrix.
We say that $A$ is Hermitian, if $A = A^\dag$ (where $A^\dag$ denotes the conjugate transpose of $A$).
Let $I_{\Hh}$ (resp. $0_{\Hh}$) be the identity (resp. zero) operator over $\Hh$.
The trace of an operator $A$ is defined by $\tr(A) \triangleq \Sigma_i \bra{i} A \ket{i}$
(the sum of entries on the main diagonal of $A$ w.r.t. any orthonormal basis of $A$).

\begin{lemma}[Spectral decomposition, cf. {\cite[Box 2.2]{nielsen2000quantum}}]\label{lem_spec_dec}
Every linear operator $A$ of Hilbert space $\Hh$ is Hermitian,
if, and only if, it can be decomposed as $A = \sum p_i \outprod{\psi_i}{\psi_i}$,
where $\{ p_i \}_i$ are reals and $\{ \ket{\psi_i} \}_i$ is an orthonormal basis of $\Hh$.
\end{lemma}

\paragraph{L\"{o}wner order between linear operators}
An operator $A$ is positive, if for all vectors $\ket{\psi}\in \Hh$, $\mixprod{\psi}{A}{\psi}\geq 0$.
Note that every positive operator is Hermitian,
and that every Hermitian operator of the form $\sum p_i \outprod{\psi_i}{\psi_i}$
is positive iff $p_i \geq 0, \forall i$ (cf. Lem. \ref{lem_spec_dec}).
The concept of positivity induces the L\"{o}wner order $\sqsubseteq$ between operators:
\begin{itemize}
  \item $A \sqsubseteq B$, if $B - A$ is positive;
  \item $A = B$, if $A \sqsubseteq B$ and $B \sqsubseteq A$.
\end{itemize}
By definition, it follows, for any operators $A$ and $B$, that
\begin{itemize}
  \item $0_{\Hh} \sqsubseteq A$ iff $A$ is positive;
  \item $A \sqsubseteq B$ iff there is a positive operator $C$ s.t. $A + C = B$.
\end{itemize}
The least upper bound L. U. B. (resp. greatest lower bound G. L. B.) operator in a complete partial order
generated by L\"{o}wner comparison is denoted as $\bigsqcup$ (resp. $\bigsqcap$).
For example, the L. U. B. (resp. G. L. B.) of a sequence of operators $\{A_n\}_{n\geq 0}$
with $\forall n \geq 0.\ A_n \sqsubseteq A_{n+1}$ (resp. $\forall n \geq 0.\ A_n \sqsupseteq A_{n+1}$)
will be denoted by $\bigsqcup_{n \geq 0} A_n$ (resp. $\bigsqcap_{n \geq 0} A_n$).
For the existence of those bounds, the reader is referred to the literature \cite{Selinger2004} and \cite{ying2016foundations}.

\paragraph{Pure quantum state}
A pure quantum state is represented by a unit vector,
i.e., a vector $\ket{\psi}$ with $\leng{\psi}=1$
(used to represent the data states of a quantum circuit).
For example, a qubit, or quantum bit, system refers to the case when $\Hh = \Cc^2$.
An important basis of a qubit system is the computational basis
with $\ket{0} \triangleq (1, 0)^\dag$ and $\ket{1} \triangleq (0, 1)^\dag$,
which corresponds to the $0/1$ in a classical bit.
Another important basis, called the $\pm$ basis,
consists of $\ket{\pm} \triangleq \frac{1}{\sqrt{2}}(\ket{0}\pm\ket{1})$.
One can represent multi-qubits by tensor-producting each qubit.
For instance, the classical two-bit string $01$ can be represented by $\ket{0}\otimes \ket{1}$ (or $\ket{01}$ for short).
An $m$-qubit system lives in the space $\Cc^{2^m} = (\Cc^2)^{\otimes m}$ that is the $m$-time tensor product of a single qubit system $\Cc^2$.

\paragraph{Mixed quantum state}
However, after applying a quantum measurement,
a (pure) quantum state is possibly changed to a mixed state,
i.e. a random distribution over an ensemble of pure states $\Ee = \{(p_i, \ket{\psi_i})\}_i$,
which states that the system is in state $\ket{\psi_i}$ with probability $p_i$.
One can also use density operators to represent both pure and mixed quantum states.
Formally, a density operator is a positive operator $\rho$ whose trace $\tr(\rho) = 1$.
For example, the density operator $\rho$ for a mixed state represented by the
ensemble $\Ee$ is $\rho = \sum_i p_i \outprod{\psi_i}{\psi_i}$; in particular,
a pure state $\ket{\psi}$ can be identified with the density operator $\rho = \outprod{\psi}{\psi}$.

\paragraph{Representation of quantum states}
If density operators are used directly to represent quantum states,
we will find that the resulting state after applying a quantum operation
is probably not a density operator but its sub-part.
Note that the missing part of the final state is due to non-termination.
To cope with this issue, the concept of partial density operator (abbr. $\PDOP$)
is introduced by Selinger \cite{selinger2004towards} to model a sub-part of a density operator.
Strictly speaking, a $\PDOP$ is a positive operator $\rho$ with $0 \leq \tr(\rho) \leq 1$
(Particularly, a density operator is a $\PDOP$ $\rho$ with $\tr(\rho) = 1$).
Defining quantum states as $\PDOP$s ensures that quantum states are closed under quantum operations.
The set of $\PDOP$s on $\Hh$ is denoted by $\Dd(\Hh)$.

\subsection{Quantum operations}

\paragraph{Definition of quantum operations}
The evolution of an open quantum system $\Hh$
can be characterized by an (admissible) quantum operation (abbr. $\QOP$) $\calE$,
which is a linear, trace-non-increasing and completely positive super operator from $\Dd(\Hh)$ to $\Dd(\Hh)$
(By complete positivity of $\calE$ on $\Hh$ is meant that
for all linear operators $A$ on $\Hh'\supseteq \Hh$ with $A \sqsupseteq 0_{\Hh'}$, $\calE(A) \sqsupseteq 0_{\Hh'}$).
Namely, for any state $\rho \in \Dd(\Hh)$,
the final state after the $\QOP$ $\calE$ is $\calE(\rho) \in \Dd(\Hh)$ with $\tr(\calE(\rho)) \leq \tr(\rho)$.
For every $\QOP$ $\calE$,
there exists a set of Kraus operators $\{E_k \}_k$ s.t.
$\calE(\rho) = \sum_k E_k \rho E_k^\dag, \forall \rho$
(See, e.g., \cite{nielsen2000quantum}).
We denote the Kraus form of $\calE$ by writing $\calE = \sum_k E_k \diamond E_k^\dag$.
Since $\calE$ is positive and trace-non-increasing, it holds that $0 \sqsubseteq \sum_k E_k^\dag E_k \sqsubseteq I$.
For example, an identity (resp. zero) operation refers to $I_{\Hh} \diamond I_{\Hh}$ (resp. $0_{\Hh} \diamond 0_{\Hh}$).

\paragraph{Dual of quantum operation}
The Schr\"{o}dinger-Heisenberg dual of a $\QOP$ $\calE$, denoted $\calE^*$, is defined as
\begin{itemize}
  \item $\tr\big(A \calE(\rho)\big) = \tr\big(\calE^*(A)\rho\big)$, $\forall \rho$ and $\forall A$;
  \item or, in the Kraus form, $\calE^* \triangleq \sum_k E_k^\dag \diamond E_k$, if $\calE = \sum_k E_k \diamond E_k^\dag$.
\end{itemize}
\begin{lemma}\label{lem_dual_qop}
Let $\lambda \geq 0$,
let $\calE_1$ and $\calE_2$ be respective $\QOP$s on $\Hh_1$ and $\Hh_2$,
let $\calF$, $\calG$ be $\QOP$s on $\Hh$,
and let $\{\calF_n\}$ be a non-decreasing sequence of $\QOP$s on $\Hh$.
It is the case that
\begin{description}
  \item[(1)] $(\calE_1 \otimes \calE_2)^* = \calE_1^* \otimes \calE_2^*$;
  \item[(2)] $(\lambda \calF)^* = \lambda \calF^*$;
  \item[(3)] $(\calF + \calG)^* = \calF^* + \calG^*$;
  \item[(4)] $(\calF \circ \calG)^* = \calG^* \circ \calF^*$;
  \item[(5)] $(\bigsqcup_{n=0}^{\infty} \calF_n)^* = \bigsqcup_{n=0}^{\infty} \calF_n^*$.
\end{description}
\end{lemma}

\paragraph{Representation of unitary operators}
Operations (or evolutions) on (closed) quantum systems can be characterized by a unitary operator.
An operator $U$ is a unitary operator if its conjugate transpose is its own inverse, i.e., $U^\dag U = U U^\dag = I$.
Common single-qubit unitary operators include the Hadamard operator $H$ and the Pauli operator $X$:
$$
H \triangleq \frac{1}{\sqrt{2}} \left[ \begin{array}{cc}
                                1 & 1 \\
                                1 & -1
                              \end{array}
 \right],
\quad
X \triangleq \left[ \begin{array}{cc}
                                0 & 1 \\
                                1 & 0
                              \end{array}
 \right]
$$
Intuitively, $H$ transforms between the computational and the $\pm$ basis,
i.e., $H\ket{0} = \ket{+}$ and $H\ket{1} = \ket{-}$,
and $X$ is a bit flip, i.e., $X \ket{0} = \ket{1}$ and $X \ket{1} = \ket{0}$.
The evolution of $\PDOP$ $\rho$ under unitary operator $U$ is $\rho \rightarrow U\rho U^\dag$,
and thus $U$ can be written as a $\QOP$ $\calE = U \diamond U^\dag$.

\paragraph{Representation of measurement}
The way to extract information about a quantum system is called a quantum measurement.
Mathematically, a quantum measurement on a system over $\Hh$
can be described by a set of linear operators $\{M_m\}_m$ with $\sum_m M_m^\dag M_m = I_{\Hh}$.
If we perform a measurement $\{M_m\}_m$ on a state $\rho$,
the outcome $m$ is observed with probability $p_m \triangleq \tr(M_m \rho M_m^\dag)$ for each $m$,
and, with the observation $m$, the state collapses to a post-measurement state $M_m \rho M_m^\dag / p_m$.
To characterize the evolution of a quantum measurement as a $\QOP$,
we remark that the probability $p_m$ can be encoded into the post-measurement state $M_m \rho M_m^\dag / p_m$,
resulting in $M_m \rho M_m^\dag$.
So, such an evolution can be written as the $\QOP$ $\calE = M_m \diamond M_m^\dag$.
One of the major differences between classical and quantum computation
is that a quantum measurement could potentially change the state itself.
For example, if we perform the standard (i.e. computational-basis) measurement
$M = \{M_0 \triangleq  \outprod{0}{0}, M_1 \triangleq \outprod{1}{1}\}$ on state $\rho = \outprod{+}{+}$,
then with probability $\frac{1}{2}$ the outcome is $0$ (resp. $1$),
and the final state becomes $\outprod{0}{0}$ (resp. $\outprod{1}{1}$).

\paragraph{Representation of open systems}
Suppose that we have a joint quantum system $\Hh \triangleq Q \otimes R$
and wish to trace out system $R$ by using partial trace function $\tr_{R}$.
Let $\{\ket{i}\}_{i}$ be an orthonormal basis for $R$.
Then $\tr_{R}$ can be defined by the $\QOP$ $\tr_{R} \triangleq \sum_i \bra{i} \diamond \ket{i}$.
We can represent a (separable) joint $\QOP$ for quantum systems $Q$ and $R$
by the tensor product $\calE_{Q} \otimes \calE_{R}$ of $\QOP$s $\calE_{Q}$ for $Q$ and $\calE_{R}$ for $R$.
It is a convention in the quantum information literature that
when operations only apply to part of a quantum system,
one should assume that an identity operation is applied on the rest.
For example, applying $\calE_{R}$ (resp. $\calE_{Q}$)
to $\rho \in \Dd(Q \otimes R)$ means applying $(I_{Q}\diamond I_{Q}) \otimes \calE_{R}$
[resp. $\calE_{Q} \otimes (I_{R} \diamond I_{R})$] to $\rho$.
Here the identity operation is usually omitted for simplicity.
Thus, writing $\calE(\rho)$ with $\QOP$ $\calE$ on $\Hh$ and $\PDOP$ $\rho \in \Dd(\Hh')$
will naturally imply that $\Hh \subseteq \Hh'$.

\paragraph{L\"{o}wner order between super operators}
Define the L\"{o}wner order between $\QOP$s based on that between linear operators:
\begin{itemize}
  \item $\calE \sqsubseteq \calF$, if $\calE(\rho) \sqsubseteq \calF(\rho)$, $\forall \rho$;
  \item $\calE = \calF$, if $\calE \sqsubseteq \calF$ and $\calF \sqsubseteq \calE$.
\end{itemize}
By definition, we have, for any $\QOP$s $\calE$ and $\calF$, that
\begin{itemize}
  \item $0\diamond 0 \sqsubseteq \calE$ iff $\calE$ is completely positive;
  \item $\calE \sqsubseteq \calF$ iff there is $\QOP$ $\calG$ s.t. $\calE + \calG = \calF$.
\end{itemize}
The least upper bound $\bigsqcup_{n \geq 0} \calE_n$ (resp. greatest lower bound $\bigsqcap_{n \geq 0} \calE_n$)
of a sequence of $\QOP$s $\{ \calE_n \}_{n \geq 0}$ with $\forall n \geq 0.\ \calE_n \sqsubseteq \calE_{n+1}$
(resp. $\forall n \geq 0.\ \calE_n \sqsupseteq \calE_{n+1}$) is defined as
\begin{itemize}
  \item $\big( \bigsqcup_{n \geq 0} \calE_n \big)(\rho) \triangleq \bigsqcup_{n \geq 0} \calE_n(\rho)$, $\forall \rho$;
  \item $\big( \bigsqcap_{n \geq 0} \calE_n \big)(\rho) \triangleq \bigsqcap_{n \geq 0} \calE_n(\rho)$, $\forall \rho$.
\end{itemize}

\subsection{Quantum predicates}

\paragraph{Definition of quantum predicates}
As defined in \cite{panangaden2006},
a quantum predicate (abbr. $\QPRED$) on a Hilbert space $\Hh$ is a Hermitian operator $M_{\Hh}$
such that $0_{\Hh} \sqsubseteq M_{\Hh} \sqsubseteq I_{\Hh}$.
The satisfiability of a state (i.e. $\PDOP$) $\rho \in \Dd(\Hh')$ in the $\QPRED$ $M_{\Hh}$ with $\Hh \subseteq \Hh'$ is defined by the trace $\tr(M_{\Hh} \rho)$ if $\Hh = \Hh'$,
and $\tr\big( (M_{\Hh}\otimes I_{\Hh''}) \rho \big)$ if $\Hh \otimes \Hh'' = \Hh'$.
Intuitively, $\tr(M_{\Hh} \rho)$ \big[resp. $\tr\big( (M_{\Hh}\otimes I_{\Hh''}) \rho \big)$\big]
is the expectation of the truth value of predicate $M_{\Hh}$ in state $\rho$.
Note that restricting $M_{\Hh}$ to between $0_{\Hh}$ and $I_{\Hh}$
ensures that $0 \leq \tr(M_{\Hh} \rho) \leq 1$
\big[resp. $0 \leq \tr\big( (M_{\Hh}\otimes I_{\Hh''}) \rho \big) \leq 1$\big]
for any $\rho \in \Dd(\Hh')$.
We shall write $\Pp(\Hh)$ for the set of $\QPRED$s on $\Hh$.

\paragraph{How to use quantum predicates}
By Lem. \ref{lem_spec_dec}, every $\PDOP$ $\rho$ has the spectral decomposition
$\rho \triangleq \sum_{i\in A} p_i \outprod{\varphi_i}{\varphi_i}$,
where $\forall i\in A.\ 0 \leq p_i \leq 1$,
$\sum_{i\in A} p_i \leq 1$ and $\sum_{i\in A} \outprod{\varphi_i}{\varphi_i} = I$.
In practice, we prefer to use projection operators, e.g.
$M \triangleq \sum_{i\in B} \outprod{\varphi_i}{\varphi_i}$,
with $B \subseteq A$, to define the (precise) properties of $\rho$.
Intuitively speaking,
$\tr(M \rho) = \sum_{i\in B} p_i$ is the probability of $\rho$ falling into the subspace represented by $M$.
In particular, $\tr(\outprod{\varphi_i}{\varphi_i} \rho) = p_i$ is the probability of $\rho$ in the state $\outprod{\varphi_i}{\varphi_i}$, which can be used to reveal internal ingredients of a quantum state;
$\tr(I \rho) = \tr(\rho)$ is the probability of $\rho$ falling into the whole space,
which is one of the usual statistical properties on quantum states.
For example, the $\PDOP$ $\rho \triangleq \frac{\outprod{0}{0} + \outprod{1}{1}}{2}$
represents the resulting state $\{(\frac{1}{2}, \ket{0}), (\frac{1}{2}, \ket{1})\}$ after a standard measurement on $\ket{+}$. Note that $\tr(\outprod{0}{0} \rho) = \frac{1}{2}$ (resp. $\tr(\outprod{1}{1} \rho) = \frac{1}{2}$) is the probability of $\rho$ in the state $\outprod{0}{0}$ (resp. $\outprod{1}{1}$), $\tr(I \rho) = 1$ is the probability of $\rho$ falling into the whole space.

\paragraph{Expressiveness of quantum predicates}
$\QPRED$s can be used to describe classical properties (at the propositional level).
For example, we can use a main diagonal matrix -- a kind of $\QPRED$ --
to express (probabilistic) boolean functions,
if classical information is encoded as (a random distribution of) states of a computational basis.
To sum up, the quantum counterpart of a classical predicate
is a main diagonal matrix that encodes its indicator function.
The strength of adopting $\QPRED$s as assertions, among many others,
lies in the fact that various properties of quantum effects can be represented thereof.
For example, the predicate $M =\outprod{+}{+}$, i.e.
\begin{eqnarray*}
  M &=& \frac{\outprod{0}{0} + \outprod{0}{1} + \outprod{1}{0} + \outprod{1}{1}}{2}
\end{eqnarray*}
describes that a state $\rho$ is in the equal superposition $\ket{+}$ with probability $\tr(M \rho)$;
the predicate $N = \frac{\ket{00} + \ket{11}}{\sqrt{2}} \frac{\bra{00} + \bra{11}}{\sqrt{2}}$, i.e.
\begin{eqnarray*}
  N &=& \frac{\outprod{00}{00} + \outprod{00}{11} + \outprod{11}{00} + \outprod{11}{11}}{2}
\end{eqnarray*}
describes that a state $\rho$ is in the maximally entangled state $\frac{\ket{00} + \ket{11}}{\sqrt{2}}$ with probability $\tr(N \rho)$.

\paragraph{Definition of quantum implication}
The L\"{o}wner comparison $M \sqsubseteq N$ between $\QPRED$s $M$ and $N$
is a quantum simulation of the classical implication ``$F \rightarrow G$''.
Recall that the validity of $F \rightarrow G$, denoted $\models F \rightarrow G$,
is defined as: for any assignment $v$, $v\models F \implies v\models G$,
where $v\models F$ denotes the satisfiability (truth value) of $F$ under $v$.
The following lemma extends the semantics of classical implication into the quantum case
\big(To better see this, we remark that
the less than or equal to $\leq$ defined on the closed real interval $[0,1]$
can be seen as a probabilistic extension of the classical implication $\implies$ on the set of truth values $\{0,1\}$\big).
\begin{lemma}[A semantic viewpoint of L\"{o}wner order {\cite[Lem. 2.1]{Ying11}}]\label{lem_lowner_compr}
  Let $M, N \in \Pp(\Hh)$. Then $M \sqsubseteq N$ if,
  and only if, $\tr(M \rho) \leq \tr(N \rho)$ for all $\rho\in \Dd(\Hh)$.
\end{lemma}

\section{Recursive quantum programs}\label{sec_rqp}

Recursive quantum programs can be viewed
as a recursive procedural extension of quantum base language $\qPL$
--- the non-while-loop part of quantum $\WHILE$-language introduced in \cite{Ying11,ying2016foundations}.
In this section, we define its syntax and formal semantics,
and introduce the example of recursive quantum Markov chain as the running example of the paper.

\subsection{Definition of the syntax}\label{subsec_syntax}

We assume a set $\Var$ of quantum variables annotated with a type $\Boolean$ or $\integer$,
and $\bar{q} \triangleq (q_1,\ldots,q_n) \subseteq \Var$.
For each $q\in \Var$, its state Hilbert space is denoted by $\Hh_q$.
Then $\bar{q}$ is associated with the Hilbert space $\Hh_{\bar{q}} \triangleq \bigotimes_{i=1}^{n} \Hh_{q_i}$.
If $\type(q) = \Boolean$,
then $\Hh_q$ is a two-dimensional Hilbert space with computational basis $\{\ket{0}, \ket{1}\}$
(i.e. the space of a qubit).
If $\type(q) = \integer$,
then $\Hh_q$ is an infinite-dimensional Hilbert space with computational basis $\{\ket{n} \colon n\in \Zz\}$
(i.e. the space of a quint, e.g. the space of a photon).
Note that we are able to use multiple (e.g. $n$) qubits to form any finite (e.g. $2^n$) dimensional Hilbert space,
and usually use the computational basis of an infinite-dimensional Hilbert space to simulate integers $\Zz$,
which can be used to implement the classical (deterministic) control of a quantum program (cf. Exm. \ref{exam_counter_par}).
(Working with infinite Hilbert space is as with finite space, see, e.g., \cite{prugovecki1982quantum}.)

Now we are able to define a (possibly recursive) procedural extension
of quantum base language $\qPL$
(the non-while-loop part of quantum $\WHILE$-language \cite{Ying11,ying2016foundations}),
denoted by $\qRP$ (We can implement a while loop as a tail recursion, cf. Subsec. \ref{subsec_loop}).
A recursive quantum program $P\in \qRP$ usually consists of a procedural declaration $D$,
associating some body with a procedure name,
followed by some statement $S$ possibly containing activation statements to declared procedures.
Formally, $\qRP$ is generated by the following grammar:
\begin{displaymath}
\begin{array}{rcll}
  P & \triangleq & D :: S & \mbox{Quantum program} \\
  D & \triangleq & \recDec{\langle \proc_1 \rangle}{S_1}, \ldots, \recDec{\langle \proc_n \rangle}{S_n} & \mbox{Procedure declaration} \\
  S & \triangleq & \bottom & \mbox{Bottom} \\
    & \mid & \SKIP & \mbox{No operation} \\
    & \mid & q \assnequal \ket{0} & \mbox{Initialization} \\
    & \mid & \lst{q} \starequal U & \mbox{Unitary operation} \\
    & \mid & S_1;S_2 & \mbox{Sequential composition} \\
    & \mid & \ifStat{\Box m\cdot M[\lst{q}] = m \rightarrow S_m} & \mbox{Probabilistic branching} \\
    & \mid & \CALL\ \langle \proc_i \rangle,\quad 1\leq i \leq n & \mbox{Procedure call}
\end{array}
\end{displaymath}
where for each declared (possibly recursive) procedure
$\recDec{\langle \proc_i \rangle}{S_i}$, $1 \leq i \leq n$,
$\proc_i$ and $S_i$ are the name and body of the procedure.
For the follow-up development,
we first assume that quantum programs have no local variables,
and that procedures $\proc_i$ dispense with parameter passing.
The treatment of these auxiliary facilities is deferred to Sec. \ref{sec_loc}.

The intended semantics of language constructs above is similar to that of their classical counterparts.
To see the quantum features of those constructs,
we remark that:
\begin{description}
  \item[(i)] for the initialization, the choice of the state $\ket{0}$ as the initial value is due to the fact that any known quantum state can be prepared by applying a unitary operator to $\ket{0}$;
  \item[(ii)] for the probabilistic branching, different branches $\{S_m\}_m$ are chosen
      according to (randomly distributed) outcomes of the measurement $M \triangleq \{M_m\}_m$ on the qubits $\bar{q}$,
      and the measurement process could possibly destroy the current state.
\end{description}

\begin{example}[Alternative definition of $\bottom$]\label{exam_bottom}
Quantum program $\bottom$, as a basic program construct, can also be implemented
as the call statement $\CALL\ P_{\bottom}$ with $P_{\bottom}$ declared by
$$
\recDec{P_{\bottom}}{\CALL\ P_{\bottom}}
$$
\end{example}

\begin{remark}
For the syntax of $\qRP$,
we choose $\bottom$ as a basic program construct, rather than as in Exm. \ref{exam_bottom},
because the denotational semantics of recursive procedures
and its derivatives (e.g. formal weakest preconditions)
will need to be defined on the ``absolute''  semantics of $\bottom$ as a meta-symbol,
otherwise a circular reasoning will occur.
\end{remark}

\paragraph{Running example of the paper.}
To illustrate that our method has the capacity for dealing with probabilistic control,
we adopt the example of recursive quantum Markov chains (abbr. $\RQMC$) \cite{feng2013reachability}
as the running example of the paper.

\begin{example}[Syntax of $\RQMC$]\label{exam_syn_RQMC}
Let us introduce a modified version of Exm. 1 in the literature \cite{feng2013reachability}.
This is a two-player (Alice and Bob) game of first tossing a dice,
simulated by a qubit system $q$,
and then making a decision for either being the final winner,
flagged as $\ket{\pm}$, or transferring $q$ to the other.
The protocol of Alice goes as follows.
She first measures $q$ immutably by the observable
\begin{eqnarray*}
  M &\triangleq& \bigg\{ M_0 \triangleq \sqrt{\frac{1}{4}} I,\ M_1 \triangleq \sqrt{\frac{1}{2}} I,\ M_2 \triangleq \sqrt{\frac{1}{4}} I \bigg\}
\end{eqnarray*}
If the outcome $0$ is observed, then she sets $q$ to be $\ket{+}$ and terminates;
if $1$ is observed, then she transfers $q$ to the Bob and lets him play;
if $2$ is observed, the game will get stuck.
The protocol of Bob goes similarly except that
he will use the following observable
\begin{eqnarray*}
  M' &\triangleq& \bigg\{ M_0' \triangleq \sqrt{\frac{1}{2}} I,\ M_1' \triangleq \sqrt{\frac{1}{2}} I \bigg\}
\end{eqnarray*}
instead, and if the measurement outcome $0$ is observed, he sets $q$ to be $\ket{-}$ and terminates.
After Bob wins, Alice will do nothing and terminate immediately.
The game starts with Alice.
The core of the game is mutually recursive procedures $\Alice$ and $\Bob$ programmed as
$$
\begin{array}{c}
  \recDec{\Alice}{\ifStat{\Box m\cdot M[q] = m \rightarrow S_m}} \\
  \recDec{\Bob}{\ifStat{\Box m\cdot M'[q] = m \rightarrow S_m'}}
\end{array}
$$
where $\{S_m\}_{m = 0,1,2}$ and $\{S_m'\}_{m = 0,1}$ are defined as
$$
\begin{array}{lll}
  S_0 \triangleq q \starequal H & S_1 \triangleq \CALL\ \Bob & S_2 \triangleq \bottom \\
  S_0' \triangleq \CALL\ \Alice & S_1' \triangleq q \starequal H X &
\end{array}
$$
The main program $\RQMC$ of the game is as follows.
\begin{eqnarray*}
  \RQMC &\triangleq& q \assnequal 0;\ \CALL\ \Alice
\end{eqnarray*}
\end{example}

\subsection{Nondeterministic operational semantics}\label{subsec_op_sem}

\begin{table}[!htbp]

  \centering

  \begin{tabular}{rlcrl}
  (Bot) & $\langle \bottom,\rho\rangle \xrightarrow{\epsilon} \langle E, 0\rangle$  &  &
  (Skip) & $\langle \SKIP,\rho \rangle \xrightarrow{\epsilon} \langle E, \rho \rangle$ \\
  \specialrule{0em}{2pt}{2pt}
  (Init) &
  $\dfrac{\sum_{i} \voutprod{i}{q}{i} = I_{q}}{\langle q \assnequal \ket{0}, \rho \rangle \xrightarrow{\epsilon} \langle E, \sum_i \ket{0}_q \bra{i}\rho \ket{i}_q \bra{0} \rangle}$
  &  & (Unit) &
  $\dfrac{U U^\dag = U^\dag U = I_{\lst{q}}}{\langle \lst{q} \starequal U, \rho \rangle \xrightarrow{\epsilon} \langle E, U\rho U^\dag \rangle}$
  \\
  \specialrule{0em}{2pt}{2pt}
  (Comp${}_1$) & $\dfrac{\langle S_1,\rho\rangle \xrightarrow{l} \langle S_1',\rho'\rangle \mbox{ and } S_1' \neq E}{\langle S_1;S_2,\rho\rangle \xrightarrow{l} \langle S_1';S_2,\rho'\rangle}$ & &
  (Comp${}_2$) & $\dfrac{\langle S_1,\rho\rangle \xrightarrow{l} \langle E,\rho'\rangle}{\langle S_1;S_2,\rho\rangle \xrightarrow{l} \langle S_2,\rho'\rangle}$ \\
  \specialrule{0em}{2pt}{2pt}
  (Case) & \multicolumn{3}{l}{$\dfrac{M = \{M_m\}_m \mbox{ and } \sum_{m} M_m^\dag M_m = I_{\lst{q}}}{\langle \IF, \rho\rangle \xrightarrow{m} \langle S_m,M_m\rho M_m^\dag \rangle}$} \\
  \specialrule{0em}{2pt}{2pt}
  (Proc) & $\dfrac{\recDec{\proc}{S} \in D}{\langle\CALL\ \proc,\rho\rangle \xrightarrow{\epsilon} \langle S, \rho\rangle}$ & &
  (Except) &$\dfrac{\mbox{Other cases of $S$ and $l$}}{\langle S,\rho\rangle \xrightarrow{l} \bot}$
  \end{tabular}
  \caption{Labeled transition relation $\xrightarrow{l}$ with $l \triangleq \epsilon \mid m$.}
  \label{tab_opSem}
\end{table}

The operational semantics of quantum programs can be defined as
a nondeterministic transition relation $\rightarrow$
between quantum configurations $\langle S,\rho \rangle$
--- a global description for quantum program $S\in \qRP$
on the current state represented as a $\PDOP$ $\rho$.
For the well-definedness of $\langle S,\rho \rangle$,
it is required that $\Var(S) \subseteq \Var(\rho)$.
Note that $S$ could be the empty statement $E$ indicating that $\rho$ is the final output.
By a labeled transition
$$
\langle S,\rho \rangle \xrightarrow{l} \langle S',\rho'\rangle
$$
we mean that program $S$ on input state $\rho$ is evaluated
in one step with label $l$ to program $S'$ with output state $\rho'$.
The transition relation $\xrightarrow{l}$ for $\qRP$ is defined in Tab. \ref{tab_opSem}.

To better understand the relation $\xrightarrow{l}$, we remark that
\begin{itemize}
  \item The transition relation $\xrightarrow{l}$ is defined in a nondeterministic manner.
The unique source of nondeterminism is execution of a case statement,
and each measurement outcome $m$ corresponds to a different path of execution.
For those deterministic one-step executions, we use $\xrightarrow{\epsilon}$ instead
($\epsilon$ means there are no other choices).
Otherwise (when the cases of $S$ and $l$ mismatch),
the execution will fail and fall into a bottom state $\bot$.
  \item The annotated outer product $\voutprod{\psi}{\lst{q}}{\phi}$ and identity operator $I_{\lst{q}}$
mean that $\outprod{\psi}{\phi}, I \in \Hh_{\lst{q}}$.
  \item Here, and in the sequel,
the statement $\ifStat{\Box m\cdot M[\lst{q}] = m \rightarrow S_m}$ is shortened with $\IF$.
\end{itemize}

\begin{example}[Operational semantics of $\RQMC$]\label{exam_opSem_RQMC}
  Let $\RQMC$ be the quantum program defined in Exm. \ref{exam_syn_RQMC}.
  Let annotated $\PDOP$ $\rho_{q} \in \Dd(\Hh_{q})$ with $\tr(\rho_{q}) = 1$.
  Then part of the operational semantics of $\RQMC$ (Alice wins after two rounds)
  can be developed step by step as follows.
  $$
  \begin{array}{rl}
     & \langle q \assnequal 0; \CALL\ \Alice, \rho_{q}\rangle \\
    \xrightarrow{\epsilon} & \langle \CALL\ \Alice, \voutprod{0}{q}{0}\rangle \\
    \xrightarrow{\epsilon} & \langle \ifStat{\Box m\cdot M[q] = m \rightarrow S_m}, \voutprod{0}{q}{0}\rangle \\
    \xrightarrow{1} & \langle \CALL\ \Bob, \frac{1}{2}\voutprod{0}{q}{0}\rangle \\
    \xrightarrow{\epsilon} & \langle \ifStat{\Box m\cdot M'[q] = m \rightarrow S_m'}, \frac{1}{2}\voutprod{0}{q}{0}\rangle \\
    \xrightarrow{0} & \langle \CALL\ \Alice, \frac{1}{4}\voutprod{0}{q}{0}\rangle \\
    \xrightarrow{\epsilon} & \langle \ifStat{\Box m\cdot M[q] = m \rightarrow S_m},
    \frac{1}{4}\voutprod{0}{q}{0}\rangle \\
    \xrightarrow{0} & \langle q \starequal H, \frac{1}{16}\voutprod{0}{q}{0}\rangle \\
    \xrightarrow{\epsilon} & \langle E, \frac{1}{16}\voutprod{+}{q}{+}\rangle
  \end{array}
  $$
  The last configuration shows that Alice wins after two rounds with probability $\frac{1}{16}$.
\end{example}

To extend the one-step labeled transition relation $\xrightarrow{l}$
to a multi-step labeled transition relation $\xrightarrow{\alpha}$
with $\alpha \triangleq l \mid \alpha_1 \alpha_2$,
we define $\langle S,\rho \rangle \xrightarrow{\alpha} C$
($C$ is either $\langle S', \rho' \rangle$ or $\bot$) as
$\xrightarrow{l}$, defined in Tab. \ref{tab_opSem}, if $\alpha = l$;
or as $\langle S,\rho \rangle \xrightarrow{\alpha_1} \langle S'',\rho'' \rangle$
and $\langle S'', \rho'' \rangle \xrightarrow{\alpha_2} C$
for some $S''$ and $\rho''$ otherwise.

\begin{remark}[Comparison with Ying's operational semantics]\label{rem_op_sem}
Ying's original operational semantics \cite{Ying11,ying2016foundations} is
defined in a nondeterministic way without resorting to the concept of labels,
and there is only one rule (followed by the side condition $E;S_2 \triangleq S_2$) for the case of composition,
which is able to combine together the (Comp${}_1$, Comp${}_2$) rules of Tab. \ref{tab_opSem}.
In other words, in order to define a more fine-grained operational semantics,
we have to introduce the concept of labeled transition relation,
at the price of introducing an extra rule (Except) dealing with the case of exception.
\end{remark}

\subsection{$\QOP$-directed denotational semantics}\label{subsec_de_sem}

\begin{table}[!htbp]

  \centering

  \begin{tabular}{rlrl}
    (Param) & $\sem{\Omega_{\lst{q}}} = \omega_{\lst{q}}$ & (Bot) & $\sem{\bottom} = 0 \diamond 0$ \\
    \specialrule{0em}{2pt}{2pt}
    (Skip) & $\sem{\SKIP} = I \diamond I$ & (Init) & $\sem{q \assnequal \ket{0}} = \sum_i \ket{0}_q \bra{i} \diamond \ket{i}_q \bra{0}$ \\
    \specialrule{0em}{2pt}{2pt}
    (Unit) & $\sem{\lst{q} \starequal U} = U \diamond U^\dag$ & (Comp) & $\sem{S_1;S_2} = \sem{S_2} \circ \sem{S_1}$
    \\
    \specialrule{0em}{2pt}{2pt}
    (Case) & $\sem{\IF} = \sum_m \sem{S_m} \circ (M_m \diamond M_m^\dag)$ & (Proc) & $\sem{\CALL\ \proc_i} = \bigsqcup_{n = 0}^{\infty} \sem{ S_i^{(n)} }$
  \end{tabular}
  \caption{Denotational semantics of $\qRP[\Omega]$.}
  \label{tab_deSem}
\end{table}

The denotational semantics of a quantum program, denoted $\sem{\cdot}$, is defined as a super operator (i.e. $\QOP$).
The semantics of each term is given in a compositional way, except for the case of call statements.
To handle this case,
we need to define the syntactic approximation (i.e., unrolling) of the bodies of mutually recursive procedures.
\begin{definition}[Syntactic approximation]\label{def_syn_approx}
For declared recursive procedures $\proc_i$ with body $S_i$, $1 \leq i \leq n$,
the $k$th syntactic approximation $S_i^{(k)}$ is defined as:
\begin{eqnarray*}
  S_i^{(0)} &\triangleq& \bottom \\
  S_i^{(k+1)} &\triangleq& S_i\big[ \ldots, \big(\SKIP; S_j^{(k)}\big) \big/ \CALL\ \proc_j, \ldots \big]
\end{eqnarray*}
where $S_i\big[\ldots, \big( \SKIP; S_j^{(k)}\big) \big/ \CALL\ \proc_j, \ldots\big]$
stands for simultaneous substitution of $\SKIP; S_j^{(k)}$ for $\CALL\ \proc_j$,
for all $1 \leq j \leq n$, occurring inside $S_i$
\big(Here $\SKIP$ is used to simulate the first-step transition $\xrightarrow{\epsilon}$ for the statement $\CALL\ \proc_j$,
cf. the (Skip, Proc) rules of Tab. \ref{tab_opSem}\big).
\end{definition}

The denotational semantics $\sem{\cdot}$ for $\qRP$ with parameters $\Omega$
(denoted $\qRP[\Omega]$) is defined in Tab. \ref{tab_deSem},
where the quantum program parameter $\Omega_{\lst{q}}$,
ranging over the set of all quantum programs for quantum variables $\lst{q}$,
is interpreted as the corresponding quantum operation parameter $\omega_{\lst{q}}$,
ranging over the set of all $\QOP$s on $\Hh_{\lst{q}}$.
This is justified by the fact that any $\QOP$ can be simulated by a non-parameterized quantum program.

\begin{lemma}[Well-definedness of $\sem{\cdot}$]\label{lem_well_def_den}
Let $S_i'$ be a parameterized adaptation of $S_i$, the body of $\CALL\ \proc_i$, with $1 \leq i \leq n$,
by substituting parameterized quantum program $\SKIP; \Omega_j$ for each call-statement $\CALL\ \proc_j$
with $1 \leq j \leq n$ occurring inside $S_i$:
\begin{eqnarray*}
  S_i' &\triangleq& S_i\big[\big( \SKIP; \Omega_1 \big) \big/ \CALL\ \proc_1, \ldots, \big( \SKIP; \Omega_j \big) \big/ \CALL\ \proc_j, \ldots, \big( \SKIP; \Omega_n \big) \big/ \CALL\ \proc_n]
\end{eqnarray*}
Let $\Ff(\omega_1, \ldots, \omega_n)$ be a vectorial function on $\QOP$s defined by
\begin{eqnarray*}
  \Ff &\triangleq& \big( \sem{S_1'}, \ldots, \sem{S_j'}, \ldots, \sem{S_n'} \big)
\end{eqnarray*}
Define the least sequence of $\QOP$s $\big\{ \Ff_i^{k}(\lst{0}) \triangleq \calE_i^{(k)} \big\}_{k \geq 0}$, $0 \leq i \leq n$, generated by $\calF$ as follows.
\begin{eqnarray*}
  \big( \calE_1^{(0)}, \ldots, \calE_j^{(0)}, \ldots, \calE_n^{(0)} \big) &\triangleq& \big( 0, \ldots, 0, \ldots, 0 \big) \\
  \big( \calE_1^{(k+1)}, \ldots, \calE_j^{(k+1)}, \ldots, \calE_n^{(k+1)} \big) &\triangleq& \calF \big( \calE_1^{(k)}, \ldots, \calE_j^{(k)}, \ldots, \calE_n^{(k)} \big)
\end{eqnarray*}
It is the case that
\begin{description}
  \item[(i)] $\sem{S_i^{(k)}} = \Ff_i^{k}(\lst{0})$ for all $k \geq 0$;
  \item[(ii)] $\sem{S_i^{(k)}} \sqsubseteq \sem{S_i^{(k+1)}}$ for all $k \geq 0$.
\end{description}
\end{lemma}
\begin{proof}
  By definition of $\sem{\cdot}$ and $\big\{ \Ff_i^{k}(\lst{0}) \big\}_{k \geq 0}$, Stat. (i) follows.
  To show Stat. (ii), we remark that $\Ff$ is monotone,
  i.e. for any $\QOP$s $\calE_j$ and $\calF_j$ with $\calE_j \sqsubseteq \calF_j$,
\begin{eqnarray*}
  \Ff(\ldots, \calE_j, \ldots) &\sqsubseteq& \Ff(\ldots, \calF_j, \ldots)
\end{eqnarray*}
(Note that $\sqsubseteq$ distributes over components of the vector.)
This is the case due to linearity of (super operators) $\Ff$,
together with the fact that $\calE_j + \calG_j = \calF_j$ for some $\QOP$ $\calG_j$.
Then Stat. (ii) follows by induction on $k$, together with monotonicity of $\calF$.
\end{proof}

\begin{remark}[Comparison with Ying's denotational semantics]\label{rem_closure_den}
The denotational semantics $\sem{\cdot}$ of Table \ref{tab_deSem} is
defined as a composition of $\QOP$s independent of the input $\rho$,
which can be seen as a parameterized extension of
Ying's original denotational semantics \cite{Ying11,ying2016foundations}.
The parameterized version of denotational semantics helps to explain that
\begin{itemize}
  \item $\big\{ \sem{S_i^{(k)}} \big\}_{k \geq 0}$, $0 \leq i \leq n$, is the least sequence of $\QOP$s generated by $\calF$;
  \item $\big\{ \sem{\CALL\ \proc_i} \big\}_{0 \leq i \leq n}$ is the least fixed point of the vectorial function $\calF$, i.e.
      \begin{itemize}
        \item $( \ldots, \sem{\CALL\ \proc_j}, \ldots) = \calF( \ldots, \sem{\CALL\ \proc_j}, \ldots)$, and
        \item for any $\QOP$s $\{ \calE_i \}_{0 \leq i \leq n}$ with $( \calE_1, \ldots, \calE_j, \ldots, \calE_n ) = \calF ( \calE_1, \ldots, \calE_j, \ldots, \calE_n )$, we have that $\sem{\CALL\ \proc_j} \sqsubseteq \calE_j$ for all $0 \leq j \leq n$.
      \end{itemize}
\end{itemize}
as implied by Lem. \ref{lem_well_def_den}.
\end{remark}

\begin{example}[Denotational semantics of $\RQMC$]\label{exam_deSem_RQMC}
  Let $\RQMC$ be the quantum program defined in Exm. \ref{exam_syn_RQMC}.
  Let $S_a$ and $S_b$ denote the bodies of procedures $\Alice$ and $\Bob$.
  Then $S_a' \triangleq S_a[\Omega_b/ \CALL\ \Bob]$ and $S_b' \triangleq S_b[\Omega_a/ \CALL\ \Alice]$.
  Thus we have that
  \begin{eqnarray*}
    \Ff &\triangleq& (\sem{S_a'}, \sem{S_b'}) = (\frac{1}{4} H \diamond H + \frac{1}{2} \omega_b,
  \frac{1}{2} \omega_a + \frac{1}{2} HX \diamond XH)
  \end{eqnarray*}
  This implies that
  \begin{eqnarray*}
    \Ff_1^{n}(\lst{0}) &=& \Big( \sum_{k \geq 1}^{2k-1 \leq n}\frac{1}{4^k} \Big) H \diamond H + \Big( \sum_{k \geq 1}^{2k \leq n} \frac{1}{4^k} \Big) HX \diamond XH \\
    \Ff_2^{n}(\lst{0}) &=& \Big( \sum_{k \geq 1}^{2k \leq n}\frac{1}{2\cdot 4^k} \Big) H \diamond H + \Big( \sum_{k \geq 1}^{2k-1 \leq n} \frac{1}{2\cdot 4^{k-1}} \Big) HX \diamond XH
  \end{eqnarray*}
  By the (Proc) rule of Tab. \ref{tab_deSem}, it follows that
  \begin{eqnarray*}
    \sem{\CALL\ \Alice} &=& \bigsqcup_{n = 0}^{\infty} \Ff_1^{n}(\lst{0}) = \frac{1}{3} H \diamond H + \frac{1}{3} HX \diamond XH
  \end{eqnarray*}
  Finally, for any $\PDOP$ $\rho$ with $\tr(\rho) = 1$, we have that
  \begin{eqnarray*}
    \sem{\RQMC}(\rho) &=& \sem{q \assnequal 0; \CALL\ \Alice}(\rho) \\
     &=& (\sem{\CALL\ \Alice} \circ \sem{q \assnequal 0}) (\rho) \\
     &=& \sem{\CALL\ \Alice}\big( \sem{q \assnequal 0}(\rho) \big) \\
     &=& \sem{\CALL\ \Alice}(\outprod{0}{0}) \\
     &=& \frac{1}{3} H \outprod{0}{0} H + \frac{1}{3} HX \outprod{0}{0} XH \\
     &=& \frac{1}{3} \outprod{+}{+} + \frac{1}{3} \outprod{-}{-} = \frac{1}{3} I
  \end{eqnarray*}
\end{example}
\begin{remark}
After round $n = 2k - 1$ ($k \geq 1$) of the two-player game (cf. Exm. \ref{exam_opSem_RQMC}),
there is a computed result $\frac{1}{4^k}\outprod{+}{+}$,
but the resulting state as a whole (for all $k \geq 0$) should include the sum of all, i.e.,
$\sum_{k = 1}^{\infty} \frac{1}{4^k}\outprod{+}{+} = \frac{1}{3}\outprod{+}{+}$ (cf. Exm. \ref{exam_deSem_RQMC}).
\end{remark}

The following Theorem reveals the connection between operational and denotational semantics.
Namely, the meaning of running program $S$ on input state $\rho$
is the sum of all possible output states $\rho'$
\big(Note that $\tr(\rho')$ denotes the probability of reaching $\rho'/\tr(\rho')$\big).

\begin{theorem}\label{thm_sem}
  For any quantum program $S\in \qRP$, we have that
  \begin{eqnarray}\label{eq_sem}
    \sem{S}(\rho) &=& \sum_{\langle S,\rho\rangle\xrightarrow{\alpha}\langle E,\rho'\rangle} \rho'
  \end{eqnarray}
  where the summation of $\rho'$ is taken for every possible $\alpha$ s.t. $\langle S,\rho\rangle\xrightarrow{\alpha}\langle E,\rho'\rangle$.
\end{theorem}
\begin{proof}
  See the proof of Thm. \ref{thm_sem_EqRP}.
\end{proof}
\begin{remark}
To the right-hand side of Eq. (\ref{eq_sem}) in Thm. \ref{thm_sem},
the summation should act upon any (possibly overlapping) $\rho'$,
as long as $\langle S,\rho\rangle\xrightarrow{\alpha}\langle E,\rho'\rangle$ holds for a different $\alpha$.
For instance, there are two branches of a program each having the computed result $\frac{1}{4}\outprod{0}{0}$,
then the resulting state as a whole should include the sum of both, i.e.,
$\frac{1}{4}\outprod{0}{0} + \frac{1}{4}\outprod{0}{0} = \frac{1}{2}\outprod{0}{0}$.
By comparison, Ying's original treatment of this issue \cite{Ying11,ying2016foundations},
due to lack of labels in transition rules, is to use a multi-set instead to collect up all possible outputs.
\end{remark}

\section{Quantum assertion logic}\label{sec_ass}

$\QPRED$s are simply employed in \cite{panangaden2006, Ying11}
as pre- and post-conditions of quantum Hoare's triples.
However, to achieve an effectively checkable (symbolic) L\"{o}wner comparison between parameterized quantum predicates,
we have to first define a parameterized symbolic abstraction for $\QPRED$s
--- Parameterized Quantum Predicate Terms (abbr. $\PQPT$),
whose definition should be a trade off between simplicity and expressibility
so that we make a minimal use of parameters ranging over a continuous space, and at the same time,
make sure both $\QPRED$s and all (possibly parameterized) intermediate assertions,
e.g. weakest (liberal) preconditions, in quantum program verification can be expressed thereof.

\paragraph{$\QOP$'s dual as predicate transformer}
To see the role of a $\QOP$'s dual in defining its weakest precondition (i.e. $\QPRED$),
note that every $\QOP$ $\calE$ can be seen as a mapping over $\PDOP$s $\rho$
and its dual $\calE^*$ as a mapping over $\QPRED$s $M$, i.e.
\begin{displaymath}
  \begin{array}{rlccc}
   \calE & \colon & \rho & \mapsto & \calE(\rho) \\
   \calE^* & \colon & \calE^*(M) & \mapsfrom & M
  \end{array}
\end{displaymath}
By definition of Schr\"{o}dinger-Heisenberg dual, we have that
\begin{eqnarray*}
  \tr\big(\calE^*(M) \rho\big) &=& \tr\big(M \calE(\rho)\big)
\end{eqnarray*}
Intuitively, the truth value of $\calE^*(M)$ at $\rho$ is equal to the truth value of $M$ at $\calE(\rho)$.
That is that, $\calE^*(M)$ is the weakest precondition of $\calE$ w.r.t. the postcondition $M$ \cite{panangaden2006}.
The beautiful duality between state-transformer (forwards) and predicate-transformer (backwards) semantics
plays a key role in defining quantum assertion logic using the predicate transform $\calE^*$.

\subsection{Parameterized quantum predicate terms}

\begin{definition}[Syntax of $\PQPT$s]\label{def_syn_pqpt}
Let $\lst{q}$, $\lst{r}$ and $\lst{s}$ be lists of pairwise distinct quantum variables
with $\lst{q} = (\lst{r}, \lst{s})$ (if any).
Let $I_{\lst{q}}$ be the constant symbol denoting the identity operator on $\Hh_{\lst{q}}$,
$\calX_{\lst{q}}$ a metavariable for all first-order variables ranging over $\Pp(\Hh_{\lst{q}})$
with $\calX_{\lst{q}}^\dag = \calX_{\lst{q}}$ (i.e. quantum predicate variables),
and $\calE_{\lst{q}}, \calF_{\lst{q}}$ $\QOP$s on $\Hh_{\lst{q}}$
with $\calE_{\lst{q}} + \calF_{\lst{q}} \sqsubseteq I_{\lst{q}}\diamond I_{\lst{q}}$.

A parameterized base $\calB_{\lst{q}}$ of $\PQPT$s on $\lst{q}$
\big(and its set of parameters $\Params{\calB_{\lst{q}}}$\big) is defined as
$$
  \calB_{\lst{q}}\ \big( \Params{\calB_{\lst{q}}} \big) \quad \triangleq \quad
  \left\{
    \begin{array}{ll}
      \calX_{\lst{q}} & \big( \{ \calX_{\lst{q}} \} \big) \\
      I_{\lst{r}}\otimes \calB_{\lst{s}} & \big( \Params{\calB_{\lst{s}}} \big) \\
      \calB_{\lst{r}}\otimes I_{\lst{s}} & \big( \Params{\calB_{\lst{r}}} \big) \\
      \calB_{\lst{r}}\otimes \calB_{\lst{s}} & \big( \Params{\calB_{\lst{r}}} \uplus \Params{\calB_{\lst{r}}} \big)
    \end{array}
  \right.
$$

A $\PQPT$ $P_{\lst{q}}$ on $\lst{q}$
\big(and its set of parameters $\Params{P_{\lst{q}}}$\big) is defined as
$$
  P_{\lst{q}}\ \big( \Params{P_{\lst{q}}} \big) \quad \triangleq \quad \calE_{\lst{q}}^*(\calB_{\lst{q}}) + \calF_{\lst{q}}^*(I_{\lst{q}})\ \big( \Params{\calB_{\lst{q}}} \big)
$$
\end{definition}
Note that a $\PQPT$ $P_{\lst{q}}$ is usually composed of two parts:
the parameterized part $\calE_{\lst{q}}^*(\calB_{\lst{q}})$
and the non-parameterized part $\calF_{\lst{q}}^*(I_{\lst{q}})$,
which can also be written as
\begin{eqnarray*}
  P_{\lst{q}} &\triangleq& \calE_{\lst{q}}^* \bigg( \bigotimes_{\calX\in \Params{\calB_{\lst{q}}}} \calX \bigg) + \calF_{\lst{q}}^*(I_{\lst{q}})
\end{eqnarray*}
due to the convention that the identity operators in $\calB_{\lst{q}}$ can be omitted.

\vspace{2mm}
\noindent {\bf Semantics of $\PQPT$s.}
Let $\mathbb{I}$ be the standard interpretation of nonlogical symbols
in the syntax of $\PQPT$s to the semantic counterparts.
Let $v$ be a mapping (i.e. assignment) from variables $\calX_{\lst{q}}$ to $\QPRED$s $\Pp(\Hh_{\lst{q}})$,
i.e. $v(\calX_{\lst{q}})\in \Pp(\Hh_{\lst{q}})$.
The denotation of a $\PQPT$ $P_{\lst{q}}$ under interpretation $\mathbb{I}$ and assignment $v$, denoted $P_{\lst{q}}^{\mathbb{I},v}$, can be defined as usual (cf., e.g., Def. \ref{def_sem_pqpt}).

To see well-definedness of $\PQPT$s, i.e. $P_{\lst{q}}^{\mathbb{I},v} \in \Pp(\Hh_{\lst{q}})$,
we note that $P_{\lst{q}}^{\mathbb{I},v}$ is Hermitian and
$$
0_{\lst{q}} = 0_{\lst{q}} \calB_{\lst{q}}^{\mathbb{I},v} 0_{\lst{q}} + 0_{\lst{q}} I_{\lst{q}} 0_{\lst{q}} \sqsubseteq P_{\lst{q}}^{\mathbb{I},v} = \calE_{\lst{q}}^*(\calB_{\lst{q}}^{\mathbb{I},v}) + \calF_{\lst{q}}^*(I_{\lst{q}}) \sqsubseteq \big( \calE_{\lst{q}} + \calF_{\lst{q}} \big)^*(I_{\lst{q}}) \sqsubseteq I_{\lst{q}}
$$
The set of $\PQPT$s on $\lst{q}$ is denoted $\Tt(\lst{q})$.
We shall write $P_{\lst{q}}$ as $P$ if $\lst{q}$ is clear from the context.

\begin{lemma}[From $\QPRED$s to $\PQPT$s]\label{lem_qpred_to_pqpt}
  For every $\QPRED$ $M \in \Pp(\Hh_{\lst{q}})$,
  there is a $\PQPT$ $P_{\lst{q}}$ of the form
  $0_{\lst{q}} \calB_{\lst{q}} 0_{\lst{q}} + \calF_{\lst{q}}^*(I_{\lst{q}}) = \calF_{\lst{q}}^*(I_{\lst{q}})$
  such that $M = P_{\lst{q}}$.
\end{lemma}
\begin{proof}
By Lem. \ref{lem_spec_dec},
$\QPRED$ $M$ has the spectral decomposition $M = \sum_k a_k \outprod{\varphi_k}{\varphi_k}$.
Then the lemma follows by defining $\calF_{\lst{q}}$ as
$\calF_{\lst{q}} \triangleq \sum_k \big(\sqrt{a_k} \outprod{\varphi_k}{\varphi_k}\big) \diamond \big(\sqrt{a_k} \outprod{\varphi_k}{\varphi_k}\big)^\dag$.
\end{proof}

\begin{remark}
The quantum tautology $I$, quantum absurdity $0$, quantum predicate variable $\calX$,
and the negation $I - P$ of $\PQPT$ $P$ with $\Params{P} = \emptyset$
are $\PQPT$s of the form $0 (\calB) 0 + I (I) I$, $0 (\calB) 0 + 0 (I) 0$, $I (\calX) I + 0 (I) 0$,
and $\calF_{\lst{q}}^*(I_{\lst{q}})$ (cf. Lem. \ref{lem_qpred_to_pqpt}), respectively.
\end{remark}

\begin{definition}[Operations on $\PQPT$s]\label{def_quant_bool_op}
  Let $P$, $\{ P_m \}_m$ (resp. $Q_i, \calX_i$ with $1 \leq i \leq l$)
  be $\PQPT$s on quantum variables $\lst{q}$ (resp. $\lst{q}_i$ s.t. $\biguplus_{1 \leq i \leq l} \lst{q}_i$ exists),
  $P_{\lst{r}}, P_{\lst{s}}$ $\PQPT$s on $\lst{r}, \lst{s}$,
  and $\{\calE_m\}_m$ $\QOP$s on $\lst{q}$ with $\sum_{m} \calE_m \sqsubseteq I_{\lst{q}} \diamond I_{\lst{q}}$.
  Define operations on $\PQPT$s as follows.
  \begin{description}
    \item[(Substitution).] $P[Q_1/\calX_1, \ldots, Q_i/\calX_i, \ldots, Q_l/\calX_l]$ is the result of (simultaneously) substituting $Q_i$ for the (at most one) occurrence of $\calX_i$ in $P$ for all $1 \leq i \leq l$, if
        \begin{itemize}
          \item $|\Params{P}| \leq 1$ with $\big\{ Q_i \triangleq \calE_{\lst{q}_i}^*(\calX_i) + \calF_{\lst{q}_i}^*(I_{\lst{q}_i}) \big\}_{1 \leq i \leq l}$, or
          \item $|\Params{P}| \geq 2$ with $\big\{ Q_i \triangleq \calE_{\lst{q}_i}^*(\calX_i) \mid \calF_{\lst{q}_i}^*(I_{\lst{q}_i}) \big\}_{1 \leq i \leq l}$.
        \end{itemize}
    \item[(Conjunction).] $P_{\lst{r}} \otimes P_{\lst{s}}$ is the quantum conjunction of $P_{\lst{r}}$ and $P_{\lst{s}}$, if
        \begin{itemize}
          \item $P_{\lst{r}} \triangleq \calE_{\lst{r}}^*(\calB_{\lst{r}})$ and $P_{\lst{s}} \triangleq \calE_{\lst{s}}^*(\calB_{\lst{s}})$, or
          \item $\Params{P_{\lst{r}}} = \emptyset$ or $\Params{P_{\lst{s}}} = \emptyset$.
        \end{itemize}
    \item[(Disjunction).] $\sum_{m} \calE_m^*( P_m )$ is the quantum disjunction of $\{P_m\}_m$ under the exclusive case selection $\{\calE_m\}_m$, if
        \begin{itemize}
          \item $\big\{ \Params{P_m} \big\}_m$ is a singleton.
        \end{itemize}
  \end{description}
\end{definition}

\begin{lemma}\label{lem_closure_of_pqpt}
$\PQPT$s are closed under those operations defined in Def. \ref{def_quant_bool_op}.
\end{lemma}
\begin{proof}
By Def. \ref{def_syn_pqpt}, together with Lem. \ref{lem_dual_qop}.
\end{proof}

\begin{example}[Quantum predicate variables]$\ $
\begin{itemize}
  \item The $\PQPT$ $P_{q} \otimes \calX_{r} \otimes \calX_{s}$ can induce, by substitution for $\calX_{r}$ and $\calX_{s}$, any $\PQPT$ of the form $P_{q} \otimes Q_{r} \otimes R_{s}$ (no entanglement between $r$ and $s$).
  \item $P_{q} \otimes \calX_{r,s}$ can produce any $\PQPT$ of the form $P_{q} \otimes Q_{r,s}$ (where $Q_{r,s}$ possibly expresses an entanglement between $r$ and $s$).
\end{itemize}
\end{example}

\begin{example}[Modeling classical predicates]\label{exam_simulation}
  Let quantum variable $q$ be such that $\Hh_{q}$ has the computational basis $\{\ket{i}\}_i$
  (i.e., $I_q = \sum_{i} \voutprod{i}{q}{i}$).
  Then the $\PQPT$
  \begin{eqnarray*}
    \sum_{i} \voutprod{i}{q}{i} \calX_q \voutprod{i}{q}{i} &=& \sum_{i} \mixprod{i}{\calX_q}{i} \voutprod{i}{q}{i}
  \end{eqnarray*}
  (By comparison, $\calX_q$ has the outer product representation
  $\calX_q = \sum_{i,j} \mixprod{i}{\calX_q}{j} \voutprod{i}{q}{j}$)
  can be used to simulate a classical first-order variable $x$
  over the domain $\{ (\mixprod{i}{\calX_q}{i}, i) \}_i$,
  each element $i$ occurring with probability $\mixprod{i}{\calX_q}{i}$.
  Based on this, the classical parameterized predicate $\varphi(x)$ can be simulated by
  $\sum_{i} \mixprod{i}{\calX_q}{i} \voutprod{i}{q}{i} \otimes \ulcorner \varphi(i) \urcorner$,
  where $\ulcorner \varphi(i) \urcorner$ is the $\PQPT$ for simulating the closed predicate $\varphi(i)$.
  For illustrating examples, see, e.g., case studies.
\end{example}

\begin{remark}[Discretization]
Every $\PQPT$ can be defined in a discrete space up to approximation.
To see this, it suffices to show that $\QOP$s can be discretized in the sense of approximation.
This is the case due to the fact that
any $\QOP$ can be obtained by tracing out the environmental part of a global unitary operation,
and any unitary operation can be approximated to arbitrary accuracy
by a quantum circuit composed of a (fixed) finite set of gates,
e.g. Hadamard, CNOT and $\pi/8$ \cite{nielsen2000quantum}.
\end{remark}

\subsection{Parameterized orders and limits}

\paragraph{Parameterized order.}
In our quantum program logic,
we shall use $\PQPT$s as pre- and post-conditions in place of $\QPRED$s.
In accordance with this change, the L\"{o}wner comparison between $\QPRED$s
will be replaced by a L\"{o}wner ordering formula for $\PQPT$s,
and quantum assertion theories will be redefined so as to provide these L\"{o}wner ordering formulas.

\begin{definition}[L\"{o}wner order between $\PQPT$s]\label{def_order}
  Let $P$ and $Q$ be $\PQPT$s on quantum variables $\lst{q}$.
  A (legitimate) L\"{o}wner ordering formula is of the form $P \sqsubseteq Q$ or $P = Q$ with $\Params{P} = \Params{Q}$,
  and its formal semantics (i.e. truth value) is defined as follows.
  \begin{itemize}
    \item $\models_{\mathbb{I}} P \sqsubseteq Q$, if $P^{\mathbb{I},v}\sqsubseteq Q^{\mathbb{I},v}, \forall v$.
    \item $\models_{\mathbb{I}} P = Q$, if $\models_{\mathbb{I}} P \sqsubseteq Q$ and $\models_{\mathbb{I}} Q \sqsubseteq P$.
  \end{itemize}
\end{definition}

\begin{example}\label{exam_order}
We illustrate valid (parameterized) L\"{o}wner ordering formulas by two items.
\begin{itemize}
  \item $\mixprod{0}{\calX}{0} \outprod{0}{0} \sqsubseteq \mixprod{0}{\calX}{0} \outprod{0}{0} + \mixprod{1}{\calX}{1} \outprod{1}{1}$;
  \item $X \outprod{0}{0} \calX \outprod{0}{0} X = \outprod{+}{0} \calX \outprod{0}{+}$.
\end{itemize}
\end{example}

\begin{definition}[Quantum assertion theories]
$\ $
\begin{itemize}
  \item The set of true $\sqsubseteq$-ordered $\PQPT$s under $\mathbb{I}$, denoted $\trueOrder$, is defined as
  $$
  \trueOrder \quad \triangleq \quad \bigcup_{\bar{q}} \big\{ P \sqsubseteq Q \colon P,Q\in \Tt(\bar{q}), \Params{P} = \Params{Q}, \textup{ and } \models_{\mathbb{I}} P \sqsubseteq Q \big\}
  $$
  \item The set of true $=$-ordered $\PQPT$s under $\mathbb{I}$, denoted $\trueEquality$, is defined as
  $$
  \trueEquality \quad \triangleq \quad \bigcup_{\bar{q}} \big\{ P = Q \colon P,Q\in \Tt(\bar{q}), \Params{P} = \Params{Q}, \textup{ and } \models_{\mathbb{I}} P = Q \big\}
  $$
\end{itemize}
\end{definition}

Note that $\trueOrder$ provides all (true) formal assertions on $\PQPT$s needed in this paper;
in the sequel, $\trueEquality$, as a subset of $\trueOrder$, can be used to reason with exact probabilities.

\begin{lemma}\label{lem_pqpt_order}
  Let $P \triangleq \calC^*(\calB) + \calD^*(I)$ and $Q \triangleq \calE^*(\calB) + \calF^*(I)$
  be $\PQPT$s with the same base $\calB$.
  Suppose that $\calC \sqsubseteq \calE$ or $\calC \sqsupset \calE$ (i.e. $\calC$ and $\calE$ are L\"{o}wner comparable).
  Then we have that
  \begin{description}
    \item[(1)] $\models_{\mathbb{I}} P \sqsubseteq Q$ if, and only if,
    \begin{itemize}
      \item $\calC \sqsubseteq \calE$ and $\calD^*(I) \sqsubseteq \calF^*(I)$, or
      \item $\calC \sqsupset \calE$ and $\calC^*(I) + \calD^*(I) \sqsubseteq \calE^*(I) + \calF^*(I)$.
    \end{itemize}
    \item[(2)] $\models_{\mathbb{I}} P = Q$ if, and only if, $\calC = \calE$ and $\calD^*(I) = \calF^*(I)$.
  \end{description}
\end{lemma}
\begin{proof}
For the proof of Stat. (1), $\calB$ is taken to be $0$ if $\calC \sqsubseteq \calE$, and $I$ otherwise.
For the proof of Stat. (2), first take $\calB$ to be $0$; then use the fact that $\calC$, $\calE$ are L\"{o}wner comparable.
\end{proof}

\begin{remark}\label{rem_lowner_order}
  The decision of $\models_{\mathbb{I}} P \sqsubseteq Q$ (resp. $\models_{\mathbb{I}} P = Q$)
  will have to resort to positivity (resp. equality) of super operators on separable states,
  which is beyond the scope of the current paper.
  However, as Lem. \ref{lem_pqpt_order} entails,
  a restricted semantics of $P \sqsubseteq Q$ (resp. $P = Q$) independent of parameters
  can be defined based on L\"{o}wner order of $\QOP$s and $\QPRED$s,
  in case that $\calC$ and $\calE$ are comparable.
  In what follows,
  we shall adopt the standard semantics of a L\"{o}wner ordering formula,
  but the restricted semantics applies too.
\end{remark}

\paragraph{Parameterized limits.}
For the definition of necessary intermediate assertions in verifying recursive procedures,
we need to introduce the concept of the L. U. B. (i.e. upper limit)
and G. L. B. (i.e. lower limit) of an infinite sequence of $\PQPT$s.

\begin{definition}[The upper and lower limits]\label{def_limits}
  Let $I$ and $P$ be $\PQPT$s, $\calE$ and $\{\calE_n\}_{n \geq 0}$ $\QOP$s
  with $\forall n \geq 0.\ \calE_n \sqsubseteq \calE_{n+1}$ and $\calE = \lim\limits_{n\to\infty} \calE_n$
  (For the existence of $\lim\limits_{n\to\infty} \calE_n$, cf. Sec. \ref{sec_pre}).
  \begin{itemize}
    \item Define the upper limit $\bigsqcup_{n = 0}^{\infty} P_n$ of the sequence of $\PQPT$s $\big\{ P_n \triangleq \calE_n^*(P) \big\}_{n \geq 0}$ by
        \begin{eqnarray*}
          \bigsqcup_{n = 0}^{\infty} P_n &\triangleq& \lim\limits_{n\to\infty} P_n = \calE^*(P)
        \end{eqnarray*}
    \item Define the lower limit $\bigsqcap_{n = 0}^{\infty} Q_n$ of the sequence of $\PQPT$s $\big\{ Q_n \triangleq \calE_n^*(P) + \big( I - \calE_n^*(I) \big) \big\}_{n \geq 0}$ by
        \begin{eqnarray*}
          \bigsqcap_{n = 0}^{\infty} Q_n &\triangleq& \lim\limits_{n\to\infty} Q_n = \calE^*(P) + \big( I - \calE^*(I) \big)
        \end{eqnarray*}
  \end{itemize}
\end{definition}

\begin{lemma}[Well-definedness of the limits]\label{lem_well_definedness_of_limits}
  Let $\{P_n\}_{n \geq 0}$ and $\{Q_n\}_{n \geq 0}$ be as in Def. \ref{def_limits}.
  Then we have that
  \begin{itemize}
    \item $\{P_n\}_{n \geq 0}$ are $\PQPT$s with $\models_{\mathbb{I}} P_n \sqsubseteq P_{n+1}$ for all $n \geq 0$ (thus denoted $\{P_n\}_{n \geq 0}^{\sqsubseteq}$);
    \item $\{Q_n\}_{n \geq 0}$ are $\PQPT$s with $\models_{\mathbb{I}} Q_n \sqsupseteq Q_{n+1}$ for all $n \geq 0$ (thus denoted $\{Q_n\}_{n \geq 0}^{\sqsupseteq}$).
  \end{itemize}
\end{lemma}
\begin{proof}
  By Lem. \ref{lem_closure_of_pqpt} (closure of $\PQPT$s under disjunction), together with Lem. \ref{lem_pqpt_order}.
\end{proof}

\begin{remark}
  We define the L. U. B. and G. L. B. of a sequence of $\PQPT$s
  as a specialized $\PQPT$ (i.e. the limit of a restricted sequence of $\PQPT$s),
  rather than directly introducing their general form into the syntax of a $\PQPT$,
  due to the fact that the current one is enough for the follow-up development
  while keeping a simple form of $\PQPT$s.
\end{remark}

\subsection{Program correctness and expressiveness}\label{subsec_cor_exp}

\begin{table}[!htbp]

  \centering

  \begin{tabular}{lcl}
    $\fwp.(\CALL\ \proc_i).P = \bigsqcup_{n=0}^\infty \fwp.S_i^{(n)}.P$ & & $\fwlp.(\CALL\ \proc_i).P = \bigsqcap_{n=0}^\infty \fwlp.S_i^{(n)}.P$ \\
      \specialrule{0em}{2pt}{2pt}
    $\fwp.\bottom.P = 0$ & & $\fwlp.\bottom.P = I$ \\
      \specialrule{0em}{2pt}{2pt}
    $\fxp.\SKIP.P = P$ & & $\fxp.(q \assnequal \ket{0}).P = \sum_i \ket{i}_q \bra{0} P \ket{0}_q \bra{i}$ \\
      \specialrule{0em}{2pt}{2pt}
    $\fxp.(\lst{q} \starequal U).P = U^\dag P U$ & & $\fxp.(S_1;S_2).P = \fxp.S_1.(\fxp.S_2.P)$ \\
      \specialrule{0em}{2pt}{2pt}
    \multicolumn{3}{l}{$\fxp.\IF.P = \sum_m M_m^\dag(\fxp.S_m.P)M_m$}
  \end{tabular}
  \caption{Definition of formal weakest (liberal) preconditions --- $\fxp \in \{ \fwp, \fwlp\}$.}
  \label{tab_wp_wlp} 
  \vspace{-8pt}
\end{table}

For now, a (legitimate) quantum partial (resp. total) correctness formula
can be defined as a quantum Hoare's triple $\pcor{P}{S}{Q}$ (resp. $\tcor{P}{S}{Q}$),
where $S$ is a quantum program, and $P,Q$ are $\PQPT$s with $\Params{P} = \Params{Q}$.
Since quantum programs can be viewed semantically as a $\QOP$,
to define the semantics of a quantum Hoare's triple,
we first need to define the correctness semantics of a $\QOP$.

\begin{definition}[Correctness of $\QOP$s]\label{def_cor_qop}
Let $M$, $N$ be $\QPRED$s and $\calE$ a $\QOP$.
We say that
\begin{description}
  \item[(Partial correctness).] $\calE$ is partially correct w.r.t. precondition $M$ and postcondition $N$, written $\pcor{M}{\calE}{N}$, if
\begin{eqnarray}\label{eq_par_sem}
  \tr(M\rho) &\leq& \tr\big(N\calE(\rho)\big) + \big[ \tr(\rho) - \tr\big(\calE(\rho)\big) \big],\quad \forall \rho.
\end{eqnarray}
  \item[(Total correctness).] $\calE$ is totally correct w.r.t. precondition $M$ and postcondition $N$, written $\tcor{M}{\calE}{N}$, if
\begin{eqnarray}\label{eq_tot_sem}
  \tr(M\rho) &\leq& \tr\big(N\calE(\rho)\big),\quad \forall \rho.
\end{eqnarray}
\end{description}
\end{definition}

\begin{remark}\label{rem_qop_cor}
Eq. (\ref{eq_par_sem}) can be seen as a probabilistic version of the following statement:
if state $\rho$ satisfies predicate $M$, then, applying operation $\calE$ to $\rho$,
either $\calE$ fails to terminate or the resulting state $\calE(\rho)$ satisfies predicate $N$;
and total correctness is a stronger version of partial correctness
by guaranteeing termination once the precondition is satisfied.
For more information on classical partial and total correctness,
the reader is referred to \cite{Francez}.
\end{remark}

As in classical Hoare logic,
the notion of weakest (liberal) precondition can be a candidate for the definition of
intermediate assertions involved in proving quantum program correctness.
The semantical weakest (liberal) preconditions (for a $\QOP$) can be defined as:

\begin{theorem}[Quantum duality theorem]\label{thm_sem_dual}
Let $N$ be a $\QPRED$ and $\calE$ a $\QOP$.
Define the semantical weakest (resp., liberal) precondition $\WP(\calE, N)$
(resp., $\WLP(\calE, N)$) of $\calE$ w.r.t. $N$ by
\begin{itemize}
  \item $\WP(\calE, N) \triangleq \calE^*(N)$
  \item $\WLP(\calE, N) \triangleq I - \WP(\calE, I - N)$
\end{itemize}
It is the case, for any $\QPRED$ $M$, that
\begin{description}
    \item[(a)] $\tcor{M}{\calE}{N}$ if, and only if, $M \sqsubseteq \WP(\calE, N)$; and
    \item[(b)] $\pcor{M}{\calE}{N}$ if, and only if, $M \sqsubseteq \WLP(\calE, N)$.
  \end{description}
\end{theorem}
\begin{proof}
  By Lem. \ref{lem_lowner_compr}, together with definition of Schr\"{o}dinger-Heisenberg dual.
\end{proof}
Quantum duality theorem implies that:
(a) $\WP$ of a $\QOP$ can be represented by its Schr\"{o}dinger-Heisenberg dual;
(b) $\WP$ and $\WLP$ are logically dual to each other.

\begin{theorem}[Quantum expressiveness theorem]\label{thm_exp}
Let quantum program $S \in \qRP$, and $P$ a $\PQPT$.
Define the formal weakest (resp. liberal) precondition
$\fwp.S.P$ (resp., $\fwlp.S.P$) of $S$ w.r.t. $P$ in Tab. \ref{tab_wp_wlp}.
It is the case that
\begin{description}
  \item[(a)] $\models_\mathbb{I} \fwp.S.P = \sem{S}^*(P)$;
  \item[(b)] $\models_\mathbb{I} \fwlp.S.P = I - \fwp.S.(I - P)$ $\big( \triangleq \fwp.S.P + (I - \fwp.S.I) \big)$.
\end{description}
\end{theorem}
\begin{proof}
  See App. \ref{app_syn_wp}.
\end{proof}

\begin{lemma}[Well-definedness of $\fwp$ and $\fwlp$]\label{lem_wp_wlp}
Let quantum program $S \in \qRP$, and $P$ a $\PQPT$.
It is the case that
\begin{description}
  \item[(a)] $\models_\mathbb{I} \fwp.S_i^{(n)}.P \sqsubseteq \fwp.S_i^{(n+1)}.P$, for all $n \geq 0$;
  \item[(b)] $\models_\mathbb{I} \fwlp.S_i^{(n)}.P \sqsupseteq \fwlp.S_i^{(n+1)}.P$, for all $n \geq 0$.
\end{description}
\end{lemma}
\begin{proof}
  By Lems. \ref{lem_well_def_den}, \ref{lem_well_definedness_of_limits}, and Thm. \ref{thm_exp}.
\end{proof}

\begin{definition}[Correctness of quantum programs]\label{def_cor_qpr}
Let $P$, $Q$ be $\PQPT$s with $\Params{P} = \Params{Q}$, and $S$ a quantum program.
We say that
\begin{description}
  \item[(Partial correctness).] $S$ is partially correct w.r.t. precondition $P$ and postcondition $Q$ under interpretation $\mathbb{I}$, denoted $\models_\mathbb{I} \pcor{P}{S}{Q}$, if $\models_\mathbb{I} P \sqsubseteq \fwlp.S.Q$;
  \item[(Total correctness).] $S$ is totally correct w.r.t. precondition $P$ and postcondition $Q$ under interpretation $\mathbb{I}$, denoted $\models_\mathbb{I} \tcor{P}{S}{Q}$, if $\models_\mathbb{I} P \sqsubseteq \fwp.S.Q$.
\end{description}
\end{definition}

\begin{remark}\label{rem_fwp_and_fwlp}
Quantum duality and expressiveness theorems together entail that
$\fwp.S.Q$ (resp. $\fwlp.S.Q$) is the weakest $\PQPT$ $R$ s.t.
$\models_\mathbb{I} \tcor{R}{S}{Q}$ (resp. $\models_\mathbb{I} \pcor{R}{S}{Q}$).
This justifies the well-definedness of correctness of quantum programs (cf. Def. \ref{def_cor_qpr}),
which can be seen as a parameterized extension of correctness of $\QOP$s (cf. Def. \ref{def_cor_qop}).
\end{remark}

\section{Starting proof systems}\label{sec_prf}

\begin{table}[!htbp]
  \centering
  \begin{tabular}{rlcrl}
    (A Bot)  & $\pcor{I}{\bottom}{P}$ \big(resp. $\tcor{0}{\bottom}{P}$\big) & &
    (A Skip) & $\pcor{P}{\SKIP}{P}$ \\
    \specialrule{0em}{3pt}{3pt}
    (A Init) & $\dfrac{\sum_{i} \voutprod{i}{q}{i} = I_{q}}{\pcor{ \sum_{i} \ket{i}_q \bra{0} P \ket{0}_q \bra{i} }{ q \assnequal \ket{0} }{ P }}$ & &
    (A Unit) & $\dfrac{U U^\dag = U^\dag U = I_{\lst{q}}}{\pcor{U^\dag P U}{\lst{q} \starequal U}{ P }}$
    \\
    \specialrule{0em}{3pt}{3pt}
    (R Comp) & $\dfrac{\pcor{P}{S_1}{Q} \quad \pcor{Q}{S_2}{R}}{\pcor{P}{ S_1;S_2 }{R}}$ & &
    (R Case) & $\dfrac{\pcor{P_m}{S_m}{Q} \mbox{ \small for each } m}{\pcor{\sum_m M_m^\dag P_m M_m}{\IF}{Q}}$ \\
    \specialrule{0em}{3pt}{3pt}
    (R Order) & $\dfrac{P\sqsubseteq P'\quad \pcor{P'}{S}{Q'}\quad Q'\sqsubseteq Q}{\pcor{P}{S}{Q}}$ & &
    (R Subst) & $\dfrac{\pcor{P}{S}{Q}}{\pcor{P[R/\calX]}{S}{Q[R/\calX]}}$
  \end{tabular}
  \caption{Base proof system $\qBase$.}
  \label{tab_qBase}
  \vspace{-8pt}
\end{table}

This section is devoted to presenting different axiom systems
for proving partial, total and even probabilistic correctness of recursive quantum programs $\qRP$.

\paragraph{Base proof system.}
The first step is to present an extension $\qBase$ (quantum Base System) to part of proof system $\qPD$ of \cite{Ying11}
for both partial and total correctness of quantum base language $\qPL$,
so that we can deal with syntactic pre- and post-conditions (i.e. $\PQPT$s).
Every formula of $\qBase$ is either a legitimate quantum Hoare's triple $\pcor{P}{S}{Q}$ or $\tcor{P}{S}{Q}$
\big(where $P,Q$ are $\PQPT$s with $\Params{P} = \Params{Q}$\big),
or a legitimate L\"{o}wner ordering formula $P \sqsubseteq Q$ or $P = Q$
\big(where $P,Q$ are $\PQPT$s with $\Params{P} = \Params{Q}$\big).

Proof system $\qBase$ features the newly added inference rule --- (R Subst) --- handling the substitution in $\PQPT$s,
where $R$ (resp. $\calX$) is an arbitrary $\PQPT$ (resp. quantum predicate variable),
and $P[R/\calX]$ stands for the result of simultaneously substituting $R$ for each occurrence of $\calX$ in $P$.
For the presentation of $\qBase$, the reader is referred to Tab. \ref{tab_qBase}.

\begin{remark}
  Every proof rule of $\qBase$ except for (A Bot) is only presented in the form of partial correctness formulas,
  but nevertheless, applies to proving total correctness.
  To see this, we note that partial and total correctness are distinguished by whether terminating almost surely
  (cf. Rem. \ref{rem_qop_cor}).
  This justifies why $\bottom$ and recursive procedures need a distinguish between partial and total correctness proof rules,
  because they are sources of non-termination.
\end{remark}

\vspace{1mm}
\noindent {\bf Intuition behind $\qBase$.}
To see the intuition of proof rules in $\qBase$, we remark that
\begin{itemize}
  \item (A Bot, A Skip, A Init, A Unit) have the form $\pcor{\fxp.S.P}{S}{P}$;
  \item (R Comp, R Case, R Subst) preserve the form $\pcor{\fxp.S.P}{S}{P}$ (bidirectionally);
  \item (R Order) can be used to relax $\pcor{\fxp.S.P}{S}{P}$ to $\pcor{Q}{S}{P}$ with
  $\models_\mathbb{I} Q \sqsubseteq \fxp.S.P$.
\end{itemize}
For the proof rule (R Case), the annotated $\IF$-statement is illustrated as follows.
$$
\{l_1: P\}\, \ifStat{\Box m\cdot M[\lst{q}] = m \rightarrow \{l_2^m: P_m\}\, S_m} \, \{l_3: Q\}
$$
Fix the input $\rho$ (at program point $l_1$).
By semantics of the $\IF$-statement,
every post-measurement state $M_m \rho M_m^{\dag}$ (containing the probability of observing outcome $m$)
will go to the corresponding branch labeled by $l_2^m$ in which $P_m$ should be satisfied and $S_m$ will be executed.
By the Turing-Floyd-Hoare principle, we have that
\begin{eqnarray*}
  \tr(P \rho) &\leq& \sum_m \tr(P_m M_m \rho M_m^{\dag})
\end{eqnarray*}
Due to the arbitrariness of $\rho$, by properties of $\tr$ and $\sqsubseteq$,
it follows that $P \sqsubseteq \sum_m M_m^{\dag} P_m M_m$.
Note that after the execution of each $S_m$,
the program point $l_3$ is reached and the attached assertion $Q$ is satisfied.
By weakening $P$ to $\sum_m M_m^{\dag} P_m M_m$ and lifting the above reasoning process into a proof rule,
the inference rule (R Case) follows.
Weakening $P$ to $\sum_m M_m^{\dag} P_m M_m$ guarantees that
(R Case) preserves the form $\pcor{\fxp.S.P}{S}{P}$ forward (i.e. compact soundness).
To make (R Case) preserve the form $\pcor{\fxp.S.P}{S}{P}$ backward (i.e. compact completeness),
we have to choose $P_m$ as $\fxp.S_m.Q$ for each $m$.

\paragraph{Soundness and Completeness.}
Of not only theoretical but also practical interest is
the question of soundness and completeness of proof systems presented as before or after.
The question of soundness concerns the correctness of the method,
whereas the question of completeness concerns the scope of its applicability
(under what circumstances it can be successfully applied).
(For a systematic introduction to the soundness and completeness issues of classical Hoare logic,
the reader is referred to the famous survey paper \cite{apt1981ten}.)

For presentational convenience in what follows,
assume that all formulas $F$ are legitimate quantum Hoare's triples or L\"{o}wner ordering formulas.
For sets of formulas $A$ and $B$,
\begin{eqnarray*}
  A &\models_\mathbb{I}& B
\end{eqnarray*}
means that if $\models_\mathbb{I} A$ then $\models_\mathbb{I} B$,
where by $\models_\mathbb{I} A$ is meant that for all formulas $F$ of $A$, $\models_\mathbb{I} F$.
Let $T$ be a quantum assertion theory (e.g. $\trueOrder$).
For a proof system $\bfH$, e.g. $\qPD$, by
\begin{eqnarray*}
  T, A &\vdash_\bfH& \bigwedge_{F\in B} F
\end{eqnarray*}
is meant that {\it every} formula of $B$ can be deduced from $T$, $A$,
or axioms of $\bfH$ by finitely applying inference rules of $\bfH$.
(We can replace $\vdash_\bfH$ by $\vdash$, if $\bfH$ is clear from the context.)
Note that the assertion theory $T$ is used to provide L\"{o}wner ordering formulas
as antecedents of the inference rule (R Order).
Let $\bfH$ be for the programming language $L$. We say that
  \begin{itemize}
    \item $\bfH$ is sound, if for all Hoare's triples $F$ of $L$ with $\trueOrder \vdash_\bfH F$,
    we have $\models_\mathbb{I} F$;
    \item $\bfH$ is (relatively) complete, if for all Hoare's triples $F$ of $L$ with $\models_\mathbb{I} F$,
    we have $\trueOrder \vdash_\bfH F$.
  \end{itemize}
Note that $\trueOrder$ provides all true L\"{o}wner ordering formulas for (R Order),
which, together with the condition of expressiveness (cf. Subsec. \ref{subsec_cor_exp}),
is sufficient to make $\bfH$ complete \cite{Cook78,BergstraT82}
($\QPRED$s are enough for while loops \cite{Ying11};
while general recursive procedures need $\PQPT$s).
As with $\qPD$, the proof system $\qBase$ (for the base language $\qPL$) is sound and complete.

Not to mention it explicitly,
various proof systems presented in the sequel are sound and complete
in the above sense (e.g. Lem. \ref{lem_sound_complete}),
except that a compact version of soundness and completeness
is introduced for exact probabilistic reasoning (cf. Lem. \ref{lem_cmpt_sound_complete}).
For a complete proof of these soundness and completeness results,
the reader is referred to App. \ref{app_sec_sound_complete}.

\begin{remark}
  The above discussion on soundness and completeness issues is purely theoretical,
  because we adopt the assertion theory as an oracle
  (rather than as a recursively axiomatizable theory),
  following the technical line of classical Hoare logic \cite{Cook78}.
  In practice, as discussed in Rem. \ref{rem_lowner_order},
  a restricted L\"{o}wner comparison between $\PQPT$s independent of parameters
  is enough to cover the correctness checking of L\"{o}wner ordering formulas.
\end{remark}

\subsection{Partial correctness}

\begin{table}[!htbp]

\centering

  \begin{minipage}[b]{\textwidth}
    \centering
    $$\dfrac{\pcor{P}{\CALL\ \proc}{Q} \reVdash{\qBase} \pcor{P}{S}{Q}}{\pcor{P}{\CALL\ \proc}{Q}}$$
    \subcaption{($\pRule$ Rec).}
    \label{sim_par_rule}
  \end{minipage}
  \vfill
  \begin{minipage}[b]{\textwidth}
    \centering
    $$\dfrac{\big\{ \pcor{P_i}{\CALL\ \proc_i}{Q_i} \big\}_{1 \leq i \leq n} \reVdash{\qBase} \bigwedge_{1 \leq i \leq n} \pcor{P_i}{S_i}{Q_i} } { \bigwedge_{1 \leq i \leq n} \pcor{P_i}{\CALL\ \proc_i}{Q_i} }$$
    \subcaption{($\pRule$ gRec).}
    \label{com_par_rule}
    \vspace{-8pt}
  \end{minipage}
  \caption{Proof rules for partial correctness.}
  \label{par_rules} 
\end{table}

We are now in a position to present inference rules for proving partial correctness of recursive procedures.
We begin with the case of simple recursion.

\paragraph{Simple recursion.}
Consider first the case of simple recursion, that is that,
the body $S$ of recursive quantum procedure $\proc$ should itself contain the re-invocation statement $\CALL\ \proc$,
but retain the exclusion of invoking other recursive quantum procedures.
The proof rule --- ($\pRule$ Rec) --- for proving partial correctness of $\CALL\ \proc$ is shown in Tab. \ref{sim_par_rule}.

\begin{table}[!htbp]

  \centering

  \begin{minipage}[b]{\textwidth}
  \centering
  \begin{displaymath}
  \begin{array}{c}
    \pcor{l_0: P}{\CALL\ \proc}{l_1: Q}
  \end{array}
  \end{displaymath}
  \subcaption{Annotated program for $\CALL\ \proc$.}
  \label{anot_call}
  \end{minipage}
  \vfill
  \begin{minipage}[b]{\textwidth}
  \centering
  \begin{displaymath}
     \{l_2: P\}
     \cdots
     \pcor{l_3: P'}{\CALL\ \proc}{l_4: Q'}
     \cdots
     \pcor{l_5: P''}{\CALL\ \proc}{l_6: Q''}
     \cdots
     \{l_7: Q\}
  \end{displaymath}
  \subcaption{Annotated program for the body $S$.}
  \label{anot_body}
  \vspace{-8pt}
  \end{minipage}
  \caption{Intermediate assertion method --- labels ($l_0$ - $l_7$) are used to indicate different program points;
  and each program point is annotated with a $\PQPT$.}
  \label{inter_asser_meth}
    \vspace{-6pt}
\end{table}

\vspace{1mm}
\noindent {\bf Intuition of (Rp Rec).}
To derive the correctness formula $\pcor{P}{\CALL\ \proc}{Q}$ about $\CALL\ \proc$,
it suffices to derive $\pcor{P}{S}{Q}$ for its body $S$ (cf. Tab. \ref{inter_asser_meth});
since $S$ itself contains the re-invocation statement (or inner) $\CALL\ \proc$,
it suffices to derive correctness formulas about the inner $\CALL\ \proc$, say $\pcor{P'}{\CALL\ \proc}{Q'}$;
by the Turing-Floyd-Hoare Principle, $\pcor{P'}{\CALL\ \proc}{Q'}$ should be adapted
from the premise $\pcor{P}{\CALL\ \proc}{Q}$ possibly by using (R Subst).
\big(In this case, data flow goes first from $l_3$ to $l_0$; and then from $l_1$ to $l_4$.
The case of $\pcor{P''}{\CALL\ \proc}{Q''}$ can be analyzed similarly.\big)
This reveals the reason for introducing (R Subst):
without it the above derivation might not proceed as desired.

\begin{example}[Counterexample, cf. App. \ref{app_exam_counter_par}]\label{exam_counter_par}
Let $q$ be a quantum variable with $\type(q) = \integer$.
We define the $(+i)$-operator $U_{+i}$ over the computational basis of $\Hh_q$ by
$$
  U_{+i}\colon \ket{x}\rightarrow\ket{x+i}
$$
and similarly for the $(-i)$-operator $U_{-i}$.
Declare the procedure $\toy$ by
  $$
  \recDec{\langle \toy \rangle}{\ifStat{\Box m\cdot M[q] = m \rightarrow S_m}}
  $$
with $M \triangleq \big\{ M_0 = \sum_{i \leq 0} \outprod{i}{i},\ M_1 = \sum_{i \geq 1} \outprod{i}{i} \big\}$
and $\{S_m\}_{m = 0, 1}$ defined by
$$
S_0 \triangleq \SKIP, \quad S_1 \triangleq q \starequal U_{-1};\ \CALL\ \toy;\ q \starequal U_{+1}
$$
Fix $n \geq 0$. We can derive the partial correctness formula
$$
\pcor{\voutprod{n}{q}{n}}{\CALL\ \toy}{\voutprod{n}{q}{n}}
$$
by using ($\pRule$ Rec). However, this is not the case if the use of (R Subst) is disallowed.
\end{example}

\paragraph{General recursion.}
We now extend (Rp Rec) for simple recursion to the general case.
For (mutual) recursive procedures $\proc_i$ with body $S_i$, $1 \leq i \leq n$,
the inference rule (Rp gRec) is introduced to prove their partial correctness in a simultaneous way
(cf. Tab. \ref{com_par_rule}).

\begin{remark}
  Suppose that the procedure $\proc$ with body $S$ has no re-invocation,
  then the inference rule (Rp Rec) will be degenerated to
  $$
  \mbox{(R Proc)} \quad \dfrac{\pcor{P}{S}{Q}}{\pcor{P}{\CALL\ \proc}{Q}}
  $$
  which is precisely the inference rule for non-recursive procedures.
  In other words,
  the proof rule (R Proc) for non-recursive procedures is a special case of (Rp Rec) for recursive procedures.
  Also, the proof rule (Rp Rec) for simple recursion can be seen as a special case of (Rp gRec) for general recursion,
  if the index variable $i$ is required to range over a singleton.
\end{remark}

\paragraph{Synthesis of recursive invariants.}
When proof rule (Rp gRec) is successfully applied to
proving partial correctness formula $\pcor{P_i}{\CALL\ \proc_i}{Q_i}$ for recursive procedures $\proc_i$,
$1 \leq i \leq n$, we call $(P_i, Q_i)$ a recursive invariant of $\proc_i$,
where $\PQPT$ $P_i$ can be replaced by $\fwlp.(\CALL\ \proc_i).Q_i$ (cf. Tab. \ref{tab_wp_wlp}),
which has the following form
\begin{eqnarray*}
  \fwlp.(\CALL\ \proc_i).Q_i &=& \sem{\CALL\ \proc_i}^*(Q_i) + \big( I - \sem{\CALL\ \proc_i}^*(I) \big)
\end{eqnarray*}
where $\big\{ \sem{\CALL\ \proc_i} \big\}_{1 \leq i \leq n}$ is the least fixed point of $\calF$
(cf. Rem. \ref{rem_closure_den} and Thm. \ref{thm_exp}).
However, $\PQPT$ $Q_i$ sometimes should be parameterized,
and the substitution for parameters will highly depend on $\proc_i$ itself (cf. Exm. \ref{exam_counter_par}),
which means that there is no uniform characterization, say fixed-point characterization,
for the recursive invariant $\big(\fwlp.(\CALL\ \proc_i).Q_i, Q_i\big)$.
In other words, the synthesis of recursive invariants is generally not purely automatic,
yet $\fwlp.(\CALL\ \proc_i).Q_i$ can be automatically synthesised provided $Q_i$ is given.

\paragraph{The scope of applicability.}
Recalling the semantical base of a partial-correctness formula,
i.e. Eq. (\ref{eq_par_sem}) in Def. \ref{def_cor_qop},
one can see that it is a straightforward extension of classical partial-correctness semantics for deterministic programs.
Thus, (Rp gRec) is applicable to reasoning about programs with ``deterministic control and quantum data''.
For this purpose, we typically use quantum variables to model classical variables,
i.e. encode classical values as states of a computational basis.
See, e.g., case studies.

On the other hand, our programming language should include
nondeterministic quantum programs (i.e. those branched by non-deterministic quantum observations),
where each nondeterministic branch is associated with an exact probability
(encoded into states).
For these quantum programs (with ``probabilistic control and quantum data''),
we need to do reasoning with exact probability,
say, given a precondition, with what probability a program
will output a particular state (or a particular class of states) or terminate?
E.g., quantum program $\RQMC$ on any input always outputs $\ket{+}$ with probability $\frac{1}{3}$
(cf. Exm. \ref{exam_deSem_RQMC}), i.e.
\begin{eqnarray}\label{ass_prob_cor_RQMC}
  \forall \rho.\ \tr(I\rho)=1 &\implies& \tr\big(\outprod{+}{+}\sem{\RQMC}(\rho)\big)=\frac{1}{3}
\end{eqnarray}
Unfortunately, interfered by probability of nontermination, i.e. $\tr(\rho)-\tr\big(\sem{\RQMC}(\rho)\big)$,
partial-correctness semantics fails to fully express Ass. (\ref{ass_prob_cor_RQMC}).
Therefore, (Rp gRec) is not very suitable for
reasoning about programs with ``probabilistic control and quantum data''.

\subsection{Total correctness}\label{subsec_prf_sys_tot}

\begin{table}[!htbp]

  \centering

  \begin{minipage}[b]{\textwidth}
    \centering
    $$\dfrac{\begin{array}{c}
    \exists\ \{P_n\}_{n \geq 0}^{\sqsubseteq} \mbox{ with } P_0 = 0 \mbox{ s.t.} \\
    \tcor{P_n}{\CALL\ \proc}{Q} \reVdash{\qBase} \tcor{P_{n+1}}{S}{Q} \mbox{ for all } n\geq 0, \\
    P \sqsubseteq \bigsqcup_{n=0}^{\infty} P_n
    \end{array}
 }{\tcor{P}{\CALL\ \proc}{Q}}$$
    \subcaption{($\tRule$ Rec).}
    \label{sim_tot_rule}
  \end{minipage}

  \vfill

  \begin{minipage}[b]{\textwidth}
    \centering
    $$\dfrac{\begin{array}{c}
    \mbox{for } 1 \leq i \leq n,\ \exists\ \{P_i^j\}_{j \geq 0}^{\sqsubseteq}  \mbox{ with } P_i^0 = 0 \mbox{ s.t.} \\
    \big\{ \tcor{P_i^j}{\CALL\ \proc_i}{Q_i} \big\}_{1 \leq i \leq n} \reVdash{\qBase} \bigwedge_{1 \leq i \leq n} \tcor{P_i^{j+1}}{S_i}{Q_i} \mbox{ for all } j \geq 0, \\
    P_i \sqsubseteq\bigsqcup_{j=0}^{\infty} P_i^j
    \end{array}
    }{\bigwedge_{1 \leq i \leq n} \tcor{P_i}{\CALL\ \proc_i}{Q_i} }$$
    \subcaption{($\tRule$ gRec).}
    \label{com_tot_rule}
  \end{minipage}
  \vspace{-20pt}
  \caption{Proof rules for total correctness.}

  \label{tot_rules} 

  \vspace{-6pt}
\end{table}

We are now positioned to present inference rules for proving total correctness of recursive procedures.
To begin with, recall from Lem. \ref{lem_well_definedness_of_limits} that
$\{P_n\}_{n \geq 0}^{\sqsubseteq}$ is an increasing sequence of $\PQPT$s (ordered by $\sqsubseteq$)
defined by $\big\{ P_n \triangleq \calE_n^*(R) \big\}_{n \geq 0}$,
where $R$ is a $\PQPT$ and $\{\calE_n\}_{n \geq 0}$ is an increasing sequence of $\QOP$s (also ordered by $\sqsubseteq$).

\paragraph{Simple recursion.}
To deal with the termination problem of recursive procedure $\proc$ (cf. Tab. \ref{inter_asser_meth}),
introduce a sequence of $\PQPT$s $\{P_n\}_{n\geq 0}^{\sqsubseteq}$ with $P_0 = 0$.
Intuitively, if the entry point of $\proc$ (say $l_0$) is attached currently with assertion $P_{n+1}$,
then, upon re-invocation in the body $S$ of $\proc$,
the data flow at entry points of inner $\CALL\ \proc$ (e.g., $l_3$ or $l_5$)
needs to be constrained by a stronger assertion,
namely $P_n$, or its substitution by using (R Subst).
Finally, to cease re-invocation, i.e. treating inner $\CALL\ \proc$ as $\bottom$,
the attached assertion at entry points should be $0$, namely $P_0$.
Combining this idea with (Rp Rec),
we thus obtain the inference rule (Rt Rec) for proving total correctness of $\CALL\ \proc$ (cf. Tab. \ref{sim_tot_rule}).

\begin{example}[Counterexample, cf. App. \ref{app_exam_counter_tot}]\label{exam_counter_tot}
Let the recursive procedure $\toy$ be as defined in Exm. \ref{exam_counter_par}.
Fix $n \geq 0$. We can derive the total correctness formula
$$
\tcor{\voutprod{n}{q}{n}}{\CALL\ \toy}{\voutprod{n}{q}{n}}
$$
by using (Rt Rec). However, this is not the case if the use of (R Subst) is disallowed.
\end{example}

\paragraph{General recursion.}
We now extend the (Rt Rec) for simple recursion to the general case.
For recursive procedures $\proc_i$ with body $S_i$, $1 \leq i \leq n$,
the inference rule (Rt gRec) is introduced to simultaneously prove their total correctness (cf. Tab. \ref{com_tot_rule}).
Note that (Rt Rec) can be seen as a special case of (Rt gRec), if the index $i$ is required to range over a singleton.

\begin{remark}
In applications, sequences of $\PQPT$s $\{ P_i^j \}^{\sqsubseteq}_{j \geq 0}$ with $P_i^0 = 0$ in (Rt gRec)
\big(or, equivalently, sequences of $\QOP$s $\{ \calE_i^j \}^{\sqsubseteq}_{j \geq 0}$ with $\calE_i^0 = 0\diamond 0$,
cf. Lem. \ref{lem_well_definedness_of_limits}\big)
usually have a closed form with $j$ as an (index) variable or are defined by induction on $j$,
thus the statement
\begin{eqnarray*}
  \big\{ \atcor{\big}{P_i^j}{\CALL\ \proc_i}{Q_i} \big\}_{1 \leq i \leq n} &\reVdash{\qBase}& \bigwedge_{1 \leq i \leq n} \atcor{\big}{P_i^{j+1}}{S_i}{Q_i}, \quad \mbox{ for all } j \geq 0
\end{eqnarray*}
can be proved either for one pass by treating $j$ as an arbitrary (but fixed) variable,
or for two passes by induction on $j$ (one for the basis and the other for the inductive step).
\end{remark}

\paragraph{Synthesis of intermediate assertions.}
To make (Rt gRec) be successfully applied to
proving total correctness of recursive procedures $\proc_i$, $1 \leq i \leq n$,
we need to provide the intermediate assertions $\{P_i^j\}_{j \geq 0}^{\sqsubseteq}$ and $Q_i$ involved.
To this end, $P_i^j$ can be replaced by $\fwp.S_i^{(j)}.Q_i$
\big(in this case $P_i$ can be selected as $\fwp.(\CALL\ \proc_i).Q_i$\big),
which has the following form
\begin{eqnarray*}
  \fwp.S_i^{(j)}.Q_i &=& \sem{S_i^{(j)}}^*(Q_i)
\end{eqnarray*}
where $\big\{ \sem{S_i^{(j)}} \big\}_{j \geq 0}$ with $1 \leq i \leq n$
is the least sequence of $\QOP$s generated by $\calF$ (cf. Rem. \ref{rem_closure_den} and Thm. \ref{thm_exp}).
As in the case of recursive invariants, there is no uniform (fixed-point) characterization for the $\PQPT$s
$\big\{ \fwp.S_i^{(j)}.Q_i \big\}_{j \geq 0}$ and $Q_i$ (entailed by Exm. \ref{exam_counter_tot}).
Therefore, the synthesis of these intermediate assertions is semi-automatic, that is to say that,
the assertions $\big\{ \fwp.S_i^{(j)}.Q_i \big\}_{j \geq 0}$ can be automatically synthesised provided that $Q_i$ is given.

\paragraph{The scope of applicability.}
Recalling Ass. (\ref{eq_tot_sem}) in Def. \ref{def_cor_qop},
one can see that the total-correctness semantics for quantum programs
is a natural extension of classical counterpart for deterministic programs.
Thus, (Rt gRec) is applicable to reasoning about programs with ``deterministic control and quantum data''.
See, for example, case studies.

However, due to inequality in Ass. (\ref{eq_tot_sem}), this (general) version of total-correctness semantics
can merely be used for reasoning with approximate probabilities, and thus fails to support precise probabilistic reasoning,
e.g. precisely describing Ass. (\ref{ass_prob_cor_RQMC}), as in the case of partial correctness.
Fortunately, a restrictive use of (Rt gRec) applies to
reasoning with exact probabilities about programs with ``probabilistic control and quantum data''.
We shall develop an axiomatic basis for (approximate or exact) probabilistic reasoning as follows.

\subsection{Probabilistic correctness}\label{subsec_prob_corr}

\paragraph{Reasoning with approximate probabilities}
As discussed above,
Ass. (\ref{eq_tot_sem}) can be used for the semantical basis of reasoning with approximate probabilities.
Then an axiomatic basis of the (approximate) probabilistic correctness follows from the soundness and completeness lemma.

\begin{lemma}[Soundness and completeness]\label{lem_sound_complete}
  For any quantum program $S\in \qRP$ and any $\PQPT$s $P,Q$, it is the case that
  $$
  \trueOrder \vdash \tcor{P}{S}{Q} \mbox{ if and only if } \models_\mathbb{I} P \sqsubseteq \fwp.S.Q
  $$
\end{lemma}
\begin{proof}
  By Def. \ref{def_cor_qpr} and Thm. \ref{thm_sound_complete_EqRP}.
\end{proof}

\begin{theorem}[Reasoning with approximate probabilities]\label{thm_reas_approx_prob}
  For any quantum program $S\in \qRP$, any $\QPRED$s $P,Q$ and any $\delta\in [0, 1]$, it is the case that
  $$
  \trueOrder \vdash \tcor{\delta P}{S}{Q} \mbox{ if and only if } \forall \rho.\ \tr(P\rho) = 1 \implies \tr\big(Q\sem{S}(\rho)\big) \geq \delta
  $$
\end{theorem}
\begin{proof}
  Contained in the proof of Thm. \ref{thm_reas_EqPR_with_prob}.
\end{proof}

\paragraph{Reasoning with exact probabilities}
A semantical basis of (exact) probabilistic reasoning
can be adapted from Eq. (\ref{eq_tot_sem}) with $=$ in place of $\leq$ (for approximate reasoning).
Based on this, the semantics of a total-correctness formula
$\tcor{P}{S}{Q}$ has the following property
$$
\models_\mathbb{I} \tcor{P}{S}{Q} \mbox{ if and only if } \models_\mathbb{I} P = \fwp.S.Q
$$
To build an axiomatic basis of this (exact) probabilistic correctness,
we propose the concept of compact soundness and completeness in Lem. \ref{lem_cmpt_sound_complete}, and, as a consequence,
a (syntactically checkable) condition for exact probabilistic reasoning is identified in Thm. \ref{thm_reas_exact_prob}.

\begin{lemma}[Compact soundness and completeness]\label{lem_cmpt_sound_complete}
  For any quantum program $S\in \qRP$ and any $\PQPT$s $P,Q$, it is the case that
  $$
  \trueEquality \vdash \tcor{P}{S}{Q} \mbox{ if and only if } \models_\mathbb{I} P = \fwp.S.Q
  $$
\end{lemma}
\begin{proof}
  Contained in the proof of Thm. \ref{thm_sound_complete_EqRP}.
\end{proof}

\begin{theorem}[Reasoning with exact probabilities]\label{thm_reas_exact_prob}
  For any quantum program $S\in \qRP$, any $\QPRED$s $P,Q$ and any $\delta\in [0, 1]$, it is the case that
  $$
  \trueEquality \vdash \tcor{\delta P}{S}{Q} \mbox{ if and only if } \forall \rho.\ \tr(P\rho) = 1 \implies \tr\big(Q\sem{S}(\rho)\big) = \delta
  $$
\end{theorem}
\begin{proof}
  Contained in the proof of Thm. \ref{thm_reas_EqPR_with_prob}.
\end{proof}

\begin{remark}\label{rem_reas_with_prob}
  Thm. \ref{thm_reas_exact_prob} (resp. Thm. \ref{thm_reas_approx_prob})
  establishes an axiomatic basis for reasoning with exact (resp. approximate) probabilities.
  Concretely speaking, if $\PQPT$s $P$ and $Q$ are chosen as projection operators,
  then Hoare's triple $\tcor{\delta P}{S}{Q}$ is able to express that
  ``In case the inputs of $S$ fall into the subspace $P$,
  the outputs will fall into $Q$ with probability $= \delta$ (resp. $\leq \delta$)''.
  In particular, when $P,Q$ are the identity operator $I$,
  Hoare's triple $\tcor{\delta I}{S}{I}$ represents termination on any input with probability $= \delta$ (resp. $\leq \delta$); and $\tcor{I}{S}{I}$ almost-sure termination in both cases.
  Note that during the reasoning with exact probabilities,
  the necessary L\"{o}wner ordering formulas are of the form $P = P'$, provided by $\trueEquality$.
\end{remark}

\begin{example}[Reasoning about $\RQMC$ with exact probabilities]\label{exam_pcor_RQMC}
Recall the game $\RQMC$ from Exms. \ref{exam_syn_RQMC}, \ref{exam_opSem_RQMC} and \ref{exam_deSem_RQMC}.
We illustrate how to do reasoning with exact probabilities
by showing probabilistic correctness and probabilistic termination of $\RQMC$.

{\bf (i) (Probabilistic correctness).}
To formally prove that Alice wins with probability $\frac{1}{3}$,
it suffices to prove the total correctness formula
\begin{equation*}
\atcor{\Big}{\frac{1}{3}I}{\RQMC}{\outprod{+}{+}}
\end{equation*}
By (A Init, R Comp), it suffices to prove
\begin{equation*}
\atcor{\Big}{\frac{\outprod{0}{0} + \outprod{1}{1}}{3}}{\CALL\ \Alice}{\outprod{+}{+}} \mbox{ and } \atcor{\Big}{\frac{\outprod{0}{0} + 4\outprod{1}{1}}{6}}{\CALL\ \Bob}{\outprod{+}{+}}
\end{equation*}
simultaneously. Defining $P_A^n$, $P_B^n$ by
\begin{equation*}
  P_A^n \triangleq \Big( \sum_{k \geq 1}^{2k-1 \leq n}\frac{1}{4^k} \Big) \outprod{0}{0} + \Big( \sum_{k \geq 1}^{2k \leq n} \frac{1}{4^k} \Big) \outprod{1}{1}, \quad P_B^n \triangleq \frac{1}{2} P_A^{n-1} + \frac{1}{2} \outprod{1}{1}
\end{equation*}
and $\Prem_A^n$, $\Prem_B^n$ by
\begin{equation*}
  \Prem_A^n \triangleq \tcor{P_A^n}{\CALL\ \Alice}{\outprod{+}{+}},
  \quad \Prem_B^n \triangleq \tcor{P_B^n}{\CALL\ \Bob}{\outprod{+}{+}}
\end{equation*}
by ($\tRule$ gRec), it suffices to prove, for all $n\geq 0$, that
\begin{eqnarray*}
  \Prem_A^n, \Prem_B^n & \vdash & \tcor{P_A^{n+1}}{\ifStat{\Box m\cdot M[q] = m \rightarrow S_m}}{\outprod{+}{+}} \\
  \Prem_A^n, \Prem_B^n & \vdash & \tcor{P_B^{n+1}}{\ifStat{\Box m\cdot M'[q] = m \rightarrow S_m'}}{\outprod{+}{+}}
\end{eqnarray*}
The proof is done by applying (R Case) to Hoare's triples (1-3) and (4,5) respectively.
$$
\begin{array}{clr}
  (1) & \tcor{\outprod{0}{0}}{q \starequal H}{\outprod{+}{+}} & \mbox{(A Unit)} \\
  (2) & \tcor{P_B^n}{\CALL\ \Bob}{\outprod{+}{+}} & \Prem_B^n \\
  (3) & \tcor{0}{\bottom}{\outprod{+}{+}} & \mbox{(A Bot)} \\
  (4) & \tcor{P_A^n}{\CALL\ \Alice}{\outprod{+}{+}} & \Prem_A^n \\
  (5) & \tcor{\outprod{1}{1}}{q \starequal H X}{\outprod{+}{+}} & \mbox{(A Unit)}
\end{array}
$$

{\bf (ii) (Probabilistic termination).}
To formally prove that $\RQMC$ terminates with probability $\frac{2}{3}$,
it suffices to prove the total correctness formula
\begin{equation*}
\atcor{\Big}{\frac{2}{3}I}{\RQMC}{I}
\end{equation*}
By (A Init, R Comp), it suffices to prove
\begin{equation*}
\atcor{\Big}{\frac{2}{3} I}{\CALL\ \Alice}{I} \mbox{ and } \atcor{\Big}{\frac{5}{6}I}{\CALL\ \Bob}{I}
\end{equation*}
simultaneously. The proof proceeds as above, by redefining $P_A^n$, $P_B^n$ by
\begin{equation*}
  P_A^n \triangleq \Big( \sum_{k \geq 1}^{2k-1 \leq n}\frac{1}{4^k} \Big) I + \Big( \sum_{k \geq 1}^{2k \leq n} \frac{1}{4^k} \Big) I, \quad P_B^n \triangleq \frac{1}{2} P_A^{n-1} + \frac{1}{2} I
\end{equation*}
and $\Prem_A^n$, $\Prem_B^n$ by
\begin{equation*}
  \Prem_A^n \triangleq \tcor{P_A^n}{\CALL\ \Alice}{I},
  \quad \Prem_B^n \triangleq \tcor{P_B^n}{\CALL\ \Bob}{I}
\end{equation*}
\end{example}

\begin{remark}[Counterexample, cf. Thm. \ref{thm_cmpt_sound_complete_qPP}]\label{rem_exam_counter_prob}
  The proof system for partial correctness of $\qRP$ has no compact soundness.
  To see this, suppose that $S \triangleq \CALL\ P_{\bottom}$ (cf. Exam. \ref{exam_bottom}),
  and $P,Q$ are $\PQPT$s with $\models_\mathbb{I} P \sqsubset I$.
  Then, by definition of $\fwlp$ (cf. Tab. \ref{tab_wp_wlp}), we have that
  \begin{eqnarray*}
    &\models_\mathbb{I}& P \sqsubset I = \fwlp.(\CALL\ P_{\bottom}).Q
  \end{eqnarray*}
  However, by (Rp pRec), it follows that
  \begin{eqnarray*}
    \trueEquality &\vdash& \pcor{P}{ \CALL\ P_{\bottom} }{Q}
  \end{eqnarray*}
  This reveals that the standard intermediate assertion method for partial correctness
  (i.e. the Turing-Floyd-Hoare Principle, cf. Tab. \ref{inter_asser_meth})
  can't be used universally for reasoning about recursive procedures with exact probabilities
  (even if involving nontermination).
\end{remark}

\subsection{Proof rules for while loops}\label{subsec_loop}

\begin{table}[!htbp]

\vspace{-12pt}

\centering

  \begin{minipage}[b]{\textwidth}
    \centering
    $$\dfrac{\pcor{P}{S}{M_0^{\dag} Q M_0 + M_1^{\dag} P M_1}}
    {\pcor{M_0^{\dag} Q M_0 + M_1^{\dag} P M_1}{\whileStat{ M[\lst{q}] = 1 }{ S }}{Q}}$$
    \subcaption{($\pRule$ Loop).}
    \label{par_loop_rule}
  \end{minipage}
  \vfill
  \begin{minipage}[b]{\textwidth}
    \centering
    $$\dfrac{\begin{array}{c}
    \exists\ \{P_n\}_{n \geq 0}^{\sqsubseteq} \mbox{ with } P_0 = 0 \mbox{ s.t.}  \\
    \tcor{P_{n+1}}{S}{M_0^\dag Q M_0 + M_1^\dag P_n M_1} \mbox{ for all } n\geq 0, \\
    P \sqsubseteq \bigsqcup_{n=0}^{\infty} P_n
    \end{array}
    }{\tcor{M_0^\dag Q M_0 + M_1^\dag P M_1}{\whileStat{ M[\lst{q}] = 1 }{ S }}{Q} }$$
    \subcaption{($\tRule$ Loop).}
    \label{tot_loop_rule}
  \end{minipage}

\vspace{-8pt}

\caption{Proof rules for while loops.}

\label{tot_rules} 

\vspace{-6pt}

\end{table}

The while-loop program $\WHILE \triangleq \whileStat{ M[\lst{q}] = 1 }{ S }$ with $M \triangleq \{M_0, M_1\}$,
can be defined as a call of tail recursion $\CALL\ T_{\WHILE}$, where $T_{\WHILE}$ has the body
\begin{eqnarray*}
  \IF &\triangleq& \ifStat{\Box m\cdot M[\lst{q}] = m \rightarrow S_m},
\end{eqnarray*}
with $S_0 \triangleq \SKIP$ and $S_1 \triangleq S; \CALL\ T_{\WHILE}$.

\vspace{1mm}
\noindent {\bf Partial correctness.}
To derive $\pcor{R}{\CALL\ T_{\WHILE}}{Q}$, by ($\pRule$ Rec), it suffices to show
\begin{eqnarray*}
  \pcor{R}{\CALL\ T_{\WHILE}}{Q} &\vdash& \pcor{R}{\IF}{Q}
\end{eqnarray*}
By (R Case), together with (A Skip) $\pcor{Q}{\SKIP}{Q}$, it suffices to show
\begin{eqnarray*}
  \pcor{R}{\CALL\ T_{\WHILE}}{Q} &\vdash& \pcor{P}{S; \CALL\ T_{\WHILE}}{Q}
\end{eqnarray*}
Here we let $R \triangleq M_0^{\dag} Q M_0 + M_1^{\dag} P M_1$.
By (R Comp), it suffices to derive
$$
\pcor{P}{S}{M_0^{\dag} Q M_0 + M_1^{\dag} P M_1}
$$
Thus, the proof rule (Rp Loop) for partial correctness of $\WHILE$ is designed in Tab. \ref{par_loop_rule}.

\vspace{1mm}
\noindent {\bf Total correctness.}
To prove total correctness of $\WHILE$,
by ($\tRule$ Rec),
we need to introduce a sequence of assertions at the same program point,
in which $R = M_0^{\dag} Q M_0 + M_1^{\dag} P M_1$ lies,
each with a different time point.
Instead of doing so, introduce $\{P_n\}_{n \geq 0}^{\sqsubseteq}$ with $P_0 = 0$
at the program point where $P$ lies.
(We remark that each time the data flow enters the loop body, the assertion $P_n$ will be encountered;
yet only after exiting the loop, should $Q$ be met.)
Thus, the proof rule ($\tRule$ Loop) for total correctness of $\WHILE$ can be designed in Tab. \ref{tot_loop_rule}.

\vspace{1mm}
\noindent {\bf Ying's rules revisited.}
Ying's proof rule for partial correctness of while loops is the same as ($\pRule$ Loop) \cite{Ying11}.
However, his solution to solving the issue of termination is based on a (semantical) notion of $(P, \epsilon)$-boundedness,
where $\epsilon$ bounds the trace of the diverging computation.
If, for any $\epsilon > 0$,
there is a $(M_1^\dag Q M_1, \epsilon)$-bound function of a while loop starting in $Q$,
then the loop terminates.
Thus, ($\pRule$ Loop) is used there jointly with the above condition
to prove total correctness of while loops \cite{Ying11}.

\begin{remark}\label{remark_tail_rec}
Illustrated by the process of
deducing proof rules for while loops from those for recursive procedures (and also by Exm. \ref{exam_pcor_RQMC}),
one can see that reasoning about a tail recursion doesn't necessarily require (R Subst),
since the Hoare's triple on a call statement as premise
can directly provide all the needed triples of that call statement in the body.
\end{remark}

\begin{lemma}[Cf. Props. 4.2.2 and 4.2.3 of \cite{ying2016foundations}]
  Let
  \begin{eqnarray*}
    \WHILE &\triangleq& \whileStat{ M[\lst{q}] = 1 }{ S }
  \end{eqnarray*}
  with $M \triangleq \{M_0, M_1\}$, and $Q$ a (non-parameterized) $\PQPT$.
  Define the $\PQPT$ $\Ff_{\fxp}^Q(\calX)$ by
  \begin{eqnarray*}
    \Ff_{\fxp}^Q(\calX) &\triangleq& M_0^{\dag} Q M_0 + M_1^{\dag} (\fxp.S.\calX) M_1
  \end{eqnarray*}
  where $\fxp\in \{\fwp, \fwlp\}$. It is the case that
  \begin{description}
    \item[(i)] $\fwlp.\WHILE.Q = \bigsqcap_{n = 0}^{\infty} Q_n$,
    where $Q_0 \triangleq I$ and $Q_{n+1} \triangleq \Ff_{\fwlp}^Q(Q_n)$, for all $n \geq 0$;
    \item[(ii)] $\fwp.\WHILE.Q = \bigsqcup_{n = 0}^{\infty} P_n$,
    where $P_0 \triangleq 0$ and $P_{n+1} \triangleq \Ff_{\fwp}^Q(P_n)$, for all $n \geq 0$.
  \end{description}
\end{lemma}

\vspace{1mm}
\noindent {\bf Synthesis of intermediate assertions.}
When applying (Rp Loop) to proving partial correctness of a while loop,
we have to provide the loop invariant $M_0^{\dag} Q M_0 + M_1^{\dag} P M_1$,
which can be selected as $\fwlp.\WHILE.Q$.
Note that $\fwlp.\WHILE.Q$ is the greatest fixed point of $\Ff_{\fwlp}^Q(\calX)$.
Similarly, in case of applying (Rt Loop),
we need to provide the intermediate assertions $\{P_n\}_{n \geq 0}$,
which can be the lease sequence of assertions generated by $\Ff_{\fwp}^Q(\calX)$.

\begin{example}[Almost-sure termination]
The following while loop
$$
\whileStat{ M[q] = 1 }{ \SKIP }
$$
with $M \triangleq \{ M_0 = M_1 \triangleq \frac{1}{\sqrt{2}} I_{q}\}$
is abstracted from quantum random walks with absorbing boundaries (modeled by quantum measurements) \cite{bach2004one}
and quantum Bernoulli factory for random number generation \cite{dale2015provable}.
To show its almost-sure termination, it suffices to prove the total correctness formula
$$
\tcor{I}{\whileStat{ M[q] = 1 }{ \SKIP }}{I}
$$
by using (Rt Loop), where the assertions $\{P_n \triangleq \sum_{i=1}^n \frac{1}{2^i} I_q \}_{n \geq 0}$ is generated by
$$
\Ff_{\fwp}^I(\calX) = M_0^{\dag} I M_0 + M_1^{\dag} (\fwp.\SKIP.\calX) M_1  =  \frac{1}{2} I_q + \frac{1}{2} \calX
$$
\end{example}

\begin{remark}
  Observe that the necessary intermediate assertions
  in proving correctness of while loops
  have a uniform (fixed-point) characterization,
  yet this is not always the case for (non-tail) recursion.
  This observation, jointly with Rem. \ref{remark_tail_rec}, entails that
  recursion is essentially more complex than while loops in the setting of program logics.
\end{remark}

\section{Expanded proof systems}\label{sec_loc}

\begin{table}[!htbp]

  \centering

  \begin{minipage}[b]{\textwidth}

    \centering

    \begin{tabular}{rcl}
       $P$ & $\triangleq$ & $D;S$ \\
       $D$ & $\triangleq$ & $\recDec{\langle \proc \rangle(\lst{y})}{S}$ \\
       $S$ & $\triangleq$ & $\LOCAL{\lst{q}}; S; \RELEASE{\lst{q}} \mid \CALL\ \langle \proc \rangle(\lst{p} ) \mid \bottom \mid \SKIP$ \\
         && $\mid q \assnequal \ket{0} \mid \lst{q} \starequal U \mid S_1;S_2 \mid \ifStat{\Box m\cdot M[\lst{q}] = m \rightarrow S_m}$
    \end{tabular}

    \subcaption{Syntax of $\EqRP$.}

    \label{syn_aux_fac}
  \end{minipage}

  \vfill

  \begin{minipage}[b]{\textwidth}
    \centering
    \begin{tabular}{rlcrl}
  (Loc) & \multicolumn{4}{l}{$\dfrac{\lst{r} \cap \Var(\rho) = \emptyset,\ |\lst{r}| = |\lst{q}|,\ \type(r_i) = \type(q_i)\ \forall\ i}{\langle \LOCAL{\lst{q}}; S; \RELEASE{\lst{q}}, \rho \rangle \xrightarrow{\epsilon} \langle S[\lst{r}/\lst{q}]; \RELEASE{\lst{r}}, \rho \otimes \ket{0}_{\lst{r}}\bra{0} \rangle}$} \\
  \specialrule{0em}{3pt}{3pt}
  (Rel) & $\dfrac{tr_{\lst{r}} \triangleq \sum_{i} \bra{i} \diamond \ket{i} \mbox{ with } \sum_{i} \outprod{i}{i} = I_{\lst{r}}}{\langle \RELEASE{\lst{r}}, \rho \rangle \xrightarrow{\epsilon} \langle E , tr_{\lst{r}}(\rho) \rangle}$
  & & (Proc) & $\dfrac{\recDec{\proc(\lst{y})}{S} \in D}{\langle\CALL\ \proc(\lst{p}),\rho\rangle \xrightarrow{\epsilon} \langle S[\lst{p}/\lst{y}], \rho\rangle}$
  \end{tabular}
    \vspace{2pt}
    \subcaption{Labeled transition rules for auxiliary facilities.}
    \label{nop_aux_fac}
  \end{minipage}
  \vfill
  \begin{minipage}[b]{\textwidth}
    \centering
    \begin{tabular}{rl}
    (Loc) & $\sem{\LOCAL{\lst{q}}; S; \RELEASE{\lst{q}}} = \tr_{ \lst{r} } \circ \sem{ S[\lst{r}/\lst{q}] } \circ (\ket{0}_{\lst{r}} \diamond \bra{0}_{\lst{r}})$ \\
    \specialrule{0em}{3pt}{3pt}
    (Proc) & $\sem{\CALL\ \proc_i(\lst{a}_i)} = \bigsqcup_{n = 0}^{\infty} \sem{ S_i^{(n)}[\lst{a}_i / \lst{y}_i] }$
  \end{tabular}
    \vspace{2pt}
    \subcaption{Denotational semantics for auxiliary facilities.}
    \label{den_aux_fac}
  \end{minipage}
  \vfill
  \begin{minipage}[b]{\textwidth}
    \centering
    \begin{tabular}{lcl}
    $\fwp.\big(\CALL\ \proc_i(\lst{a}_i)\big).P$ & = & $\bigsqcup_{n=0}^\infty \fwp.S_i^{(n)}[\lst{a}_i / \lst{y}_i].P$ \\
      \specialrule{0em}{3pt}{3pt}
    $\fwlp.\big(\CALL\ \proc_i(\lst{a}_i)\big).P$ & = & $\bigsqcap_{n=0}^\infty \fwlp.S_i^{(n)}[\lst{a}_i / \lst{y}_i].P$ \\
    \specialrule{0em}{3pt}{3pt}
    $\fxp(\LOCAL{\lst{q}}; S; \RELEASE{\lst{q}}).P$ & = & $\bra{0}_{\lst{r}} \big( \fxp.S[\lst{r}/\lst{q}].(P \otimes I_{\lst{r}}) \big) \ket{0}_{\lst{r}}$
  \end{tabular}
    \vspace{2pt}
    \subcaption{$\fwp$ and $\fwlp$ for auxiliary facilities --- $\fxp \in \{\fwp, \fwlp\}$.}
    \label{wp_aux_fac}
  \end{minipage}
  \vfill
  \begin{minipage}[b]{\textwidth}
    \centering
    \begin{tabular}{rlcrl}
    (R Loc) & $\dfrac{\pcor{P \otimes I_{\lst{r}}}{\lst{r}\assnequal \ket{0}; S[\lst{r}/\lst{q}]}{Q \otimes I_{\lst{r}}}}{\pcor{P}{\LOCAL{\lst{q}}; S; \RELEASE{\lst{q}}}{Q}}$ & &
    (R Adap) & $\dfrac{\pcor{P}{S}{Q}}{\pcor{ P[\lst{p}/\lst{q}] }{S[\lst{p}/\lst{q}]}{Q[\lst{p}/\lst{q}] }}$ \\
    \specialrule{0em}{4pt}{4pt}
    (Rp pRec) & \multicolumn{4}{l}{$\dfrac{\big\{ \pcor{P_i}{\CALL\ \proc_i(\lst{y}_i)}{Q_i} \big\}_{1 \leq i \leq n} \reVdash{\qBE} \bigwedge_{1 \leq i \leq n} \pcor{P_i}{S_i}{Q_i}}{\bigwedge_{1 \leq i \leq n} \pcor{P_i}{\CALL\ \proc_i(\lst{y}_i)}{Q_i}}$} \\
    \specialrule{0em}{4pt}{4pt}
    (Rt pRec) & \multicolumn{4}{l}{$\dfrac{\begin{array}{c}
    \mbox{for } 1 \leq i \leq n,\ \exists\ \{P_i^j\}_{j \geq 0}^{\sqsubseteq}  \mbox{ with } P_i^0 = 0 \mbox{ s.t.} \\
    \big\{ \atcor{\big}{P_i^j}{\CALL\ \proc_i(\lst{y}_i)}{Q_i} \big\}_{1 \leq i \leq n} \reVdash{\qBE} \bigwedge_{1 \leq i \leq n} \atcor{\big}{P_i^{j+1}}{S_i}{Q_i} \mbox{ for all } j \geq 0, \\
    P_i \sqsubseteq\bigsqcup_{j=0}^{\infty} P_i^j
    \end{array}
    }{\bigwedge_{1 \leq i \leq n} \tcor{P_i}{\CALL\ \proc_i(\lst{y}_i)}{Q_i} }$}
    \end{tabular}
    \vspace{2pt}
    \subcaption{Proof rules for auxiliary facilities.}
    \label{prf_aux_fac}
  \end{minipage}

  \vspace{-8pt}
\caption{QHL for $\EqRP$ --- $\qBE \triangleq \qBase + \mbox{(R Loc)} + \mbox{(R Adap)}$.}

\label{tab_QHL_eRqPL} 
  \vspace{-6pt}
\end{table}

In this section we augment the language $\qRP$ with facilities of variable localization and parameter passing,
with which the applicability scope of recursive quantum programs will be broadened.
This argumentation is also in line with the spirit of $\QPL$ \cite{Selinger2004}.

\subsection{Quantum variable localization}

\paragraph{Definition of the syntax.}
The construct of variable localization allows variables whose value is accessible only in a specified program fragment.
The syntax of such a construct with header $\LOCAL{\langle \mathit{qvar\_list} \rangle}$, body $S$ and tailer $\RELEASE{\langle \mathit{qvar\_list} \rangle}$ is given by
$$
\LOCAL{\langle \mathit{qvar\_list} \rangle};\ S;\ \RELEASE{\langle \mathit{qvar\_list} \rangle}
$$
\begin{example}[The system-environment model of a $\QOP$]\label{exam_quant_op}
The dynamics of an open quantum system (modeled by quantum variables $\lst{q}$),
interacted by a unitary interaction $U$ with an environment (modeled by $\lst{p}$ with initial state $\ket{0}$),
can be programmed as a structure of quantum variable localization:
$$
\LOCAL{\lst{p}};\ (\lst{p},\lst{q}) \starequal U;\ \RELEASE{\lst{p}}
$$
For instance, we can use this structure to program a circuit implementation
for the controlled operation $C^n(U)$ in Fig. 4.10 of \cite{nielsen2000quantum}.
The circuit makes use of a small number $(n - 1)$ of working qubits, which all start and end in the state $\ket{0}$.
\end{example}

\paragraph{Definition of the semantics.}
The intended meaning of the construct of variable localization is first expanding the state
with the default value $\ket{0}$ of local variables $\mathit{qvar\_list}$
declared by the header,
then executing the body possibly accessing $\mathit{qvar\_list}$,
and finally releasing $\mathit{qvar\_list}$ by the tailer.
Since the names of local variables, say $\lst{q}$, may conflict with those of state variables outside the structure,
to define the formal semantics of variable localization, we need a reservoir of fresh quantum variables,
say $\lst{r}$ (of the same length and of the same componentwise type as $\lst{q}$), to be used to express different instances of the local variables before binding them to values.
We can use partial trace function, say $tr_{\Hh_{\lst{r}}}$ (abbr. $tr_{\lst{r}}$),
to define the formal semantics of quantum variable localization
(cf. Tabs. \ref{nop_aux_fac} and \ref{den_aux_fac}).

\paragraph{Proof rules for the correctness.}
We invent the proof rule --- (R Loc)
--- for proving both partial and total correctness of variable localization (cf. Tab. \ref{prf_aux_fac}).
Here, by convention, $P \otimes I_{\lst{r}}$ and $Q \otimes I_{\lst{r}}$ can be simplified to $P$ and $Q$ respectively.

\vspace{1mm}
\noindent {\bf Intuition of (R Loc).}
Note that $\LOCAL{\lst{q}}; S; \RELEASE{\lst{q}}$ is semantically equivalent to
$\lst{r}\assnequal \ket{0}; S[\lst{r}/\lst{q}]$,
if the local variables $\lst{q}$ are thought of as the fresh global variables $\lst{r}$.
Under this assumption, Hoare's triple $\pcor{P}{\LOCAL{\lst{q}}; S; \RELEASE{\lst{q}}}{Q}$
is semantically equivalent to
\begin{equation}\label{QHT_loc_1}
  \pcor{P\otimes I_{\lst{r}}}{\lst{r}\assnequal \ket{0}; S[\lst{r}/\lst{q}]}{Q\otimes I_{\lst{r}}}
\end{equation}
By lifting this semantical equivalence to the syntactical case, (R Loc) follows naturally.

If, on the other hand, we choose to substitute $\lst{r}$ for $\lst{q}$ in assertions instead of in programs,
then we find that Hoare's triple $\pcor{P}{\LOCAL{\lst{q}}; S; \RELEASE{\lst{q}}}{Q}$
is semantically equivalent to
\begin{equation}\label{QHT_loc_2}
  \pcor{P[\lst{r}/\lst{q}]\otimes \voutprod{0}{\lst{q}}{0}}{S}{Q[\lst{r}/\lst{q}]\otimes I_{\lst{q}}}
\end{equation}

By elevating this semantical deduction to an inference rule, we obtain
$$
\mbox{(R' Loc)} \quad \dfrac{\pcor{P[\lst{r}/\lst{q}]\otimes \voutprod{0}{\lst{q}}{0}}{S}{Q[\lst{r}/\lst{q}]\otimes I_{\lst{q}}}}{\pcor{P}{\LOCAL{\lst{q}}; S; \RELEASE{\lst{q}}}{Q}}
$$

\vspace{1mm}
\noindent {\bf Comparison of (R Loc) and (R' Loc).}
To show the (syntactic) equivalence of the two proof rules,
it suffices to show Hoare's triples (\ref{QHT_loc_1}) and (\ref{QHT_loc_2}) can be transformed to each other.
This is the case by using (R Adap) (cf. Tab. \ref{prf_aux_fac}), (A Init) and (R Order).

To see the difference of the two proof rules,
we remark that (R Loc) is more in line with the formal semantics
and weakest (liberal) preconditions of variable localization (cf. Tabs. \ref{den_aux_fac} and \ref{wp_aux_fac}),
but (R' Loc) is purely inductive and thus more applicable in practice.

\begin{example}[Grover's search]\label{}
In Grover's original search algorithm (cf. Chap. 6 of \cite{nielsen2000quantum}),
we can use a (unitary) oracle $O$, defined by its action on the computational basis:
\begin{eqnarray*}
  \ket{x}\ket{y} &\overset{O}{\longrightarrow}& \ket{x}\ket{y\oplus f(x)}
\end{eqnarray*}
to check whether an item $x$ is a solution to the search problem.
Note that $f$ is the characteristic function of the search problem,
and the oracle ancilla $\ket{y}$ is a single qubit which is flipped if $f(x) = 1$,
and is unchanged otherwise.
It is useful to initialize the oracle ancilla in state $\ket{-}$,
in which case the state of the ancilla is not changed,
and $f(x)$ will occur as the exponent of a factor $(-1)^{f(x)}$ of relative phases.
Thus the action of the oracle can be rewritten:
\begin{eqnarray*}
  \ket{\varphi} \triangleq \sum_{x} \alpha_x \ket{x} &\overset{O}{\longrightarrow}& \ket{\psi} \triangleq \sum_{x} (-1)^{f(x)} \alpha_x \ket{x}
\end{eqnarray*}
Let quantum variables $q$, $p$ denote resp. $\ket{x}$, $\ket{y}$.
The verified program of $O$ is as follows.
$$
\begin{array}{clr}
   & \{ \voutprod{\varphi}{q}{\varphi} \} & \\
   & \LOCAL{p}; \{ \voutprod{\varphi}{q}{\varphi} \otimes \voutprod{0}{p}{0} \} & \mbox{(R' Loc)} \\
   & p \starequal HX; \{ \voutprod{\varphi}{q}{\varphi} \otimes \voutprod{-}{p}{-} \} & \mbox{(A Unit)} \\
   & (q,p) \starequal O; \{ \voutprod{\psi}{q}{\psi} \otimes \voutprod{-}{p}{-} \} & \mbox{(A Unit)} \\
   & p \starequal XH; \{ \voutprod{\psi}{q}{\psi} \otimes \voutprod{0}{p}{0} \} & \mbox{(A Unit)} \\
   & \{ \voutprod{\psi}{q}{\psi} \otimes I_{p} \} & \mbox{(R Order)} \\
   & \RELEASE{p}\ \{ \voutprod{\psi}{q}{\psi} \} & \mbox{(R' Loc)}
\end{array}
$$
\end{example}

\subsection{Quantum pointer passing}

Due to the no-cloning theorem,
it's impossible to realize all quantum value copying implicitly by a universal copying machine,
implemented as a unitary operator, as required by the principle of quantum mechanics.
The problem of parameter passing is well-understood
in the context of functional quantum programming languages,
and type systems for such quantum languages
usually rely on linear types and pointer-passing.
In other words, instead of passing values,
function calls pass wire identifiers, or register names.

\paragraph{Definition of the syntax.}
We now extend $\qRP$ with parameterized procedures.
The parameters consists of names of registers used in the global
environment that can be referred to inside the procedure.
A quantum program $P$ with parameterized procedures now have the form of Tab. \ref{syn_aux_fac}
(The extended programming language is coined $\EqRP$).
Note that the two lists of quantum variables $\lst{y}$ and $\lst{p}$,
called, respectively, formal and actual parameters,
are required to have equal length and equal componentwise type.

\paragraph{Definition of the semantics.}
The intended meaning of invoking a parameterized recursive procedure
is first expanding the body of the procedure with actual parameters in place of formal parameters,
which makes the procedure now able to act on already existing registers, and then executing this expanded body.
To define the formal semantics of parameterized procedures,
we can adapt the counterpart of non-parameterized procedures by adding syntactic substitution
(cf. Tabs. \ref{nop_aux_fac} and \ref{den_aux_fac}),
where the syntactic approximation of the bodies of parameterized procedures can be defined
by parameterizing Def. \ref{def_syn_approx} (cf. Def. \ref{def_para_unroll}).

\paragraph{Proof rules for the correctness.}
To parameterized recursive procedures $\proc_i(\lst{y}_i)$ with body $S_i$,
$1 \leq i \leq n$, the proof rule (Rp pRec) together with (R Adap)
can be used to prove their partial correctness;
the proof rule (Rt pRec) together with (R Adap)
can be used to prove their total correctness (cf. Tab. \ref{prf_aux_fac}).
Note that those non-parameterized proof rules for recursion,
e.g. (Rp gRec) defined in Tab. \ref{sim_par_rule},
can be thought of as a special case of their parameterized counterpart
by restricting the formal parameters $\lst{y}$ to $\emptyset$.

\begin{example}
Let the parameterized procedure $\proc(\lst{p})$ be defined by
$$
\recDec{\proc(\lst{p})}{\LOCAL{\lst{p}};\ \SKIP;\ \RELEASE{\lst{p}}}
$$
The operational semantics of $\CALL\ \proc(\lst{q})$ is developed step by step as
$$
\begin{array}{rl}
     & \langle \CALL\ \proc(\lst{q}), \rho\rangle \\
    \xrightarrow{\epsilon} & \langle \LOCAL{\lst{q}};\ \SKIP;\ \RELEASE{\lst{q}}, \rho\rangle \\
    \xrightarrow{\epsilon} & \langle \SKIP;\ \RELEASE{\lst{r}}, \rho \otimes \voutprod{0}{\lst{r}}{0}\rangle \\
    \xrightarrow{\epsilon} & \langle \RELEASE{\lst{r}}, \rho \otimes \voutprod{0}{\lst{r}}{0}\rangle \\
    \xrightarrow{\epsilon} & \langle E, \rho \rangle
  \end{array}
$$
The denotational semantics of $\CALL\ \proc(\lst{q})$ is defined by
\begin{eqnarray*}
  \sem{\CALL\ \proc(\lst{q})}
   &=& \bigsqcup_{n = 0}^{\infty} \sem{ ( \LOCAL{\lst{p}};\ \SKIP;\ \RELEASE{\lst{p}} )^{(n)}[\lst{q}/\lst{p}] } \\
   &=& \bigsqcup_{n = 0}^{\infty} \sem{ \LOCAL{\lst{q}};\ \SKIP;\ \RELEASE{\lst{q}} } \\
   &=& \bigsqcup_{n = 0}^{\infty} \tr_{ \lst{r} } \circ \sem{ \SKIP[\lst{r}/\lst{q}] } \circ (\ket{0}_{\lst{r}} \diamond \bra{0}_{\lst{r}}) = I_{\lst{q}}
\end{eqnarray*}
To prove the Hoare's triple
$$
\pcor{P_{\lst{q}}}{\CALL\ \proc(\lst{q})}{P_{\lst{q}}}
\quad \big(\mbox{resp. } \tcor{P_{\lst{q}}}{\CALL\ \proc(\lst{q})}{P_{\lst{q}}}\big)
$$
by (R Adap), it suffices to prove
$$
\pcor{P_{\lst{p}}}{\CALL\ \proc(\lst{p})}{P_{\lst{p}}}
$$
By (Rp pRec) \big(resp. (Rt pRec)\big), it suffices to prove
$$
\pcor{P_{\lst{p}}}{\LOCAL{\lst{p}};\ \SKIP;\ \RELEASE{\lst{p}}}{P_{\lst{p}}}
$$
following by (R Loc), together with (A Skip), (A Init) and (R Comp).
\end{example}

\subsection{Extension of previous results}

Various results developed in previous sections,
e.g. Thms. \ref{thm_sem}, \ref{thm_exp}, \ref{thm_reas_approx_prob}, \ref{thm_reas_exact_prob}
and the soundness and completeness results (for both partial and total correctness),
can be extended to covering the case of auxiliary facilities discussed in this section
(cf. Apps. \ref{app_sec_prqp}, \ref{app_sec_ass_exp} and \ref{app_sec_sound_complete}).

\section{Case studies}\label{sec_castu}

\subsection{Grover's fixed-point search}

\begin{table}
 \centering
 \begin{minipage}[b]{0.48\textwidth}
  \centering
  \begin{tabular}{rcl}
  \hline
  \multicolumn{3}{c}{$\recDec{\qSearch}{S}$} \\
  \hline
  $S$   & $\triangleq$ & $\ifStat{\Box m\cdot M[q_1] = m \rightarrow S_m}$ \\
  $S_0$ & $\triangleq$ & $q_2 \starequal V$ \\
  $S_1$ & $\triangleq$ & $q_1 \starequal U_{-1};$ \\
      & & $\CALL\ \qSearch;$ \\
      & & $q_2 \starequal R_t;$ \\
      & & $\CALL\ \qSearchDag;$ \\
      & & $q_2 \starequal R_s;$ \\
      & & $\CALL\ \qSearch;$ \\
      & & $q_1 \starequal U_{+1}$ \\
  \hline
\end{tabular}
  \subcaption{$\qSearch$}
  \label{qSearch}
  \end{minipage}
  \hfill
  \begin{minipage}[b]{0.48\textwidth}
  \centering
  \begin{tabular}{rcl}
  \hline
  \multicolumn{3}{c}{$\recDec{\qSearchDag}{S'}$}  \\
  \hline
  $S'$   & $\triangleq$ & $\ifStat{\Box m\cdot M[q_1] = m \rightarrow S_m'}$ \\
  $S_0'$ & $\triangleq$ & $q_2 \starequal V^\dag$ \\
  $S_1'$ & $\triangleq$ & $q_1 \starequal U_{-1};$ \\
      & & $\CALL\ \qSearchDag;$     \\
      & & $q_2 \starequal R_s^{\dag};$   \\
      & & $\CALL\ \qSearch;$ \\
      & & $q_2 \starequal R_t^{\dag};$  \\
      & & $\CALL\ \qSearchDag;$   \\
      & & $q_1 \starequal U_{+1}$ \\
  \hline
\end{tabular}
  \subcaption{$\qSearchDag$}
  \label{qSearchDag}
  \end{minipage}
  \vspace{-12pt}
\caption{
  $\qSearch$ and $\qSearchDag$ implement the search engine $V_n$ and its adjoint $V_n^{\dag}$.
  }
  \label{qSearch_examPrg}
  \vspace{-8pt}
\end{table}

Grover's search is a quantum algorithm of finding a target item in an unsorted database,
which has a square-root speedup over the corresponding classical algorithm.
The original idea is to design an iterative transformation
in a way that each iteration results in a small rotation of the moving state
in a two-dimensional plane spanned by the (orthogonal) target and nontarget vectors.
The moving state rotates in the plane from a starting state to the target state.
If we choose the right number of iterative steps,
the moving state will stop close to the target state, otherwise it will drift away.
Fixed-point Grover's search supplements the original search algorithm by permitting
the moving state converges monotonically to the target state as the number of iteration goes from zero to infinity.
This feature leads to robust search algorithms and also to new schemes for quantum control and error correction \cite{grover2005fixed}.

\paragraph{Programming the algorithm.}
Let $\ket{s}$ and $\ket{t}$ be the respective starting and target states in a Hilbert space,
where $\ket{s}$ is possibly superposed,
and $\ket{t}$ is a (not necessarily uniform) superposition of all possible solutions.
The core of the algorithm is to design a search engine
--- a series of unitary operators $\{V_n\}_{n \geq 0}$ inductively defined by
\begin{equation*}
  \begin{array}{cc}
     V_0 \triangleq V, & V_{n+1} \triangleq V_n R_s V_n^{\dag} R_t V_n
  \end{array}
\end{equation*}
where the $\frac{\pi}{3}$-phase shifts (i.e., unitary operators) $R_s$ and $R_t$ for $\ket{s}$ and $\ket{t}$ are defined as
\begin{equation*}
  \begin{array}{cc}
     R_s \triangleq I - \big(1-\exp(i \frac{\pi}{3})\big)\outprod{s}{s},
     & R_t \triangleq I - \big(1-\exp(i \frac{\pi}{3})\big)\outprod{t}{t}
  \end{array}
\end{equation*}
such that the resulting state $V_n \ket{s}$ after applying $V_n$ to $\ket{s}$
converges monotonically to $\ket{t}$ as $n$ approaches infinity, i.e.,
\begin{eqnarray*}
  \lim\limits_{n \to \infty }{V_n \ket{s}} &=& \ket{t}
\end{eqnarray*}
Then we are able to fetch information of the solution $\ket{t}$ by a measurement on $V_n \ket{s}$.
Note that $\ket{t}$ can be thought of as the least fixed point of a function induced by $\big\{V_n \ket{s}\big\}_{n \geq 0}$.
This is the reason why this version of Grover's search is called fixed-point Grover's search.
For the sake of simplicity, $\ket{s}$, $\ket{t}$ and $V$ are treated as black boxes.

To program the search engine,
let us use quantum variable $q_1$ to denote the moving state from $\ket{s}$ to $\ket{t}$.
To model the counter of the search engine (used to denote the subscript $n$ of $V_n$),
we shall use quantum variable $q_2$ over a $2^m$-dimensional Hilbert space $\Hh_c$
with orthonormal basis states $\big\{ \ket{n}\colon 0\leq n < 2^m\big\}$,
which can be used to encode an upper-bounded set of natural numbers.
Here $m$ should be large enough
so that the basis states of $\Hh_c$ suffice to encode all needed counter values.
We define $(+i)$-operator $U_{+i}$ of $\Hh_c$ by
$$
\begin{array}{c}
  U_{+i}\colon \ket{x}\rightarrow\ket{(x+i) \mod 2^m},
\end{array}
$$
to model the classical modular $(+i)$-operator,
and similarly for $(-i)$-operator $U_{-i}$.
Whether the value of the counter is zero can be identified by the outcome of the measurement
\begin{eqnarray*}
  M &\triangleq& \bigg\{ M_0 \triangleq \ket{0}\bra{0},\ M_1 \triangleq \sum_{i=1}^{2^m-1} \ket{i}\bra{i} \bigg\}
\end{eqnarray*}
Recursive quantum procedure $\qSearch$ for the search engine is designed in Tab. \ref{qSearch_examPrg}.

\paragraph{Partial correctness.}
We claim that, on input $\ket{n}_{q_1}\otimes \ket{s}_{q_2}$,
quantum activation statement $\CALL\ \qSearch$ executes with output $\ket{n}_{q_1}\otimes V_n \ket{s}_{q_2}$ (if terminates).
Formally speaking, the claim can be expressed as a partial-correctness formula:
\begin{equation}\label{eq_qSearch_par}
\apcor{\big}{\voutprod{n}{q_1}{n}\otimes \voutprod{s}{q_2}{s}}{\CALL\ \qSearch}{\voutprod{n}{q_1}{n}\otimes V_n \voutprod{s}{q_2}{s} V_n^\dag}
\end{equation}
Let $A$ be a quantum predicate variable on $\Hh_c$, and $B$ a quantum predicate variable on $\Hh_s$.
The $(i,j)$-component $\mixprod{i}{A}{j}$ of $A$ is abbreviated as $A_{i,j}$,
so $A = \sum_{i,j} A_{i,j} \outprod{i}{j}$.
To prove Hoare's triple (\ref{eq_qSearch_par}),
by (R Subst), together with the simultaneous substitution
$$
\big[\voutprod{n}{q_1}{n} / A,\ \voutprod{s}{q_2}{s} / B\big]
$$
it suffices to prove
\begin{equation*}
\apcor{\bigg}{\sum_{i=0}^{2^m-1} A_{i,i} \outprod{i}{i} \otimes B}{\CALL\ \qSearch}{\sum_{i=0}^{2^m-1} A_{i,i} \outprod{i}{i} \otimes V_i B V_i^\dag}
\end{equation*}
\big(Intuitively, the precondition (resp. postcondition) of the last Hoare's triple says that
the control flow arrives at each recursion depth $0 \leq i \leq 2^m-1$ (denoted by variable $q_1$) of procedure $\qSearch$ with probability $A_{i,i}$, and at depth $i$, the state of variable $q_2$ should satisfy
the predicate $B$ (resp. $V_i B V_i^\dag$).\big) Defining $\big \{ \Prem_i \big \}_{i = 0,1}$ by
\begin{eqnarray*}
  \Prem_1 &\triangleq& \apcor{\bigg}{\sum_{i=0}^{2^m-1} A_{i,i} \outprod{i}{i} \otimes B}{\CALL\ \qSearch}{\sum_{i=0}^{2^m-1} A_{i,i} \outprod{i}{i} \otimes V_i B V_i^\dag} \\
  \Prem_2 &\triangleq& \apcor{\bigg}{\sum_{i=0}^{2^m-1} A_{i,i} \outprod{i}{i} \otimes B}{\CALL\ \qSearchDag}{\sum_{i=0}^{2^m-1} A_{i,i} \outprod{i}{i} \otimes V_i^\dag B V_i}
\end{eqnarray*}
by (Rp gRec), it suffices to show that
\begin{eqnarray*}
  \trueOrder, \big\{ \Prem_i \big\}_{i = 0,1}
  &\reVdash{\qBase}& \apcor{\bigg}{\sum_{i=0}^{2^m-1} A_{i,i} \outprod{i}{i} \otimes B}{S}{\sum_{i=0}^{2^m-1} A_{i,i} \outprod{i}{i} \otimes V_i B V_i^\dag} \\
  \trueOrder, \big\{ \Prem_i \big\}_{i = 0,1}
  &\reVdash{\qBase}& \apcor{\bigg}{\sum_{i=0}^{2^m-1} A_{i,i} \outprod{i}{i} \otimes B}{S'}{\sum_{i=0}^{2^m-1} A_{i,i} \outprod{i}{i} \otimes V_i^\dag B V_i}
\end{eqnarray*}
The routine verification work is left to App. \ref{app_qSearch_par}.

\paragraph{Total correctness.}
We claim that quantum activation statement $\CALL\ \qSearch$, on input $\ket{n}_{q_1}\otimes \ket{s}_{q_2}$,
always terminates with output $\ket{n}_{q_1}\otimes V_n \ket{s}_{q_2}$.
In a formal way, the claim can be expressed as a total-correctness formula:
\begin{equation}\label{eq_qSearch_tot}
\atcor{\big}{\voutprod{n}{q_1}{n}\otimes \voutprod{s}{q_2}{s}}{\CALL\ \qSearch}{\voutprod{n}{q_1}{n}\otimes V_n \voutprod{s}{q_2}{s} V_n^\dag}
\end{equation}
Let quantum predicate variables $A$ and $B$ be as defined above.
To prove Hoare's triple (\ref{eq_qSearch_tot}), by (R Subst),
together with the simultaneous substitution
$$
\big[\voutprod{n}{q_1}{n} / A,\ \voutprod{s}{q_2}{s} / B \big]
$$
it suffices to prove
\begin{equation*}
\atcor{\bigg}{\sum_{i=0}^{2^m-1} A_{i,i} \outprod{i}{i} \otimes B}{\CALL\ \qSearch}{\sum_{i=0}^{2^m-1} A_{i,i} \outprod{i}{i} \otimes V_i B V_i^\dag}
\end{equation*}
Defining a sequence of $\PQPT$s $\big\{ P_j[A,B]\big\}_{j \geq 0}^{\sqsubseteq}$ by
\begin{eqnarray*}
  P_j[A,B] &\triangleq& \left\{
                      \begin{array}{ll}
                        \sum_{i=0}^{j} A_{i,i} \outprod{i}{i} \otimes B & \hbox{if $0 \leq j < 2^m$} \\
                        \sum_{i=0}^{2^m-1} A_{i,i} \outprod{i}{i} \otimes B & \hbox{if $j\geq 2^m$}
                      \end{array}
                    \right.
\end{eqnarray*}
and a set of premises $\big \{ \Prem_i^j \big \}_{i = 1,2}^{j \geq 0}$ by
\begin{eqnarray*}
  \Prem_1^j &\triangleq& \atcor{\bigg}{P_j[A,B]}{\CALL\ \qSearch}{\sum_{i=0}^{2^m-1} A_{i,i} \outprod{i}{i} \otimes V_i B V_i^\dag} \\
  \Prem_2^j &\triangleq& \atcor{\bigg}{P_j[A,B]}{\CALL\ \qSearchDag}{\sum_{i=0}^{2^m-1} A_{i,i} \outprod{i}{i} \otimes V_i^\dag B V_i}
\end{eqnarray*}
by (Rt gRec), it suffices to show, for all $j \geq 0$, that
\begin{eqnarray*}
  \trueOrder, \big \{\Prem_i^j \big \}_{i=1,2}
  &\reVdash{\qBase}& \atcor{\bigg}{P_{j+1}[A,B]}{S}{\sum_{i=0}^{2^m-1} A_{i,i} \outprod{i}{i} \otimes V_i B V_i^\dag} \\
  \trueOrder, \big \{\Prem_i^j \big \}_{i=1,2}
  &\reVdash{\qBase}& \atcor{\bigg}{P_{j+1}[A,B]}{S'}{\sum_{i=0}^{2^m-1} A_{i,i} \outprod{i}{i} \otimes V_i^\dag B V_i}
\end{eqnarray*}
The routine verification work is left to App. \ref{app_qSearch_tot}.

\subsection{Recursive quantum Fourier sampling}

\begin{table}[!htbp]

 \centering

 \begin{tabular}{rcl}
  \hline
  \multicolumn{3}{c}{$\recDec{\RQFS(q,Y)}{\ifStat{\Box m\cdot M[q] = m \rightarrow S_m}}$} \\
  \hline
  $S_0$ & $\triangleq$ & $q \starequal U_{+1};$  \\
  &&$\LOCAL{\bbX[q],Y'};$ \\
  &&$\big( \bbX[q], Y' \big) \starequal H^{\otimes n}\otimes HX;$ \\
  &&$\CALL\ \RQFS(q,Y');$   \\
  &&$\bbX[q] \starequal H^{\otimes n};$ \\
  &&$\big( \bbX[q], Y \big) \starequal \calG;$  \\
  &&$\bbX[q] \starequal H^{\otimes n};$ \\
  &&$\CALL\ \RQFS(q,Y');$           \\
  &&$\big( \bbX[q], Y' \big) \starequal H^{\otimes n}\otimes XH;$ \\
  &&$\RELEASE{\bbX[q],Y'};$ \\
  &&$q \starequal U_{-1}$ \\
  \hline
  $S_1$ & $\triangleq$ & $\big( \bbX[1], \ldots, \bbX[l], Y \big) \starequal \calO$ \\
  \hline
  $S_2$ & $\triangleq$ & $\bottom$       \\
  \hline
  \end{tabular}

  \caption{Recursive quantum procedure $\RQFS$.}

  \label{tab_RQFS_proc}

\end{table}

\paragraph{Problem description.}
Let us first briefly recall recursive quantum Fourier sampling \cite{mckague2012interactive}.
We begin by defining a type of tree.
Let $n,l$ be positive integers and consider a symmetric tree where each node,
except the leaves, has $2^n$ children, and the depth is $l$.
Let the root be labelled by $(\emptyset)$.
The root's children are labelled $(x_1)$ with $x_1\in \{0,1\}^n$.
Each child of $(x_1)$ is, in turn, labelled $(x_1,x_2)$ with $x_2 \in \{0,1\}^n$.
We continue until we have reached the leaves, which are labelled by $(x_1,\ldots,x_l)$.

Next we add the Fourier component to the tree.
We begin by fixing a computable function $g\colon \{0,1\}^n\rightarrow \{0,1\}$.
With each node of the tree $(x_1,\ldots,x_k)$
we associate a ``secret'' string $s_{(x_1,\ldots,x_k)} \in \{0,1\}^n$ s.t.
\begin{eqnarray*}
  g( s_{(x_1,\ldots,x_k)} ) &\triangleq& s_{(x_1,\ldots,x_{k-1})} \cdot x_k \mod 2
\end{eqnarray*}
(Here we take $s_{(x_1,\ldots,x_{k-1})}$ to mean $s_{(\emptyset)}$ if $k = 1$.)
In this way, each node's secret encodes one bit of information about its parent's secret.
Suppose we are given an oracle $o\colon (\{ 0, 1 \}^n)^l \rightarrow \{0,1\}$
for the leaves of the tree s.t.
\begin{eqnarray*}
  o(x_1,\ldots,x_l) &\triangleq& g\big(s_{(x_1,\ldots,x_l)}\big)
\end{eqnarray*}
Our goal is to find $g(s_{(\emptyset)})$.

\paragraph{Quantum solution.}
Define the descendant space $\Hh_{d}$ to be the $2^n$-dimensional Hilbert space with orthonormal basis states
--- $\big\{\ket{i}\colon 0 \leq i < 2^n \big\}$ --- to index each of $2^n$ children for any parental node.
The counting space $\Hh_{c}$, $(+i)$-operator $U_{+i}$ and $(-i)$-operator $U_{-i}$ of $\Hh_c$
are defined as previous (Here, to index the depth of the tree, we require that $l < 2^m$).

Let $p,q$ be quantum (individual) variables over $\Hh_c$,
$Y,Y',Z$ quantum variables over $\Hh_2$,
and $\bbX$ a quantum array-like variable over $\Hh_{d}$
with one argument, say $q$, indexing each component of the array,
s.t. each component $\bbX[q]$ acts like a quantum variable over $\Hh_d$.
We shall treat $\bbX[\ket{i}]$ as $\bbX[i]$ for simplicity.
Let quantum oracle $\calG$ \big(on $\bbX[q]$, $\itY$\big) calculate $g$ as
\begin{eqnarray*}
  \calG\ \ket{s}\ket{y} &\triangleq& \ket{s}\ket{y\oplus g(s)}
\end{eqnarray*}
Let quantum oracle $\calO$ \big(on $\bbX[1],\ldots,\bbX[l]$, $\itY$\big) model $o$ as
\begin{eqnarray*}
  \calO\ \ket{x_1}\ldots\ket{x_l}\ket{y} &\triangleq& \ket{x_1}\ldots\ket{x_l}\ket{y\oplus g(s_{(x_1,\ldots,x_l)})}
\end{eqnarray*}

The procedure for Recursive quantum Fourier sampling is designed in Tab. \ref{tab_RQFS_proc},
where the measurement $M$ (on $q$) is defined by
\begin{eqnarray*}
  M &\triangleq& \bigg\{ M_0 \triangleq \sum_{0 \leq i < l} \ket{i}\bra{i},\ M_1 \triangleq \ket{l}\bra{l},\
  M_2 \triangleq \sum_{l < i < 2^m} \ket{i}\bra{i} \bigg\}
\end{eqnarray*}

\vspace{1mm}
\noindent {\bf Notations and Definitions.}
Let $K$ be a quantum predicate variable over $\Hh_{c}$.
Let $\big\{\ket{i}\big\}_i$ be the computational basis of $\Hh_{q}$.
For notational convenience, the $(i,i)$-component $\mixprod{i}{K}{i}$ of $K$ is abbreviated as $K_{i}$.
(To see the intuitive meaning of $K_{i}$, we remark that $K = \sum_{i,j} \mixprod{i}{K}{j} \outprod{i}{j}$.)
Define $C(i)$ and $D(i)$ by
\begin{eqnarray*}
  C(i) &\triangleq& \bigotimes_{j=0}^i \Big ( \sum_{k = 0}^{ 2^n - 1 }\voutprod{k}{x_j}{k} \otimes \alpha_j \Big ) \\
  D(i) &\triangleq& \bigotimes_{j=0}^i \Big ( \sum_{k = 0}^{ 2^n - 1 }\voutprod{k}{x_j}{k} \otimes \beta_j \Big )
\end{eqnarray*}
where $\alpha_j$ and $\beta_j$ are defined by
\begin{eqnarray*}
  \alpha_j &\triangleq& \left\{
                      \begin{array}{ll}
                        \voutprod{0}{y_0}{0}, & \hbox{if $j = 0$} \\
                        \voutprod{-}{y_j}{-}, & \hbox{if $1 \leq j \leq l$}
                      \end{array}
                    \right. \\
  \beta_j &\triangleq& \left\{
                      \begin{array}{ll}
                        \voutprod{g(s_{(\emptyset)})}{y_0}{g(s_{(\emptyset)})}, & \hbox{if $j = 0$} \\
                        \voutprod{-}{y_j}{-}, & \hbox{if $1 \leq j \leq l$}
                      \end{array}
                    \right.
\end{eqnarray*}
Define the $\PQPT$s $P(K)$ and $Q(K)$ by
\begin{eqnarray*}
  P(K) &\triangleq& \sum_{i=0}^l  K_i \voutprod{i}{q}{i} \otimes C(i) \\
  Q(K) &\triangleq& \sum_{i=0}^l  K_i \voutprod{i}{q}{i} \otimes D(i)
\end{eqnarray*}
Intuitively, $P(K)$ \big(resp. $Q(K)$\big) says that the control flow arrives
at each recursion depth $0 \leq i \leq l$ (denoted by variable $q$) of algorithm $\RQFS$ with probability $K_i$
(where $i$ corresponds to each level of the tree, in particular, $i=0$ points to the root and $i=l$ to the leaves),
and at depth $i$, variable $\bbY[0]$ lies in the state $\ket{0}$ \big(resp. $\ket{g(s_{(\emptyset)})}$\big),
$\bbY[j]$ with $1 \leq j \leq i$ in $\ket{-}$,
and $\bbX[j]$ with $0 \leq j \leq i$ can lie in any state (by the predicate $I_{x_j}$).

\paragraph{Partial correctness.}
We claim the partial correctness of $\RQFS$ by proving Hoare's triple
\begin{equation*}
\apcor{\big}{\voutprod{0}{q}{0}\otimes \voutprod{0}{y_0}{0}}{\CALL\ \RQFS\big( q, \bbY[q] \big)}{\voutprod{0}{q}{0}\otimes \voutprod{g(s_{(\emptyset)})}{y_0}{g(s_{(\emptyset)})}}
\end{equation*}
By (R Subst), together with the substitution $\big[ \voutprod{0}{q}{0} / K \big]$, it suffices to show
$$
(\Prem \triangleq \:) \pcor{P(K)}{\CALL\ \RQFS\big( q, \bbY[q] \big)}{Q(K)}
$$
By (Rp pRec), it suffices to show
\begin{eqnarray*}
  \trueOrder, \Prem &\reVdash{\qBE}& \pcor{P(K)}{\ifStat{\Box m\cdot M[q] = m \rightarrow S_m}}{Q(K)}
\end{eqnarray*}
The routine verification work is left to App. \ref{app_RQFS_par}.

\paragraph{Total correctness.}
We claim the total correctness of $\RQFS$ by proving Hoare's triple
\begin{equation*}
\atcor{\big}{\voutprod{0}{q}{0}\otimes \voutprod{0}{y_0}{0}}{\CALL\ \RQFS\big( q, \bbY[q] \big)}{\voutprod{0}{q}{0}\otimes \voutprod{g(s_{(\emptyset)})}{y_0}{g(s_{(\emptyset)})}}
\end{equation*}
By (R Subst), together with the substitution $\big[ \voutprod{0}{q}{0} / K \big]$, it suffices to show
\begin{equation*}
\tcor{P(K)}{\CALL\ \RQFS\big( q, \bbY[q] \big)}{Q(K)}
\end{equation*}
Defining a sequence of $\PQPT$s $\big\{P_h(K)\big\}_{h \geq 0}^{\sqsubseteq}$ by
\begin{eqnarray*}
   P_h(K)&\triangleq& \left\{
                    \begin{array}{ll}
                      \sum_{i=l+1-h}^l K_i \voutprod{i}{q}{i} \otimes C(i) & \hbox{$0 \leq h < l$} \\
                      P(K), & \hbox{otherwise}
                    \end{array}
                  \right.
\end{eqnarray*}
and a set of premises $\big\{\Prem_h\big\}_{h \geq 0}$ by
\begin{eqnarray*}
  \Prem_h &\triangleq& \atcor{\big}{P_h(K)}{\CALL\ \RQFS\big( q, \bbY[q] \big)}{Q(K)}
\end{eqnarray*}
by (Rt pRec), it suffices to show, for all $h \geq 0$, that
\begin{eqnarray*}
  \trueOrder, \Prem_h &\vdash_{\qBE}& \atcor{\big}{P_{h+1}(K)}{\ifStat{\Box m\cdot M[q] = m \rightarrow S_m}}{Q(K)}
\end{eqnarray*}
The routine verification work is left to App. \ref{app_RQFS_tot}.

\section{Related work}\label{sec_relwk}

In this section, we will compare our quantum Hoare logic (QHL)
with classical deterministic Hoare logic (DHL), classical probabilistic Hoare logic (PHL), and other QHLs.
We shall discuss global and local reasoning in the setting of quantum computation.
Comparison with other related work, e.g. on termination and loop invariants etc., is also presented.

\subsection{Comparison with DHL}

Verification techniques for classical (deterministic) recursive programs
have been systematically developed since Hoare's pioneering work \cite{Hoare1971}
(see, e.g., Chaps. 4, 5 of \cite{Apt} and Chap. 6 of \cite{Francez}).
However, the semantics of quantum programs and quantum assertions are complex-matrix-based,
neither relations nor formulas on a discrete space as in the classical case.
This semantical difference entails that
there is no uniform mechanism for encoding any finite computational sequence of quantum programs
(For classical coding functions, e.g. G\"{o}del's $\beta$-function and pairing functions,
the reader is referred to \cite{boolos2002computability}),
so quantum logics have no closed (finite) form as in classical logics.
To address the challenge, we have to accept an infinite representation of program semantics and assertions,
e.g. the upper and lower limits,
and try to find a closed-form (e.g. fixed-point) characterization for them (cf. Rem. \ref{rem_closure_den}).

\subsection{Comparison with PHL}

A theory of probabilistic recursion and recursive probabilistic programming
has recently been developed in a series of papers \cite{LagoZG14,BreuvartLH17}.
In the last few years, significant progress has been made in verifying recursive probabilistic programs,
including weakest pre-expectation calculus and proof rules \cite{OlmedoKKM16,KaminskiKMO18}
as well as termination problem \cite{kobayashi2020termination}.
We remark that the assertion languages of \cite{OlmedoKKM16,KaminskiKMO18,barthe2018assertion} have no parameters.
The lack of parameters would restrict the scope of application of their logics to the case of tail recursion
(cf. Exms \ref{exam_counter_par}, \ref{exam_counter_tot} and Rem. \ref{remark_tail_rec}).
The introduction of parameters for random variables,
which can be thought of as a probabilistic degenerate of quantum predicate variables proposed in this paper,
could help to repair the technical deficiency for the sake of completeness.
Note that the termination problem of probabilistic higher-order programs \cite{kobayashi2020termination}
is reduced to the probabilistic reachability problem of a higher-order extension of recursive Markov chains
(not in an axiomatic way).

\subsection{Comparison with other QHLs}

Other quantum Hoare-like logics,
e.g. \cite{chadha2006reasoning,feng2007proof,kakutani2009logic},
have been developed with the Turing-Floyd-Hoare principle.
For the landscape of some of these quantum logics, the reader is referred to the survey paper \cite{rand2019verification}.
The assertion languages of these QHLs surveyed in \cite{rand2019verification} are compositional
and contain parameters, yet are not purely quantum (matrix-based).
It's worthy to note that a kind of (existentially quantified) ghost variables,
which can be entangled with program variables,
are introduced in the assertion language of \cite{Unruh19a}.
In contrast, the quantum predicate variables defined in the current paper
are free higher-order variables (if program variables are seen as first-order variables),
and can be used to describe properties of entangled program variables
(yet quantum predicate variables never entangle with program variables).
The ghost variables can be used to simulate local variables,
but the interaction between ghost and program variables is uncontrolled due to existential quantification,
thus Unruh's QHL fails to reason about a particular quantum operation as in Exm. \ref{exam_quant_op}.
The paper \cite{hung2019quantitative} develops a formal semantics for erroneous quantum $\WHILE$-programs,
as well as a logic for reasoning about their robustness (i.e. error bounds of outputs).
We remark that assertions of this QHL are $\QPRED$s,
and that inference rules for a while loop are not designed in a purely syntactical way.
An applied QHL \cite{zhou2019applied} (also supporting reasoning about robustness of quantum programs)
is defined by restricting assertions of QHL \cite{Ying11} (i.e. $\QPRED$s)
to subspaces of a Hilbert space (i.e. projection operators).
Thus the applied QHL fails to do reasoning with probabilities $\delta\in (0,1)$
(cf. Rem. \ref{rem_reas_with_prob} for a detailed comparison).

\subsection{Comparison with local reasoning}

Hoare logic is a general framework for global reasoning about programs.
As an alternative technical line,
local reasoning is proposed by using the Frame Rule under the assumption that
the underlying program states be separated
(in, e.g., separation logic \cite{reynolds2002separation, batz2019quantitative}
and quantum relational Hoare logic \cite{Unruh19b,BartheHYYZ20}).
In principle, the scope of applicability of global reasoning is larger than that of local reasoning.
In practice, local reasoning could bring some advantages, e.g. simpler semantics and shorter assertions,
leading to a good trade-off between application scope and program scale.
However, in the quantum field, global reasoning is indispensable,
since sometimes a global quantum state should be treated as an inseparable entity (e.g. an entangled state),
in which case a global assertion (e.g. an inseparable Hermitian matrix) should be used instead of concatenating local assertions by using the Frame Rule.

\subsection{Comparison with other related work}

The (almost-sure) termination problem of quantum programs \cite{li2018termination}
is reduced to the realisability and synthesis problem of linear ranking super-martingales,
which can be solved by resorting to an SDP (Semi-Definite Programming) solver
(This approach comes from the constraint-based solution to the termination problem of probabilistic programs \cite{chakarov2013probabilistic,chatterjee2016algorithmic,mciver2017new,kaminski2019hardness}).
In contrast, this paper proposes an axiomatic approach to
the probabilistic (including almost-sure) termination problem of quantum programs.
Characterizations and generation of loop invariants of quantum programs \cite{ying2017invariants}
have been developed in a framework of super-operator-valued transition system based on SDP.
By contrast, this paper discusses the existence issue of
a uniform fixed-point characterisation for (general) recursive invariants,
and explains why the synthesis of recursive invariants can not be automated completely.
A theorem prover \cite{liu2019formal} for verifying partial correctness of quantum $\WHILE$-programs
using QHL \cite{Ying11} has been developed within the framework of Isabelle/HOL
(A detailed partial-correctness proof of Grover's original search algorithm is implemented thereof).
Quantum relational Hoare logic \cite{Unruh19b,BartheHYYZ20}
allows to reason about how the outputs of two quantum programs relate given a relation between their inputs.
Finally, quantum Hoare type theory \cite{singhal2020quantum} is inspired by classical Hoare type theory
and extends Quantum IO Monad by indexing it with pre- and post-conditions that serve as program specifications,
which has the potential to be a unified system for programming, specifying, and reasoning about quantum programs.

\section{Conclusion and future work}\label{sec_con}

This paper has systematically investigated the problem of how to verify
parameterized recursive quantum programs with ancilla data and probabilistic control.
We have defined a new quantum assertion logic, a parameterized extension of quantum predicates,
so that, by using formulas of this assertion logic as pre- and post-conditions,
Hoare's approach for both partial and total correctness can be extended to the case of general recursive procedures
(i.e. soundness and completeness of our quantum Hoare logic).
The assertion logic makes it realizable to reduce
reasoning about quantum programs with both approximate and exact probabilities
to a total-correctness proof.
In particular, two counterexamples for illustrating
incompleteness of quantum predicates in verifying recursive procedures,
and, one counterexample for showing the failure of reasoning with exact probabilities based on partial correctness,
have also been constructed.
The usefulness of the quantum Hoare logic has been illustrated by three main examples:
recursive quantum Markov chain (with probabilistic control),
fixed-point Grover's search,
and recursive quantum Fourier sampling.

For the future work, we find that satisfiability of L\"{o}wner order (resp. equality) on the quantum assertion logic
has to be reduced to positivity (resp. equality) of super operators on separable states (cf. Rem. \ref{rem_lowner_order}).
Note that complete positivity of a super operator can be reduced to positivity of a linear operator
(called Choi-Jamiolkowski isomorphism or Channel-state duality;
for more information on this topic, cf., e.g., \cite{wolf2012quantum}).
However, a useful characterization of positivity (resp. equality) of super operators is still out of reach
(For more information on this area of research, cf., e.g., \cite{johnstonequivalencesI,johnstonequivalencesII}).
We could also consider introducing quantifiers into the assertion logic.
It would be interesting to compare this logic with other quantum logics,
e.g. linear logic \cite{girard1987linear} (a logic for a linear use of quantum resources),
within the framework of orthomodular lattices or category theory.
On the other hand, quantum programs with quantum control have been studied in a series of previous work
\big(see, e.g., \cite{altenkirch2005functional}, Chaps. 6 and 7 of \cite{ying2016foundations},
\cite{BadescuP15, SabryVV18}\big),
and the notions of quantum recursion with quantum control were already introduced there.
However, we are still at a very beginning stage of understanding quantum recursions with quantum controls,
and we feel that a program logic for reasoning about them requires some ideas very different from those used in this paper.

%

\bibliographystyle{ACM-Reference-Format}
\bibliography{RQFHL}


\begin{thebibliography}{76}


\ifx \showCODEN    \undefined \def \showCODEN     #1{\unskip}     \fi
\ifx \showDOI      \undefined \def \showDOI       #1{#1}\fi
\ifx \showISBNx    \undefined \def \showISBNx     #1{\unskip}     \fi
\ifx \showISBNxiii \undefined \def \showISBNxiii  #1{\unskip}     \fi
\ifx \showISSN     \undefined \def \showISSN      #1{\unskip}     \fi
\ifx \showLCCN     \undefined \def \showLCCN      #1{\unskip}     \fi
\ifx \shownote     \undefined \def \shownote      #1{#1}          \fi
\ifx \showarticletitle \undefined \def \showarticletitle #1{#1}   \fi
\ifx \showURL      \undefined \def \showURL       {\relax}        \fi
\providecommand\bibfield[2]{#2}
\providecommand\bibinfo[2]{#2}
\providecommand\natexlab[1]{#1}
\providecommand\showeprint[2][]{arXiv:#2}

\bibitem[\protect\citeauthoryear{Abhari, Faruque, Dousti, Svec, Catu,
  Chakrabati, Chiang, Vanderwilt, Black, and Chong}{Abhari
  et~al\mbox{.}}{2012}]%
        {abhari2012scaffold}
\bibfield{author}{\bibinfo{person}{Ali~J Abhari}, \bibinfo{person}{Arvin
  Faruque}, \bibinfo{person}{Mohammad~J Dousti}, \bibinfo{person}{Lukas Svec},
  \bibinfo{person}{Oana Catu}, \bibinfo{person}{Amlan Chakrabati},
  \bibinfo{person}{Chen-Fu Chiang}, \bibinfo{person}{Seth Vanderwilt},
  \bibinfo{person}{John Black}, {and} \bibinfo{person}{Fred Chong}.}
  \bibinfo{year}{2012}\natexlab{}.
\newblock \bibinfo{booktitle}{\emph{Scaffold: Quantum programming language}}.
\newblock \bibinfo{type}{{T}echnical {R}eport}. \bibinfo{institution}{PRINCETON
  UNIV NJ DEPT OF COMPUTER SCIENCE}.
\newblock


\bibitem[\protect\citeauthoryear{Altenkirch and Grattage}{Altenkirch and
  Grattage}{2005}]%
        {altenkirch2005functional}
\bibfield{author}{\bibinfo{person}{Thorsten Altenkirch} {and}
  \bibinfo{person}{Jonathan Grattage}.} \bibinfo{year}{2005}\natexlab{}.
\newblock \showarticletitle{A functional quantum programming language}. In
  \bibinfo{booktitle}{\emph{Logic in Computer Science, 2005. LICS 2005.
  Proceedings. 20th Annual IEEE Symposium on}}. IEEE,
  \bibinfo{pages}{249--258}.
\newblock


\bibitem[\protect\citeauthoryear{Apt}{Apt}{1981}]%
        {apt1981ten}
\bibfield{author}{\bibinfo{person}{Krzysztof~R Apt}.}
  \bibinfo{year}{1981}\natexlab{}.
\newblock \showarticletitle{Ten years of Hoare's logic: A survey¡ªPart I}.
\newblock \bibinfo{journal}{\emph{ACM Transactions on Programming Languages and
  Systems (TOPLAS)}} \bibinfo{volume}{3}, \bibinfo{number}{4}
  (\bibinfo{year}{1981}), \bibinfo{pages}{431--483}.
\newblock


\bibitem[\protect\citeauthoryear{Apt, de~Boer, and Olderog}{Apt
  et~al\mbox{.}}{2009}]%
        {Apt}
\bibfield{author}{\bibinfo{person}{Krzysztof~R. Apt}, \bibinfo{person}{Frank~S.
  de Boer}, {and} \bibinfo{person}{Ernst{-}R{\"{u}}diger Olderog}.}
  \bibinfo{year}{2009}\natexlab{}.
\newblock \bibinfo{booktitle}{\emph{Verification of Sequential and Concurrent
  Programs}}.
\newblock \bibinfo{publisher}{Springer}.
\newblock
\showISBNx{978-1-84882-744-8}
\urldef\tempurl%
\url{https://doi.org/10.1007/978-1-84882-745-5}
\showDOI{\tempurl}


\bibitem[\protect\citeauthoryear{Apt and Olderog}{Apt and Olderog}{2019}]%
        {apt2019fifty}
\bibfield{author}{\bibinfo{person}{Krzysztof~R Apt} {and}
  \bibinfo{person}{Ernst-R{\"u}diger Olderog}.}
  \bibinfo{year}{2019}\natexlab{}.
\newblock \showarticletitle{Fifty years of Hoare¡¯s logic}.
\newblock \bibinfo{journal}{\emph{Formal Aspects of Computing}}
  \bibinfo{volume}{31}, \bibinfo{number}{6} (\bibinfo{year}{2019}),
  \bibinfo{pages}{751--807}.
\newblock


\bibitem[\protect\citeauthoryear{Bach, Coppersmith, Goldschen, Joynt, and
  Watrous}{Bach et~al\mbox{.}}{2004}]%
        {bach2004one}
\bibfield{author}{\bibinfo{person}{Eric Bach}, \bibinfo{person}{Susan
  Coppersmith}, \bibinfo{person}{Marcel~Paz Goldschen}, \bibinfo{person}{Robert
  Joynt}, {and} \bibinfo{person}{John Watrous}.}
  \bibinfo{year}{2004}\natexlab{}.
\newblock \showarticletitle{One-dimensional quantum walks with absorbing
  boundaries}.
\newblock \bibinfo{journal}{\emph{J. Comput. System Sci.}}
  \bibinfo{volume}{69}, \bibinfo{number}{4} (\bibinfo{year}{2004}),
  \bibinfo{pages}{562--592}.
\newblock


\bibitem[\protect\citeauthoryear{Badescu and Panangaden}{Badescu and
  Panangaden}{2015}]%
        {BadescuP15}
\bibfield{author}{\bibinfo{person}{Costin Badescu} {and}
  \bibinfo{person}{Prakash Panangaden}.} \bibinfo{year}{2015}\natexlab{}.
\newblock \showarticletitle{Quantum Alternation: Prospects and Problems}. In
  \bibinfo{booktitle}{\emph{Proceedings 12th International Workshop on Quantum
  Physics and Logic, {QPL} 2015, Oxford, UK, July 15-17, 2015.}}
  \bibinfo{pages}{33--42}.
\newblock
\urldef\tempurl%
\url{https://doi.org/10.4204/EPTCS.195.3}
\showDOI{\tempurl}


\bibitem[\protect\citeauthoryear{Baltag and Smets}{Baltag and Smets}{2011}]%
        {baltag2011quantum}
\bibfield{author}{\bibinfo{person}{Alexandru Baltag} {and}
  \bibinfo{person}{Sonja Smets}.} \bibinfo{year}{2011}\natexlab{}.
\newblock \showarticletitle{Quantum logic as a dynamic logic}.
\newblock \bibinfo{journal}{\emph{Synthese}} \bibinfo{volume}{179},
  \bibinfo{number}{2} (\bibinfo{year}{2011}), \bibinfo{pages}{285--306}.
\newblock


\bibitem[\protect\citeauthoryear{Barthe, Espitau, Gaboardi, Gr{\'e}goire, Hsu,
  and Strub}{Barthe et~al\mbox{.}}{2018}]%
        {barthe2018assertion}
\bibfield{author}{\bibinfo{person}{Gilles Barthe}, \bibinfo{person}{Thomas
  Espitau}, \bibinfo{person}{Marco Gaboardi}, \bibinfo{person}{Benjamin
  Gr{\'e}goire}, \bibinfo{person}{Justin Hsu}, {and}
  \bibinfo{person}{Pierre-Yves Strub}.} \bibinfo{year}{2018}\natexlab{}.
\newblock \showarticletitle{An Assertion-Based Program Logic for Probabilistic
  Programs}. In \bibinfo{booktitle}{\emph{European Symposium on Programming}}.
  Springer, Cham, \bibinfo{pages}{117--144}.
\newblock


\bibitem[\protect\citeauthoryear{Barthe, Hsu, Ying, Yu, and Zhou}{Barthe
  et~al\mbox{.}}{2020}]%
        {BartheHYYZ20}
\bibfield{author}{\bibinfo{person}{Gilles Barthe}, \bibinfo{person}{Justin
  Hsu}, \bibinfo{person}{Mingsheng Ying}, \bibinfo{person}{Nengkun Yu}, {and}
  \bibinfo{person}{Li Zhou}.} \bibinfo{year}{2020}\natexlab{}.
\newblock \showarticletitle{Relational proofs for quantum programs}.
\newblock \bibinfo{journal}{\emph{{PACMPL}}}  \bibinfo{volume}{4}
  (\bibinfo{year}{2020}), \bibinfo{pages}{21:1--21:29}.
\newblock
\urldef\tempurl%
\url{https://doi.org/10.1145/3371089}
\showDOI{\tempurl}


\bibitem[\protect\citeauthoryear{Batz, Kaminski, Katoen, Matheja, and
  Noll}{Batz et~al\mbox{.}}{2019}]%
        {batz2019quantitative}
\bibfield{author}{\bibinfo{person}{Kevin Batz},
  \bibinfo{person}{Benjamin~Lucien Kaminski}, \bibinfo{person}{Joost-Pieter
  Katoen}, \bibinfo{person}{Christoph Matheja}, {and} \bibinfo{person}{Thomas
  Noll}.} \bibinfo{year}{2019}\natexlab{}.
\newblock \showarticletitle{Quantitative separation logic: a logic for
  reasoning about probabilistic pointer programs}.
\newblock \bibinfo{journal}{\emph{Proceedings of the ACM on Programming
  Languages}} \bibinfo{volume}{3}, \bibinfo{number}{POPL}
  (\bibinfo{year}{2019}), \bibinfo{pages}{1--29}.
\newblock


\bibitem[\protect\citeauthoryear{Bergstra and Tucker}{Bergstra and
  Tucker}{1982}]%
        {BergstraT82}
\bibfield{author}{\bibinfo{person}{Jan~A. Bergstra} {and}
  \bibinfo{person}{J.~V. Tucker}.} \bibinfo{year}{1982}\natexlab{}.
\newblock \showarticletitle{Expressiveness and the Completeness of Hoare's
  Logic}.
\newblock \bibinfo{journal}{\emph{J. Comput. Syst. Sci.}} \bibinfo{volume}{25},
  \bibinfo{number}{3} (\bibinfo{year}{1982}), \bibinfo{pages}{267--284}.
\newblock
\urldef\tempurl%
\url{https://doi.org/10.1016/0022-0000(82)90013-7}
\showDOI{\tempurl}


\bibitem[\protect\citeauthoryear{Bernstein and Vazirani}{Bernstein and
  Vazirani}{1997}]%
        {bernstein1997quantum}
\bibfield{author}{\bibinfo{person}{Ethan Bernstein} {and}
  \bibinfo{person}{Umesh Vazirani}.} \bibinfo{year}{1997}\natexlab{}.
\newblock \showarticletitle{Quantum complexity theory}.
\newblock \bibinfo{journal}{\emph{SIAM Journal on computing}}
  \bibinfo{volume}{26}, \bibinfo{number}{5} (\bibinfo{year}{1997}),
  \bibinfo{pages}{1411--1473}.
\newblock


\bibitem[\protect\citeauthoryear{Boolos, Burgess, and Jeffrey}{Boolos
  et~al\mbox{.}}{2002}]%
        {boolos2002computability}
\bibfield{author}{\bibinfo{person}{George~S Boolos}, \bibinfo{person}{John~P
  Burgess}, {and} \bibinfo{person}{Richard~C Jeffrey}.}
  \bibinfo{year}{2002}\natexlab{}.
\newblock \bibinfo{booktitle}{\emph{Computability and logic}}.
\newblock \bibinfo{publisher}{Cambridge university press}.
\newblock


\bibitem[\protect\citeauthoryear{Breuvart, Lago, and Herrou}{Breuvart
  et~al\mbox{.}}{2017}]%
        {BreuvartLH17}
\bibfield{author}{\bibinfo{person}{Flavien Breuvart}, \bibinfo{person}{Ugo~Dal
  Lago}, {and} \bibinfo{person}{Agathe Herrou}.}
  \bibinfo{year}{2017}\natexlab{}.
\newblock \showarticletitle{On Higher-Order Probabilistic Subrecursion}. In
  \bibinfo{booktitle}{\emph{Foundations of Software Science and Computation
  Structures - 20th International Conference, {FOSSACS} 2017, Held as Part of
  the European Joint Conferences on Theory and Practice of Software, {ETAPS}
  2017, Uppsala, Sweden, April 22-29, 2017, Proceedings}}.
  \bibinfo{pages}{370--386}.
\newblock
\urldef\tempurl%
\url{https://doi.org/10.1007/978-3-662-54458-7\_22}
\showDOI{\tempurl}


\bibitem[\protect\citeauthoryear{Brunet and Jorrand}{Brunet and
  Jorrand}{2004}]%
        {brunet2004dynamic}
\bibfield{author}{\bibinfo{person}{Olivier Brunet} {and}
  \bibinfo{person}{Philippe Jorrand}.} \bibinfo{year}{2004}\natexlab{}.
\newblock \showarticletitle{Dynamic quantum logic for quantum programs}.
\newblock \bibinfo{journal}{\emph{International Journal of Quantum
  Information}} \bibinfo{volume}{2}, \bibinfo{number}{01}
  (\bibinfo{year}{2004}), \bibinfo{pages}{45--54}.
\newblock


\bibitem[\protect\citeauthoryear{Chadha, Mateus, and Sernadas}{Chadha
  et~al\mbox{.}}{2006}]%
        {chadha2006reasoning}
\bibfield{author}{\bibinfo{person}{Rohit Chadha}, \bibinfo{person}{Paulo
  Mateus}, {and} \bibinfo{person}{Am{\'\i}lcar Sernadas}.}
  \bibinfo{year}{2006}\natexlab{}.
\newblock \showarticletitle{Reasoning about imperative quantum programs}.
\newblock \bibinfo{journal}{\emph{Electronic Notes in Theoretical Computer
  Science}}  \bibinfo{volume}{158} (\bibinfo{year}{2006}),
  \bibinfo{pages}{19--39}.
\newblock


\bibitem[\protect\citeauthoryear{Chakarov and Sankaranarayanan}{Chakarov and
  Sankaranarayanan}{2013}]%
        {chakarov2013probabilistic}
\bibfield{author}{\bibinfo{person}{Aleksandar Chakarov} {and}
  \bibinfo{person}{Sriram Sankaranarayanan}.} \bibinfo{year}{2013}\natexlab{}.
\newblock \showarticletitle{Probabilistic program analysis with martingales}.
  In \bibinfo{booktitle}{\emph{International Conference on Computer Aided
  Verification}}. Springer, \bibinfo{pages}{511--526}.
\newblock


\bibitem[\protect\citeauthoryear{Chatterjee, Fu, Novotn{\`y}, and
  Hasheminezhad}{Chatterjee et~al\mbox{.}}{2016}]%
        {chatterjee2016algorithmic}
\bibfield{author}{\bibinfo{person}{Krishnendu Chatterjee},
  \bibinfo{person}{Hongfei Fu}, \bibinfo{person}{Petr Novotn{\`y}}, {and}
  \bibinfo{person}{Rouzbeh Hasheminezhad}.} \bibinfo{year}{2016}\natexlab{}.
\newblock \showarticletitle{Algorithmic analysis of qualitative and
  quantitative termination problems for affine probabilistic programs}. In
  \bibinfo{booktitle}{\emph{Proceedings of the 43rd Annual ACM SIGPLAN-SIGACT
  Symposium on Principles of Programming Languages}}.
  \bibinfo{pages}{327--342}.
\newblock


\bibitem[\protect\citeauthoryear{Chiribella}{Chiribella}{2012}]%
        {chiribella2012perfect}
\bibfield{author}{\bibinfo{person}{Giulio Chiribella}.}
  \bibinfo{year}{2012}\natexlab{}.
\newblock \showarticletitle{Perfect discrimination of no-signalling channels
  via quantum superposition of causal structures}.
\newblock \bibinfo{journal}{\emph{Physical Review A}} \bibinfo{volume}{86},
  \bibinfo{number}{4} (\bibinfo{year}{2012}), \bibinfo{pages}{040301}.
\newblock


\bibitem[\protect\citeauthoryear{Chiribella, D¡¯Ariano, Perinotti, and
  Valiron}{Chiribella et~al\mbox{.}}{2013}]%
        {chiribella2013quantum}
\bibfield{author}{\bibinfo{person}{Giulio Chiribella},
  \bibinfo{person}{Giacomo~Mauro D¡¯Ariano}, \bibinfo{person}{Paolo Perinotti},
  {and} \bibinfo{person}{Benoit Valiron}.} \bibinfo{year}{2013}\natexlab{}.
\newblock \showarticletitle{Quantum computations without definite causal
  structure}.
\newblock \bibinfo{journal}{\emph{Physical Review A}} \bibinfo{volume}{88},
  \bibinfo{number}{2} (\bibinfo{year}{2013}), \bibinfo{pages}{022318}.
\newblock


\bibitem[\protect\citeauthoryear{Cook}{Cook}{1978}]%
        {Cook78}
\bibfield{author}{\bibinfo{person}{Stephen~A. Cook}.}
  \bibinfo{year}{1978}\natexlab{}.
\newblock \showarticletitle{Soundness and Completeness of an Axiom System for
  Program Verification}.
\newblock \bibinfo{journal}{\emph{{SIAM} J. Comput.}} \bibinfo{volume}{7},
  \bibinfo{number}{1} (\bibinfo{year}{1978}), \bibinfo{pages}{70--90}.
\newblock
\urldef\tempurl%
\url{https://doi.org/10.1137/0207005}
\showDOI{\tempurl}


\bibitem[\protect\citeauthoryear{Dale, Jennings, and Rudolph}{Dale
  et~al\mbox{.}}{2015}]%
        {dale2015provable}
\bibfield{author}{\bibinfo{person}{Howard Dale}, \bibinfo{person}{David
  Jennings}, {and} \bibinfo{person}{Terry Rudolph}.}
  \bibinfo{year}{2015}\natexlab{}.
\newblock \showarticletitle{Provable quantum advantage in randomness
  processing}.
\newblock \bibinfo{journal}{\emph{Nature communications}} \bibinfo{volume}{6},
  \bibinfo{number}{1} (\bibinfo{year}{2015}), \bibinfo{pages}{1--4}.
\newblock


\bibitem[\protect\citeauthoryear{D'Hondt and Panangaden}{D'Hondt and
  Panangaden}{2006}]%
        {panangaden2006}
\bibfield{author}{\bibinfo{person}{Ellie D'Hondt} {and}
  \bibinfo{person}{Prakash Panangaden}.} \bibinfo{year}{2006}\natexlab{}.
\newblock \showarticletitle{Quantum weakest preconditions}.
\newblock \bibinfo{journal}{\emph{Mathematical Structures in Computer Science}}
  \bibinfo{volume}{16}, \bibinfo{number}{3} (\bibinfo{year}{2006}),
  \bibinfo{pages}{429--451}.
\newblock
\urldef\tempurl%
\url{https://doi.org/10.1017/S0960129506005251}
\showDOI{\tempurl}


\bibitem[\protect\citeauthoryear{Etessami and Yannakakis}{Etessami and
  Yannakakis}{2009}]%
        {etessami2009recursive}
\bibfield{author}{\bibinfo{person}{Kousha Etessami} {and}
  \bibinfo{person}{Mihalis Yannakakis}.} \bibinfo{year}{2009}\natexlab{}.
\newblock \showarticletitle{Recursive Markov chains, stochastic grammars, and
  monotone systems of nonlinear equations}.
\newblock \bibinfo{journal}{\emph{Journal of the ACM (JACM)}}
  \bibinfo{volume}{56}, \bibinfo{number}{1} (\bibinfo{year}{2009}),
  \bibinfo{pages}{1--66}.
\newblock


\bibitem[\protect\citeauthoryear{Feng, Duan, Ji, and Ying}{Feng
  et~al\mbox{.}}{2007}]%
        {feng2007proof}
\bibfield{author}{\bibinfo{person}{Yuan Feng}, \bibinfo{person}{Runyao Duan},
  \bibinfo{person}{Zhengfeng Ji}, {and} \bibinfo{person}{Mingsheng Ying}.}
  \bibinfo{year}{2007}\natexlab{}.
\newblock \showarticletitle{Proof rules for the correctness of quantum
  programs}.
\newblock \bibinfo{journal}{\emph{Theoretical Computer Science}}
  \bibinfo{volume}{386}, \bibinfo{number}{1-2} (\bibinfo{year}{2007}),
  \bibinfo{pages}{151--166}.
\newblock


\bibitem[\protect\citeauthoryear{Feng, Hahn, Turrini, and Zhang}{Feng
  et~al\mbox{.}}{2015}]%
        {feng2015qpmc}
\bibfield{author}{\bibinfo{person}{Yuan Feng}, \bibinfo{person}{Ernst~Moritz
  Hahn}, \bibinfo{person}{Andrea Turrini}, {and} \bibinfo{person}{Lijun
  Zhang}.} \bibinfo{year}{2015}\natexlab{}.
\newblock \showarticletitle{QPMC: A model checker for quantum programs and
  protocols}. In \bibinfo{booktitle}{\emph{International Symposium on Formal
  Methods}}. Springer, \bibinfo{pages}{265--272}.
\newblock


\bibitem[\protect\citeauthoryear{Feng, Yu, and Ying}{Feng
  et~al\mbox{.}}{2013a}]%
        {feng2013model}
\bibfield{author}{\bibinfo{person}{Yuan Feng}, \bibinfo{person}{Nengkun Yu},
  {and} \bibinfo{person}{Mingsheng Ying}.} \bibinfo{year}{2013}\natexlab{a}.
\newblock \showarticletitle{Model checking quantum Markov chains}.
\newblock \bibinfo{journal}{\emph{J. Comput. System Sci.}}
  \bibinfo{volume}{79}, \bibinfo{number}{7} (\bibinfo{year}{2013}),
  \bibinfo{pages}{1181--1198}.
\newblock


\bibitem[\protect\citeauthoryear{Feng, Yu, and Ying}{Feng
  et~al\mbox{.}}{2013b}]%
        {feng2013reachability}
\bibfield{author}{\bibinfo{person}{Yuan Feng}, \bibinfo{person}{Nengkun Yu},
  {and} \bibinfo{person}{Mingsheng Ying}.} \bibinfo{year}{2013}\natexlab{b}.
\newblock \showarticletitle{Reachability analysis of recursive quantum Markov
  chains}. In \bibinfo{booktitle}{\emph{International Symposium on Mathematical
  Foundations of Computer Science}}. Springer, \bibinfo{pages}{385--396}.
\newblock


\bibitem[\protect\citeauthoryear{Floyd}{Floyd}{1967}]%
        {floyd1967assigning}
\bibfield{author}{\bibinfo{person}{Robert~W Floyd}.}
  \bibinfo{year}{1967}\natexlab{}.
\newblock \showarticletitle{Assigning meanings to programs}.
\newblock \bibinfo{journal}{\emph{Mathematical aspects of computer science}}
  \bibinfo{volume}{19}, \bibinfo{number}{19-32} (\bibinfo{year}{1967}),
  \bibinfo{pages}{1}.
\newblock


\bibitem[\protect\citeauthoryear{Francez}{Francez}{1992}]%
        {Francez}
\bibfield{author}{\bibinfo{person}{Nissim Francez}.}
  \bibinfo{year}{1992}\natexlab{}.
\newblock \bibinfo{booktitle}{\emph{Program verification}}.
\newblock \bibinfo{publisher}{Addison-Wesley}.
\newblock
\showISBNx{978-0-201-41608-4}


\bibitem[\protect\citeauthoryear{Gay}{Gay}{2006}]%
        {gay2006quantum}
\bibfield{author}{\bibinfo{person}{Simon~J Gay}.}
  \bibinfo{year}{2006}\natexlab{}.
\newblock \showarticletitle{Quantum programming languages: Survey and
  bibliography}.
\newblock \bibinfo{journal}{\emph{Mathematical Structures in Computer Science}}
  \bibinfo{volume}{16}, \bibinfo{number}{4} (\bibinfo{year}{2006}),
  \bibinfo{pages}{581--600}.
\newblock


\bibitem[\protect\citeauthoryear{Gay, Nagarajan, and Papanikolaou}{Gay
  et~al\mbox{.}}{2008}]%
        {gay2008qmc}
\bibfield{author}{\bibinfo{person}{Simon~J Gay}, \bibinfo{person}{Rajagopal
  Nagarajan}, {and} \bibinfo{person}{Nikolaos Papanikolaou}.}
  \bibinfo{year}{2008}\natexlab{}.
\newblock \showarticletitle{QMC: A model checker for quantum systems}. In
  \bibinfo{booktitle}{\emph{International Conference on Computer Aided
  Verification}}. Springer, \bibinfo{pages}{543--547}.
\newblock


\bibitem[\protect\citeauthoryear{Girard}{Girard}{1987}]%
        {girard1987linear}
\bibfield{author}{\bibinfo{person}{Jean-Yves Girard}.}
  \bibinfo{year}{1987}\natexlab{}.
\newblock \showarticletitle{Linear logic}.
\newblock \bibinfo{journal}{\emph{Theoretical computer science}}
  \bibinfo{volume}{50}, \bibinfo{number}{1} (\bibinfo{year}{1987}),
  \bibinfo{pages}{1--101}.
\newblock


\bibitem[\protect\citeauthoryear{Green, Lumsdaine, Ross, Selinger, and
  Valiron}{Green et~al\mbox{.}}{2013}]%
        {green2013quipper}
\bibfield{author}{\bibinfo{person}{Alexander~S Green},
  \bibinfo{person}{Peter~LeFanu Lumsdaine}, \bibinfo{person}{Neil~J Ross},
  \bibinfo{person}{Peter Selinger}, {and} \bibinfo{person}{Beno{\^\i}t
  Valiron}.} \bibinfo{year}{2013}\natexlab{}.
\newblock \showarticletitle{Quipper: a scalable quantum programming language}.
  In \bibinfo{booktitle}{\emph{ACM SIGPLAN Notices}},
  Vol.~\bibinfo{volume}{48}. ACM, \bibinfo{pages}{333--342}.
\newblock


\bibitem[\protect\citeauthoryear{Grover}{Grover}{1996}]%
        {grover1996fast}
\bibfield{author}{\bibinfo{person}{Lov~K Grover}.}
  \bibinfo{year}{1996}\natexlab{}.
\newblock \showarticletitle{A fast quantum mechanical algorithm for database
  search}.
\newblock \bibinfo{journal}{\emph{arXiv preprint quant-ph/9605043}}
  (\bibinfo{year}{1996}).
\newblock


\bibitem[\protect\citeauthoryear{Grover}{Grover}{2005}]%
        {grover2005fixed}
\bibfield{author}{\bibinfo{person}{Lov~K Grover}.}
  \bibinfo{year}{2005}\natexlab{}.
\newblock \showarticletitle{Fixed-point quantum search}.
\newblock \bibinfo{journal}{\emph{Physical Review Letters}}
  \bibinfo{volume}{95}, \bibinfo{number}{15} (\bibinfo{year}{2005}),
  \bibinfo{pages}{150501}.
\newblock


\bibitem[\protect\citeauthoryear{Hoare}{Hoare}{1969}]%
        {hoare1969axiomatic}
\bibfield{author}{\bibinfo{person}{Charles Antony~Richard Hoare}.}
  \bibinfo{year}{1969}\natexlab{}.
\newblock \showarticletitle{An axiomatic basis for computer programming}.
\newblock \bibinfo{journal}{\emph{Commun. ACM}} \bibinfo{volume}{12},
  \bibinfo{number}{10} (\bibinfo{year}{1969}), \bibinfo{pages}{576--580}.
\newblock


\bibitem[\protect\citeauthoryear{Hoare}{Hoare}{1971}]%
        {Hoare1971}
\bibfield{author}{\bibinfo{person}{C.~A.~R. Hoare}.}
  \bibinfo{year}{1971}\natexlab{}.
\newblock \showarticletitle{Procedures and parameters: An axiomatic approach}.
\newblock In \bibinfo{booktitle}{\emph{Symposium on Semantics of Algorithmic
  Languages}}. \bibinfo{pages}{102--116}.
\newblock
\urldef\tempurl%
\url{https://doi.org/10.1007/BFb0059696}
\showDOI{\tempurl}


\bibitem[\protect\citeauthoryear{Hung, Hietala, Zhu, Ying, Hicks, and Wu}{Hung
  et~al\mbox{.}}{2019}]%
        {hung2019quantitative}
\bibfield{author}{\bibinfo{person}{Shih-Han Hung}, \bibinfo{person}{Kesha
  Hietala}, \bibinfo{person}{Shaopeng Zhu}, \bibinfo{person}{Mingsheng Ying},
  \bibinfo{person}{Michael Hicks}, {and} \bibinfo{person}{Xiaodi Wu}.}
  \bibinfo{year}{2019}\natexlab{}.
\newblock \showarticletitle{Quantitative robustness analysis of quantum
  programs}.
\newblock \bibinfo{journal}{\emph{Proceedings of the ACM on Programming
  Languages}} \bibinfo{volume}{3}, \bibinfo{number}{POPL}
  (\bibinfo{year}{2019}), \bibinfo{pages}{1--29}.
\newblock


\bibitem[\protect\citeauthoryear{Johnston}{Johnston}{[n.d.]a}]%
        {johnstonequivalencesI}
\bibfield{author}{\bibinfo{person}{Nathaniel Johnston}.}
  \bibinfo{year}{[n.d.]}\natexlab{a}.
\newblock \showarticletitle{The Equivalences of the Choi-Jamiolkowski
  Isomorphism (Part I)}.
\newblock  (\bibinfo{year}{[n.\,d.]}).
\newblock


\bibitem[\protect\citeauthoryear{Johnston}{Johnston}{[n.d.]b}]%
        {johnstonequivalencesII}
\bibfield{author}{\bibinfo{person}{Nathaniel Johnston}.}
  \bibinfo{year}{[n.d.]}\natexlab{b}.
\newblock \showarticletitle{The Equivalences of the Choi-Jamiolkowski
  Isomorphism (Part II)}.
\newblock  (\bibinfo{year}{[n.\,d.]}).
\newblock


\bibitem[\protect\citeauthoryear{Kakutani}{Kakutani}{2009}]%
        {kakutani2009logic}
\bibfield{author}{\bibinfo{person}{Yoshihiko Kakutani}.}
  \bibinfo{year}{2009}\natexlab{}.
\newblock \showarticletitle{A logic for formal verification of quantum
  programs}. In \bibinfo{booktitle}{\emph{Annual Asian Computing Science
  Conference}}. Springer, \bibinfo{pages}{79--93}.
\newblock


\bibitem[\protect\citeauthoryear{Kaminski, Katoen, Matheja, and
  Olmedo}{Kaminski et~al\mbox{.}}{2018}]%
        {KaminskiKMO18}
\bibfield{author}{\bibinfo{person}{Benjamin~Lucien Kaminski},
  \bibinfo{person}{Joost{-}Pieter Katoen}, \bibinfo{person}{Christoph Matheja},
  {and} \bibinfo{person}{Federico Olmedo}.} \bibinfo{year}{2018}\natexlab{}.
\newblock \showarticletitle{Weakest Precondition Reasoning for Expected
  Runtimes of Randomized Algorithms}.
\newblock \bibinfo{journal}{\emph{J. {ACM}}} \bibinfo{volume}{65},
  \bibinfo{number}{5} (\bibinfo{year}{2018}), \bibinfo{pages}{30:1--30:68}.
\newblock
\urldef\tempurl%
\url{https://doi.org/10.1145/3208102}
\showDOI{\tempurl}


\bibitem[\protect\citeauthoryear{Kaminski, Katoen, and Matheja}{Kaminski
  et~al\mbox{.}}{2019}]%
        {kaminski2019hardness}
\bibfield{author}{\bibinfo{person}{Benjamin~Lucien Kaminski},
  \bibinfo{person}{Joost-Pieter Katoen}, {and} \bibinfo{person}{Christoph
  Matheja}.} \bibinfo{year}{2019}\natexlab{}.
\newblock \showarticletitle{On the hardness of analyzing probabilistic
  programs}.
\newblock \bibinfo{journal}{\emph{Acta Informatica}} \bibinfo{volume}{56},
  \bibinfo{number}{3} (\bibinfo{year}{2019}), \bibinfo{pages}{255--285}.
\newblock


\bibitem[\protect\citeauthoryear{Kobayashi, Lago, and Grellois}{Kobayashi
  et~al\mbox{.}}{2020}]%
        {kobayashi2020termination}
\bibfield{author}{\bibinfo{person}{Naoki Kobayashi}, \bibinfo{person}{Ugo~Dal
  Lago}, {and} \bibinfo{person}{Charles Grellois}.}
  \bibinfo{year}{2020}\natexlab{}.
\newblock \bibinfo{title}{On the Termination Problem for Probabilistic
  Higher-Order Recursive Programs}.
\newblock
\newblock
\showeprint[arxiv]{1811.02133}~[cs.PL]


\bibitem[\protect\citeauthoryear{Lago, Zuppiroli, and Gabbrielli}{Lago
  et~al\mbox{.}}{2014}]%
        {LagoZG14}
\bibfield{author}{\bibinfo{person}{Ugo~Dal Lago}, \bibinfo{person}{Sara
  Zuppiroli}, {and} \bibinfo{person}{Maurizio Gabbrielli}.}
  \bibinfo{year}{2014}\natexlab{}.
\newblock \showarticletitle{Probabilistic Recursion Theory and Implicit
  Computational Complexity}.
\newblock \bibinfo{journal}{\emph{Sci. Ann. Comp. Sci.}} \bibinfo{volume}{24},
  \bibinfo{number}{2} (\bibinfo{year}{2014}), \bibinfo{pages}{177--216}.
\newblock
\urldef\tempurl%
\url{https://doi.org/10.7561/SACS.2014.2.177}
\showDOI{\tempurl}


\bibitem[\protect\citeauthoryear{Li and Ying}{Li and Ying}{2018}]%
        {li2018termination}
\bibfield{author}{\bibinfo{person}{Yangjia Li} {and} \bibinfo{person}{Mingsheng
  Ying}.} \bibinfo{year}{2018}\natexlab{}.
\newblock \showarticletitle{Algorithmic analysis of termination problems for
  quantum programs}. In \bibinfo{booktitle}{\emph{ACM SIGPLAN Notices}},
  Vol.~\bibinfo{volume}{53}. ACM, \bibinfo{pages}{35:1--29}.
\newblock


\bibitem[\protect\citeauthoryear{Liu, Zhan, Wang, Ying, Liu, Li, Ying, and
  Zhan}{Liu et~al\mbox{.}}{2019}]%
        {liu2019formal}
\bibfield{author}{\bibinfo{person}{Junyi Liu}, \bibinfo{person}{Bohua Zhan},
  \bibinfo{person}{Shuling Wang}, \bibinfo{person}{Shenggang Ying},
  \bibinfo{person}{Tao Liu}, \bibinfo{person}{Yangjia Li},
  \bibinfo{person}{Mingsheng Ying}, {and} \bibinfo{person}{Naijun Zhan}.}
  \bibinfo{year}{2019}\natexlab{}.
\newblock \showarticletitle{Formal verification of quantum algorithms using
  quantum Hoare logic}. In \bibinfo{booktitle}{\emph{International conference
  on computer aided verification}}. Springer, \bibinfo{pages}{187--207}.
\newblock


\bibitem[\protect\citeauthoryear{McIver, Morgan, Kaminski, and Katoen}{McIver
  et~al\mbox{.}}{2017}]%
        {mciver2017new}
\bibfield{author}{\bibinfo{person}{Annabelle McIver}, \bibinfo{person}{Carroll
  Morgan}, \bibinfo{person}{Benjamin~Lucien Kaminski}, {and}
  \bibinfo{person}{Joost-Pieter Katoen}.} \bibinfo{year}{2017}\natexlab{}.
\newblock \showarticletitle{A new proof rule for almost-sure termination}.
\newblock \bibinfo{journal}{\emph{Proceedings of the ACM on Programming
  Languages}} \bibinfo{volume}{2}, \bibinfo{number}{POPL}
  (\bibinfo{year}{2017}), \bibinfo{pages}{1--28}.
\newblock


\bibitem[\protect\citeauthoryear{McKague}{McKague}{2012}]%
        {mckague2012interactive}
\bibfield{author}{\bibinfo{person}{Matthew McKague}.}
  \bibinfo{year}{2012}\natexlab{}.
\newblock \showarticletitle{Interactive proofs with efficient quantum prover
  for recursive Fourier sampling}.
\newblock \bibinfo{journal}{\emph{Chicago J. Theor. Comput. Sci}}
  \bibinfo{volume}{6} (\bibinfo{year}{2012}), \bibinfo{pages}{1--10}.
\newblock


\bibitem[\protect\citeauthoryear{Nielsen and Chuang}{Nielsen and
  Chuang}{2000}]%
        {nielsen2000quantum}
\bibfield{author}{\bibinfo{person}{Michael~A Nielsen} {and}
  \bibinfo{person}{Isaac~L Chuang}.} \bibinfo{year}{2000}\natexlab{}.
\newblock \bibinfo{title}{Quantum computation and quantum information}.
\newblock
\newblock


\bibitem[\protect\citeauthoryear{Olmedo, Kaminski, Katoen, and Matheja}{Olmedo
  et~al\mbox{.}}{2016}]%
        {OlmedoKKM16}
\bibfield{author}{\bibinfo{person}{Federico Olmedo},
  \bibinfo{person}{Benjamin~Lucien Kaminski}, \bibinfo{person}{Joost{-}Pieter
  Katoen}, {and} \bibinfo{person}{Christoph Matheja}.}
  \bibinfo{year}{2016}\natexlab{}.
\newblock \showarticletitle{Reasoning about Recursive Probabilistic Programs}.
  In \bibinfo{booktitle}{\emph{Proceedings of the 31st Annual {ACM/IEEE}
  Symposium on Logic in Computer Science, {LICS} '16, New York, NY, USA, July
  5-8, 2016}}. \bibinfo{pages}{672--681}.
\newblock
\urldef\tempurl%
\url{https://doi.org/10.1145/2933575.2935317}
\showDOI{\tempurl}


\bibitem[\protect\citeauthoryear{{\"O}mer}{{\"O}mer}{2003}]%
        {omer2003structured}
\bibfield{author}{\bibinfo{person}{Bernhard {\"O}mer}.}
  \bibinfo{year}{2003}\natexlab{}.
\newblock \bibinfo{booktitle}{\emph{Structured quantum programming}}.
\newblock \bibinfo{publisher}{na}.
\newblock
\urldef\tempurl%
\url{http://www.itp.tuwien.ac.at/~oemer/doc/structquprog.pdf}
\showURL{%
\tempurl}


\bibitem[\protect\citeauthoryear{Paykin, Rand, and Zdancewic}{Paykin
  et~al\mbox{.}}{2017}]%
        {paykin2017qwire}
\bibfield{author}{\bibinfo{person}{Jennifer Paykin}, \bibinfo{person}{Robert
  Rand}, {and} \bibinfo{person}{Steve Zdancewic}.}
  \bibinfo{year}{2017}\natexlab{}.
\newblock \showarticletitle{QWIRE: a core language for quantum circuits}. In
  \bibinfo{booktitle}{\emph{ACM SIGPLAN Notices}}, Vol.~\bibinfo{volume}{52}.
  ACM, \bibinfo{pages}{846--858}.
\newblock


\bibitem[\protect\citeauthoryear{Prugovecki}{Prugovecki}{1982}]%
        {prugovecki1982quantum}
\bibfield{author}{\bibinfo{person}{Eduard Prugovecki}.}
  \bibinfo{year}{1982}\natexlab{}.
\newblock \bibinfo{booktitle}{\emph{Quantum mechanics in Hilbert space}}.
\newblock \bibinfo{publisher}{Academic Press}.
\newblock


\bibitem[\protect\citeauthoryear{Rand}{Rand}{2019}]%
        {rand2019verification}
\bibfield{author}{\bibinfo{person}{Robert Rand}.}
  \bibinfo{year}{2019}\natexlab{}.
\newblock \showarticletitle{Verification logics for quantum programs}.
\newblock \bibinfo{journal}{\emph{arXiv preprint arXiv:1904.04304}}
  (\bibinfo{year}{2019}).
\newblock


\bibitem[\protect\citeauthoryear{Reynolds}{Reynolds}{2002}]%
        {reynolds2002separation}
\bibfield{author}{\bibinfo{person}{John~C Reynolds}.}
  \bibinfo{year}{2002}\natexlab{}.
\newblock \showarticletitle{Separation logic: A logic for shared mutable data
  structures}. In \bibinfo{booktitle}{\emph{Proceedings 17th Annual IEEE
  Symposium on Logic in Computer Science}}. IEEE, \bibinfo{pages}{55--74}.
\newblock


\bibitem[\protect\citeauthoryear{Sabry}{Sabry}{2003}]%
        {sabry2003modeling}
\bibfield{author}{\bibinfo{person}{Amr Sabry}.}
  \bibinfo{year}{2003}\natexlab{}.
\newblock \showarticletitle{Modeling quantum computing in Haskell}. In
  \bibinfo{booktitle}{\emph{Proceedings of the 2003 ACM SIGPLAN workshop on
  Haskell}}. ACM, \bibinfo{pages}{39--49}.
\newblock


\bibitem[\protect\citeauthoryear{Sabry, Valiron, and Vizzotto}{Sabry
  et~al\mbox{.}}{2018}]%
        {SabryVV18}
\bibfield{author}{\bibinfo{person}{Amr Sabry}, \bibinfo{person}{Beno{\^{\i}}t
  Valiron}, {and} \bibinfo{person}{Juliana~Kaizer Vizzotto}.}
  \bibinfo{year}{2018}\natexlab{}.
\newblock \showarticletitle{From Symmetric Pattern-Matching to Quantum
  Control}. In \bibinfo{booktitle}{\emph{Foundations of Software Science and
  Computation Structures - 21st International Conference, {FOSSACS} 2018, Held
  as Part of the European Joint Conferences on Theory and Practice of Software,
  {ETAPS} 2018, Thessaloniki, Greece, April 14-20, 2018, Proceedings}}.
  \bibinfo{pages}{348--364}.
\newblock
\urldef\tempurl%
\url{https://doi.org/10.1007/978-3-319-89366-2\_19}
\showDOI{\tempurl}


\bibitem[\protect\citeauthoryear{Sanders and Zuliani}{Sanders and
  Zuliani}{2000}]%
        {sanders2000quantum}
\bibfield{author}{\bibinfo{person}{Jeff~W Sanders} {and} \bibinfo{person}{Paolo
  Zuliani}.} \bibinfo{year}{2000}\natexlab{}.
\newblock \showarticletitle{Quantum programming}. In
  \bibinfo{booktitle}{\emph{International Conference on Mathematics of Program
  Construction}}. Springer, \bibinfo{pages}{80--99}.
\newblock


\bibitem[\protect\citeauthoryear{Selinger}{Selinger}{2004a}]%
        {Selinger2004}
\bibfield{author}{\bibinfo{person}{Peter Selinger}.}
  \bibinfo{year}{2004}\natexlab{a}.
\newblock \showarticletitle{Towards a quantum programming language}.
\newblock \bibinfo{journal}{\emph{Mathematical Structures in Computer Science}}
  \bibinfo{volume}{14}, \bibinfo{number}{4} (\bibinfo{year}{2004}),
  \bibinfo{pages}{527--586}.
\newblock
\urldef\tempurl%
\url{https://doi.org/10.1017/S0960129504004256}
\showDOI{\tempurl}


\bibitem[\protect\citeauthoryear{Selinger}{Selinger}{2004b}]%
        {selinger2004towards}
\bibfield{author}{\bibinfo{person}{Peter Selinger}.}
  \bibinfo{year}{2004}\natexlab{b}.
\newblock \showarticletitle{Towards a quantum programming language}.
\newblock \bibinfo{journal}{\emph{Mathematical Structures in Computer Science}}
  \bibinfo{volume}{14}, \bibinfo{number}{4} (\bibinfo{year}{2004}),
  \bibinfo{pages}{527--586}.
\newblock


\bibitem[\protect\citeauthoryear{Shor}{Shor}{1994}]%
        {shor1994algorithms}
\bibfield{author}{\bibinfo{person}{Peter~W Shor}.}
  \bibinfo{year}{1994}\natexlab{}.
\newblock \showarticletitle{Algorithms for quantum computation: Discrete
  logarithms and factoring}. In \bibinfo{booktitle}{\emph{Proceedings 35th
  annual symposium on foundations of computer science}}. Ieee,
  \bibinfo{pages}{124--134}.
\newblock


\bibitem[\protect\citeauthoryear{Singhal}{Singhal}{2020}]%
        {singhal2020quantum}
\bibfield{author}{\bibinfo{person}{Kartik Singhal}.}
  \bibinfo{year}{2020}\natexlab{}.
\newblock \showarticletitle{Quantum Hoare Type Theory}.
\newblock \bibinfo{journal}{\emph{arXiv preprint arXiv:2012.02154}}
  (\bibinfo{year}{2020}).
\newblock


\bibitem[\protect\citeauthoryear{Svore, Geller, Troyer, Azariah, Granade, Heim,
  Kliuchnikov, Mykhailova, Paz, and Roetteler}{Svore et~al\mbox{.}}{2018}]%
        {svore2018q}
\bibfield{author}{\bibinfo{person}{Krysta Svore}, \bibinfo{person}{Alan
  Geller}, \bibinfo{person}{Matthias Troyer}, \bibinfo{person}{John Azariah},
  \bibinfo{person}{Christopher Granade}, \bibinfo{person}{Bettina Heim},
  \bibinfo{person}{Vadym Kliuchnikov}, \bibinfo{person}{Mariia Mykhailova},
  \bibinfo{person}{Andres Paz}, {and} \bibinfo{person}{Martin Roetteler}.}
  \bibinfo{year}{2018}\natexlab{}.
\newblock \showarticletitle{Q\#: Enabling scalable quantum computing and
  development with a high-level dsl}. In \bibinfo{booktitle}{\emph{Proceedings
  of the Real World Domain Specific Languages Workshop 2018}}. ACM,
  \bibinfo{pages}{7}.
\newblock


\bibitem[\protect\citeauthoryear{Unruh}{Unruh}{2019a}]%
        {Unruh19b}
\bibfield{author}{\bibinfo{person}{Dominique Unruh}.}
  \bibinfo{year}{2019}\natexlab{a}.
\newblock \showarticletitle{Quantum Hoare Logic with Ghost Variables}.
\newblock \bibinfo{journal}{\emph{CoRR}}  \bibinfo{volume}{abs/1902.00325}
  (\bibinfo{year}{2019}).
\newblock
\showeprint[arxiv]{1902.00325}
\urldef\tempurl%
\url{http://arxiv.org/abs/1902.00325}
\showURL{%
\tempurl}


\bibitem[\protect\citeauthoryear{Unruh}{Unruh}{2019b}]%
        {Unruh19a}
\bibfield{author}{\bibinfo{person}{Dominique Unruh}.}
  \bibinfo{year}{2019}\natexlab{b}.
\newblock \showarticletitle{Quantum relational Hoare logic}.
\newblock \bibinfo{journal}{\emph{{PACMPL}}} \bibinfo{volume}{3},
  \bibinfo{number}{{POPL}} (\bibinfo{year}{2019}),
  \bibinfo{pages}{33:1--33:31}.
\newblock
\urldef\tempurl%
\url{https://doi.org/10.1145/3290346}
\showDOI{\tempurl}


\bibitem[\protect\citeauthoryear{Wecker and Svore}{Wecker and Svore}{2014}]%
        {wecker2014liqui}
\bibfield{author}{\bibinfo{person}{Dave Wecker} {and} \bibinfo{person}{Krysta~M
  Svore}.} \bibinfo{year}{2014}\natexlab{}.
\newblock \showarticletitle{LIQUi$|\rangle$;: A software design architecture
  and domain-specific language for quantum computing}.
\newblock \bibinfo{journal}{\emph{arXiv preprint arXiv:1402.4467}}
  (\bibinfo{year}{2014}).
\newblock


\bibitem[\protect\citeauthoryear{Wolf}{Wolf}{2012}]%
        {wolf2012quantum}
\bibfield{author}{\bibinfo{person}{Michael~M Wolf}.}
  \bibinfo{year}{2012}\natexlab{}.
\newblock \showarticletitle{Quantum channels \& operations: Guided tour}.
\newblock \bibinfo{journal}{\emph{Lecture notes available at http://www-m5. ma.
  tum. de/foswiki/pub M}}  \bibinfo{volume}{5} (\bibinfo{year}{2012}).
\newblock


\bibitem[\protect\citeauthoryear{Ying}{Ying}{2011}]%
        {Ying11}
\bibfield{author}{\bibinfo{person}{Mingsheng Ying}.}
  \bibinfo{year}{2011}\natexlab{}.
\newblock \showarticletitle{Floyd-Hoare logic for quantum programs}.
\newblock \bibinfo{journal}{\emph{{ACM} Trans. Program. Lang. Syst.}}
  \bibinfo{volume}{33}, \bibinfo{number}{6} (\bibinfo{year}{2011}),
  \bibinfo{pages}{19:1--19:49}.
\newblock
\urldef\tempurl%
\url{https://doi.org/10.1145/2049706.2049708}
\showDOI{\tempurl}


\bibitem[\protect\citeauthoryear{Ying}{Ying}{2016}]%
        {ying2016foundations}
\bibfield{author}{\bibinfo{person}{Mingsheng Ying}.}
  \bibinfo{year}{2016}\natexlab{}.
\newblock \bibinfo{booktitle}{\emph{Foundations of Quantum Programming}}.
\newblock \bibinfo{publisher}{Morgan Kaufmann}.
\newblock


\bibitem[\protect\citeauthoryear{Ying, Li, Yu, and Feng}{Ying
  et~al\mbox{.}}{2014}]%
        {ying2014model}
\bibfield{author}{\bibinfo{person}{Mingsheng Ying}, \bibinfo{person}{Yangjia
  Li}, \bibinfo{person}{Nengkun Yu}, {and} \bibinfo{person}{Yuan Feng}.}
  \bibinfo{year}{2014}\natexlab{}.
\newblock \showarticletitle{Model-checking linear-time properties of quantum
  systems}.
\newblock \bibinfo{journal}{\emph{ACM Transactions on Computational Logic
  (TOCL)}} \bibinfo{volume}{15}, \bibinfo{number}{3} (\bibinfo{year}{2014}),
  \bibinfo{pages}{1--31}.
\newblock


\bibitem[\protect\citeauthoryear{Ying, Ying, and Wu}{Ying
  et~al\mbox{.}}{2017}]%
        {ying2017invariants}
\bibfield{author}{\bibinfo{person}{Mingsheng Ying}, \bibinfo{person}{Shenggang
  Ying}, {and} \bibinfo{person}{Xiaodi Wu}.} \bibinfo{year}{2017}\natexlab{}.
\newblock \showarticletitle{Invariants of quantum programs: characterisations
  and generation}. In \bibinfo{booktitle}{\emph{ACM SIGPLAN Notices}},
  Vol.~\bibinfo{volume}{52}. ACM, \bibinfo{pages}{818--832}.
\newblock


\bibitem[\protect\citeauthoryear{Yoder, Low, and Chuang}{Yoder
  et~al\mbox{.}}{2014}]%
        {yoder2014fixed}
\bibfield{author}{\bibinfo{person}{Theodore~J Yoder},
  \bibinfo{person}{Guang~Hao Low}, {and} \bibinfo{person}{Isaac~L Chuang}.}
  \bibinfo{year}{2014}\natexlab{}.
\newblock \showarticletitle{Fixed-point quantum search with an optimal number
  of queries}.
\newblock \bibinfo{journal}{\emph{Physical review letters}}
  \bibinfo{volume}{113}, \bibinfo{number}{21} (\bibinfo{year}{2014}),
  \bibinfo{pages}{210501}.
\newblock


\bibitem[\protect\citeauthoryear{Zhou, Yu, and Ying}{Zhou
  et~al\mbox{.}}{2019}]%
        {zhou2019applied}
\bibfield{author}{\bibinfo{person}{Li Zhou}, \bibinfo{person}{Nengkun Yu},
  {and} \bibinfo{person}{Mingsheng Ying}.} \bibinfo{year}{2019}\natexlab{}.
\newblock \showarticletitle{An applied quantum Hoare logic}. In
  \bibinfo{booktitle}{\emph{Proceedings of the 40th ACM SIGPLAN Conference on
  Programming Language Design and Implementation}}.
  \bibinfo{pages}{1149--1162}.
\newblock


\end{thebibliography}

\appendix

\section{Parameterized recursive quantum programs}\label{app_sec_prqp}

This section is devoted to investigating the formal semantics of parameterized recursive quantum programs $\EqRP$.
Formally, $\EqRP$ can be defined by the following grammar:
\begin{displaymath}
\begin{array}{rcll}
  P & \triangleq & D :: S & \mbox{Quantum program} \\
  D & \triangleq & \recDec{\langle \proc_1 \rangle(\lst{y}_1)}{S_1}, \ldots, \recDec{\langle \proc_n \rangle(\lst{y}_n)}{S_n} & \mbox{Procedure declaration} \\
  S & \triangleq & \bottom & \mbox{Bottom} \\
    & \mid & \SKIP & \mbox{No operation} \\
    & \mid & q \assnequal \ket{0} & \mbox{Initialization} \\
    & \mid & \lst{q} \starequal U & \mbox{Unitary operation} \\
    & \mid & S_1;S_2 & \mbox{Sequential composition} \\
    & \mid & \ifStat{\Box m\cdot M[\lst{q}] = m \rightarrow S_m} & \mbox{Probabilistic branching} \\
    & \mid & \LOCAL{\lst{q}}; S; \RELEASE{\lst{q}} & \mbox{Variable localization} \\
    & \mid & \CALL\ \langle \proc_i \rangle(\lst{p}_i),\quad 1\leq i \leq n & \mbox{Procedure call}
\end{array}
\end{displaymath}

\subsection{Nondeterministic operational semantics}

\begin{table}[!htbp]

  \centering

  \begin{tabular}{rlcrl}
  (Bot) & $\langle \bottom,\rho\rangle \xrightarrow{\epsilon} \langle E, 0\rangle$  &  &
  (Skip) & $\langle \SKIP,\rho \rangle \xrightarrow{\epsilon} \langle E, \rho \rangle$ \\
  \specialrule{0em}{2pt}{2pt}
  (Init) &
  $\dfrac{\sum_{i} \voutprod{i}{q}{i} = I_{q}}{\langle q \assnequal \ket{0}, \rho \rangle \xrightarrow{\epsilon} \langle E, \sum_i \ket{0}_q \bra{i}\rho \ket{i}_q \bra{0} \rangle}$
  &  & (Unit) &
  $\dfrac{U U^\dag = U^\dag U = I_{\lst{q}}}{\langle \lst{q} \starequal U, \rho \rangle \xrightarrow{\epsilon} \langle E, U\rho U^\dag \rangle}$
  \\
  \specialrule{0em}{2pt}{2pt}
  (Comp${}_1$) & $\dfrac{\langle S_1,\rho\rangle \xrightarrow{l} \langle S_1',\rho'\rangle \mbox{ and } S_1' \neq E}{\langle S_1;S_2,\rho\rangle \xrightarrow{l} \langle S_1';S_2,\rho'\rangle}$ & &
  (Comp${}_2$) & $\dfrac{\langle S_1,\rho\rangle \xrightarrow{l} \langle E,\rho'\rangle}{\langle S_1;S_2,\rho\rangle \xrightarrow{l} \langle S_2,\rho'\rangle}$ \\
  \specialrule{0em}{2pt}{2pt}
  (Case) & \multicolumn{4}{l}{$\dfrac{M = \{M_m\}_m \mbox{ and } \sum_{m} M_m^\dag M_m = I_{\lst{q}}}{\langle \ifStat{\Box m\cdot M[\lst{q}] = m \rightarrow S_m},\rho\rangle \xrightarrow{m} \langle S_m,M_m\rho M_m^\dag \rangle}$} \\
  \specialrule{0em}{2pt}{2pt}
  (Loc) & \multicolumn{4}{l}{$\dfrac{\lst{r}\not\in \Var(\rho),\ |\lst{r}| = |\lst{q}|,\ \type(r_i) = \type(q_i)\ \forall\ i.}{\langle \LOCAL{\lst{q}}; S; \RELEASE{\lst{q}}, \rho \rangle \xrightarrow{\epsilon} \langle S[\lst{r}/\lst{q}]; \RELEASE{\lst{r}}, \rho \otimes \ket{0}_{\lst{r}}\bra{0} \rangle}$} \\
  \specialrule{0em}{2pt}{2pt}
  (Rel) & \multicolumn{4}{l}{$\dfrac{tr_{\lst{r}} \triangleq \sum_{i} \bra{i} \diamond \ket{i} \mbox{ with } \sum_{i} \outprod{i}{i} = I_{\lst{r}}}{\langle \RELEASE{\lst{r}}, \rho \rangle \xrightarrow{\epsilon} \langle E , tr_{\lst{r}}(\rho) \rangle}$} \\
  \specialrule{0em}{2pt}{2pt}
  (Proc) & $\dfrac{\recDec{\proc(\lst{y})}{S} \in D}{\langle\CALL\ \proc(\lst{p}),\rho\rangle \xrightarrow{\epsilon} \langle S[\lst{p}/\lst{y}], \rho\rangle}$ & &
  (Except) & $\dfrac{\mbox{Other cases of $S$ and $l$}}{\langle S,\rho\rangle \xrightarrow{l} \bot}$
  \end{tabular}
  \caption{Labeled transition relation $\xrightarrow{l}$ with $l \triangleq \epsilon \mid m$.}
  \label{tab_EopSem}
\end{table}

The transition relation $\xrightarrow{l}$ for $\EqRP$ is defined in Tab. \ref{tab_EopSem}.
The multi-step labeled transition relation $\xrightarrow{\alpha}$
with $\alpha \triangleq l \mid \alpha_1 \alpha_2$ can be defined on $\xrightarrow{l}$ as before.

\subsection{$\QOP$-directed denotational semantics}

\begin{table}[!htbp]

  \centering

  \begin{tabular}{rlrl}
    (Bot) & $\sem{\bottom} = 0 \diamond 0$ & (Skip) & $\sem{\SKIP} = I \diamond I$ \\
    \specialrule{0em}{2pt}{2pt}
    (Init) & $\sem{q \assnequal \ket{0}} = \sum_i \ket{0}_q \bra{i} \diamond \ket{i}_q \bra{0}$ & (Unit) & $\sem{\lst{q} \starequal U} = U \diamond U^\dag$ \\
    \specialrule{0em}{2pt}{2pt}
    (Comp) & $\sem{S_1;S_2} = \sem{S_2} \circ \sem{S_1}$ & (Case) & $\sem{\IF} = \sum_m \sem{S_m} \circ (M_m \diamond M_m^\dag)$
    \\
    \specialrule{0em}{2pt}{2pt}
    (Loc) & \multicolumn{3}{l}{$\sem{\LOCAL{\lst{q}}; S; \RELEASE{\lst{q}}} = \tr_{ \lst{r} } \circ \sem{ S[\lst{r}/\lst{q}] } \circ (\ket{0}_{\lst{r}} \diamond \bra{0}_{\lst{r}})$} \\
    \specialrule{0em}{2pt}{2pt}
    (Proc) & \multicolumn{3}{l}{$\sem{\CALL\ \proc_i(\lst{a}_i)} = \bigsqcup_{n = 0}^{\infty} \sem{ S_i^{(n)}[\lst{a}_i / \lst{y}_i] }$}
  \end{tabular}
  \caption{Denotational semantics of $\EqRP$.}
  \label{tab_EdeSem}
\end{table}

The denotational semantics of a quantum program, denoted $\sem{\cdot}$, is defined as a $\QOP$.
The semantics of each term is given in a compositional way, except for the case of parameterized call statements.
To handle this case, we need to define the syntactic approximation (i.e., unrolling) of the bodies
of mutually recursive procedures with parameter passing.
\begin{definition}[Parameterized approximation]\label{def_para_unroll}
Let $\proc_i(\lst{y}_i)\in D$, $1 \leq i \leq n$, be parameterized recursive quantum procedure with body $S_i$.
Let $\CALL\ \proc_j(\lst{p}_{j_k})$, $1 \leq j \leq n$, be any parameterized call statement inside $S_i$
\big(Here $k$ means $S_i$ has (possibly different) activations of $\proc_j$ each parameterized $\lst{p}_{j_k}$\big).
Then the $k$th syntactic approximation $S_i^{(k)}$ of $S_i$ is defined as:
\begin{eqnarray*}
  S_i^{(0)} &\triangleq& \bottom \\
  S_i^{(k+1)} &\triangleq& S_i\big[\ldots, \big( \SKIP; S_j^{(k)}[\lst{p}_{j_k}/\lst{y}_j] \big) \big/ \CALL\ \proc_j(\lst{p}_{j_k}), \ldots\big]
\end{eqnarray*}
where $S_i\big[\ldots, \big( \SKIP; S_j^{(k)}[\lst{p}_{j_k}/\lst{y}_j] \big) \big/ \CALL\ \proc_j(\lst{p}_{j_k}), \ldots\big]$
stands for simultaneous substitution of the statement $\SKIP; S_j^{(k)}[\lst{p}_{j_k}/\lst{y}_j]$
for every $\CALL\ \proc_j(\lst{p}_{j_k})$ in $S_i$
\big(Here $\SKIP$ is used to simulate the first-step transition $\xrightarrow{\epsilon}$
for the statement $\CALL\ \proc_j(\lst{p}_{j_k})$, cf. the (Skip, Proc) rules of Tab. \ref{tab_EopSem}\big).
\end{definition}

The denotational semantics $\sem{\cdot}$ for $\EqRP[\Omega]$ is defined in Tab. \ref{tab_EdeSem}.
Note that the side conditions for many equations of denotational semantics are omitted,
because they follow the same ones as in the rules of operational semantics.
By extending Lem. \ref{lem_well_def_den} to the parameterized case,
well-definedness of the (Proc) rule in Tab. \ref{tab_EdeSem} follows.

\subsection{Connection between the two semantics}

As a preliminary, we need to prove two structural properties of the operational semantics
for the syntactic approximation of a parameterized recursive procedure.

\begin{lemma}\label{lem_syn_approx_pre}
  Let $S_i$ be the body of procedure $\proc_i(\lst{y}_i)$ with $1 \leq i \leq n$.
  \begin{description}
    \item[(1)] For any $\rho,\rho'$ and $\alpha$, we have that
  \begin{eqnarray*}
    \langle S_i, \rho \rangle \xrightarrow{\alpha} \langle E, \rho' \rangle &\iff& \exists k \geq 0.\ \langle S_i^{(k)}, \rho \rangle \xrightarrow{\alpha} \langle E, \rho' \rangle
  \end{eqnarray*}
    \item[(2)] Suppose, for some $k \geq 0$, that
  \begin{eqnarray*}
    \langle S_i^{(k)}, \rho \rangle &\xrightarrow{\alpha}& \langle E, \rho' \rangle
  \end{eqnarray*}
  Then we have, for all $l \geq k$, that
  \begin{eqnarray*}
    \langle S_i^{(l)}, \rho \rangle &\xrightarrow{\alpha}& \langle E, \rho' \rangle
  \end{eqnarray*}
  \end{description}
\end{lemma}
\begin{proof}
  The proof proceeds by induction on the length $|\alpha|$ of $\alpha$,
  and do a case analysis for the last step of the inductive definition of $S_i$.
\end{proof}

\begin{remark}
  Lem. \ref{lem_syn_approx_pre} says that every execution of $S_i$ (labeled by $\alpha$)
  can be simulated in a finite unrolling of $S_i$ (denoted by $S_i^{(k)}$),
  and also in any larger unrolling of $S_i$ (denoted by $S_i^{(l)}$ with $l \geq k$).
  Note that, by the least number principle, there exists the least such $k$.
\end{remark}

We are now positioned to relate the operational and denotational semantics.

\begin{theorem}\label{thm_sem_EqRP}
  For any quantum program $S\in \EqRP$, we have that
  \begin{eqnarray*}
    \sem{S}(\rho) &=& \sum_{\langle S,\rho\rangle\xrightarrow{\alpha}\langle E,\rho'\rangle} \rho'
  \end{eqnarray*}
  where the summation of $\rho'$ is taken for every possible $\alpha$ s.t. $\langle S,\rho\rangle\xrightarrow{\alpha}\langle E,\rho'\rangle$.
\end{theorem}
\begin{proof}
  The proof can be done by induction on the depth of the formation tree of $S$.
  We only consider the cases of variable localization and parameterized activation
  \big(The proof of other cases can be adapted from the corresponding proof in \cite{Ying11}\big).

  \paragraph{Case: (Loc).}
  \begin{eqnarray}
    \sem{\LOCAL{\lst{q}}; S; \RELEASE{\lst{q}}}(\rho) &=& \big[ tr_{ \lst{r} } \circ \sem{ S[\lst{r}/\lst{q}] } \circ (\ket{0}_{\lst{r}} \diamond \bra{0}_{\lst{r}}) \big](\rho) \nonumber\\
     &=& tr_{ \lst{r} } \big( \sem{ S[\lst{r}/\lst{q}] }(\rho \otimes \ket{0}_{\lst{r}}\bra{0}) \big) \nonumber\\
     &=& tr_{ \lst{r} } \Big( \sum_{\langle S[\lst{r}/\lst{q}], \rho \otimes \ket{0}_{\lst{r}}\bra{0}\rangle \xrightarrow{\alpha} \langle E,\rho'\rangle} \rho' \Big) \label{eq_param_sem_6} \\
     &=& \sum_{\langle S[\lst{r}/\lst{q}], \rho \otimes \ket{0}_{\lst{r}}\bra{0}\rangle \xrightarrow{\alpha} \langle E,\rho'\rangle} tr_{ \lst{r} } ( \rho' ) \label{eq_param_sem_7} \\
     &=& \sum_{\langle S[\lst{r}/\lst{q}], \rho \otimes \ket{0}_{\lst{r}}\bra{0}\rangle \xrightarrow{\alpha} \langle E,\rho'\rangle} \Big( \sum_{\langle \RELEASE{\lst{r}}, \rho'\rangle \xrightarrow{\alpha'} \langle E,\rho''\rangle} \rho'' \Big) \label{eq_param_sem_8} \\
     &=& \sum_{\langle S[\lst{r}/\lst{q}]; \RELEASE{\lst{r}}, \rho \otimes \ket{0}_{\lst{r}}\bra{0}\rangle \xrightarrow{\alpha} \langle \RELEASE{\lst{r}}, \rho'\rangle} \Big( \sum_{\langle \RELEASE{\lst{r}}, \rho'\rangle \xrightarrow{\alpha'} \langle E,\rho''\rangle} \rho'' \Big) \label{eq_param_sem_9} \\
     &=& \sum_{\langle S[\lst{r}/\lst{q}]; \RELEASE{\lst{r}}, \rho \otimes \ket{0}_{\lst{r}}\bra{0}\rangle \xrightarrow{\alpha} \langle E, \rho''\rangle} \rho'' \nonumber \\
     &=& \sum_{\langle \LOCAL{\lst{q}}; S; \RELEASE{\lst{q}}, \rho \rangle \xrightarrow{\alpha'} \langle E, \rho''\rangle} \rho'' \nonumber
  \end{eqnarray}
  where Eq. (\ref{eq_param_sem_6}) follows by induction hypothesis applied to $S[\lst{r}/\lst{q}]$,
  (\ref{eq_param_sem_7}) by linearity of partial trace function,
  (\ref{eq_param_sem_8}) by (Rel) rule of operational semantics,
  and (\ref{eq_param_sem_9}) by (Comp${}_2$) rule of operational semantics.
  \paragraph{Case: (Proc).}
  \begin{eqnarray}
    \sem{\CALL\ \proc_i(\lst{a}_i)}(\rho) &=& \big( \bigsqcup_{n = 0}^{\infty} \sem{ S_i^{(n)}[\lst{a}_i / \lst{y}_i] } \big) (\rho) \nonumber \\
     &=& \bigsqcup_{n = 0}^{\infty} \sem{ S_i^{(n)}[\lst{a}_i / \lst{y}_i] }(\rho) \nonumber \\
     &=& \bigsqcup_{n = 0}^{\infty} \Big( \sum_{\langle S_i^{(n)}[\lst{a}_i / \lst{y}_i], \rho \rangle\xrightarrow{\alpha}\langle E,\rho'\rangle} \rho' \Big) \label{eq_param_sem_10} \\
     &=& \sum_{\langle S_i[\lst{a}_i / \lst{y}_i],\rho\rangle\xrightarrow{\alpha}\langle E,\rho'\rangle} \rho' \label{eq_param_sem_11} \\
     &=& \sum_{\langle \CALL\ \proc_i(\lst{a}_i),\rho\rangle\xrightarrow{\alpha}\langle E,\rho'\rangle} \rho' \label{eq_param_sem_12}
  \end{eqnarray}
  where Eq. (\ref{eq_param_sem_10}) follows from a nonrecursive version of this theorem (proved similarly),
  (\ref{eq_param_sem_11}) by Lem. \ref{lem_syn_approx_pre},
  and (\ref{eq_param_sem_12}) by (Proc) rule of operational semantics.
\end{proof}

\section{Quantum assertion language and expressiveness}\label{app_sec_ass_exp}

\subsection{Formal semantics of $\PQPT$s}\label{app_subsec_pqpt}

\begin{definition}[Semantics of $\PQPT$s]\label{def_sem_pqpt}
Follow symbols and notations in Def. \ref{def_syn_pqpt}.
Recall that $\mathbb{I}$ is the standard interpretation
from constant and function symbols in the syntax of $\PQPT$s to their semantic counterparts,
and $v$ is an assignment s.t. $v(\calX_{\lst{q}})\in \Pp(\Hh_{\lst{q}})$.

The denotation of $P_{\lst{q}}$ under interpretation $\mathbb{I}$ and assignment $v$,
denoted $P_{\lst{q}}^{\mathbb{I},v}$, is defined as
\begin{eqnarray*}
  P_{\lst{q}}^{\mathbb{I},v} &\triangleq& \calE_{\lst{q}}^*(\calB_{\lst{q}}^{\mathbb{I},v}) + \calF_{\lst{q}}^*(I_{\lst{q}})
\end{eqnarray*}
where $\calB_{\lst{q}}^{\mathbb{I},v}$,
the denotation of $\calB_{\lst{q}}$ under interpretation $\mathbb{I}$ and assignment $v$,
is defined as
\begin{eqnarray*}
  \calB_{\lst{q}}^{\mathbb{I},v} &\triangleq& \left\{
                           \begin{array}{ll}
                             v(\calX_{\lst{q}}) & \hbox{if $\calB_{\lst{q}} = \calX_{\lst{q}}$} \\
                             I_{\lst{r}}\otimes \calB_{\lst{s}}^{\mathbb{I},v} & \hbox{if $\calB_{\lst{q}} = I_{\lst{r}}\otimes \calB_{\lst{s}}$} \\
                             \calB_{\lst{r}}^{\mathbb{I},v} \otimes I_{\lst{s}} & \hbox{if $\calB_{\lst{q}} = \calB_{\lst{r}}\otimes I_{\lst{s}}$} \\
                             \calB_{\lst{r}}^{\mathbb{I},v} \otimes \calB_{\lst{s}}^{\mathbb{I},v} & \hbox{if $\calB_{\lst{q}} = \calB_{\lst{r}}\otimes \calB_{\lst{s}}$}
                           \end{array}
                         \right.
\end{eqnarray*}
where a bit notational abuses between semantics (the left) and syntax (the right) are allowed, e.g., the left $I_{\lst{r}}$, or strictly ${I_{\lst{r}}}^{\mathbb{I}}$, is the denotation of the right $I_{\lst{r}}$ under interpretation $\mathbb{I}$.
\end{definition}

\subsection{Weakest (liberal) preconditions}\label{app_syn_wp}

\begin{table}[!htbp]

  \centering

  \begin{tabular}{lcl}
    \multicolumn{3}{l}{$\fwp.\big(\CALL\ \proc_i(\lst{a}_i)\big).P = \bigsqcup_{n=0}^\infty \fwp.S_i^{(n)}[\lst{a}_i / \lst{y}_i].P$} \\
      \specialrule{0em}{2pt}{2pt}
    \multicolumn{3}{l}{$\fwlp.\big(\CALL\ \proc_i(\lst{a}_i)\big).P = \bigsqcap_{n=0}^\infty \fwlp.S_i^{(n)}[\lst{a}_i / \lst{y}_i].P$} \\
      \specialrule{0em}{2pt}{2pt}
    $\fwp.\bottom.P = 0$ & & $\fwlp.\bottom.P = I$ \\
      \specialrule{0em}{2pt}{2pt}
    $\fxp.\SKIP.P = P$ & & $\fxp.(q \assnequal \ket{0}).P = \sum_i \ket{i}_q \bra{0} P \ket{0}_q \bra{i}$ \\
      \specialrule{0em}{2pt}{2pt}
    $\fxp.(\lst{q} \starequal U).P = U^\dag P U$ & & $\fxp.(S_1;S_2).P = \fxp.S_1.(\fxp.S_2.P)$ \\
      \specialrule{0em}{2pt}{2pt}
    \multicolumn{3}{l}{$\fxp.\IF.P = \sum_m M_m^\dag(\fxp.S_m.P)M_m$} \\
      \specialrule{0em}{2pt}{2pt}
    \multicolumn{3}{l}{$\fxp.(\LOCAL{\lst{q}}; S; \RELEASE{\lst{q}}).P = \bra{0}_{\lst{r}} \big( \fxp.S[\lst{r}/\lst{q}].(P \otimes I_{\lst{r}}) \big) \ket{0}_{\lst{r}}$}
  \end{tabular}
  \caption{Definition of formal weakest (liberal) preconditions --- $\fxp \in \{ \fwp, \fwlp\}$.}
  \label{app_tab_wp_wlp} 
  \vspace{-8pt}
\end{table}

\begin{theorem}[Quantum expressiveness theorem]\label{app_thm_exp}
Let quantum program $S \in \EqRP$, and $P$ a $\PQPT$.
Define $\fwp.S.P$ and $\fwlp.S.P$ as in Tab. \ref{app_tab_wp_wlp}.
(Note that by extending Lem. \ref{lem_wp_wlp} to the parameterized case,
well-definedness of $\fwp$ and $\fwlp$ for parameterized activation follows.)
It is the case that
\begin{description}
  \item[(a)] $\models_\mathbb{I} \fwp.S.P = \sem{S}^*(P)$;
  \item[(b)] $\models_\mathbb{I} \fwlp.S.P = I - \fwp.S.(I - P)$.
\end{description}
\end{theorem}
\begin{proof}
Since $\sem{S}$ can be seen as a $\QOP$,
by the inductive definition of $\sem{S}$ (cf. Tab. \ref{tab_EdeSem}), together with Lem. \ref{lem_dual_qop},
we are able to obtain an inductive definition of $\sem{S}^*$ (cf. Tab. \ref{tab_dualOfSem}).
Then Stat. (a) follows by induction on $S$, in which for the case of parameterized $\CALL$-statement,
it suffices to show, for all $n \geq 0$, that
\begin{eqnarray}\label{eq_syn_wp_1}
  &\models_\mathbb{I}& \fwp.S_i^{(n)}[\lst{a}_i / \lst{y}_i].P = \sem{S_i^{(n)}[\lst{a}_i / \lst{y}_i]}^*(P)
\end{eqnarray}
This is indeed the case because $S_i^{(n)}$ is a non-recursive quantum program.
\big(We can first prove a non-recursive version of Stat. (a) by induction on $S$.\big)
Stat. (b) can be obtained similar to (a) by using the fact that
the G. L. B. operator $\bigsqcap$ can be defined
as the logical dual of the L. U. B. operator $\bigsqcup$,
together with linearity of super operators.
\begin{table}[!htbp]
  \centering
  \begin{tabular}{lcl}
    $\sem{\bottom}^* = 0 \diamond 0$ & & $\sem{\SKIP}^* = I \diamond I$ \\
    \specialrule{0em}{2pt}{2pt}
    $\sem{q \assnequal \ket{0}}^* = \sum_i \ket{i}_q \bra{0} \diamond \ket{0}_q \bra{i}$ & & $\sem{\lst{q} \starequal U}^* = U^\dag \diamond U$ \\
    \specialrule{0em}{2pt}{2pt}
    $\sem{S_1;S_2}^* = \sem{S_1}^* \circ \sem{S_2}^*$ & & $\sem{\IF}^* = \sum_m (M_m^\dag \diamond M_m) \circ \sem{S_m}^*$ \\
    \specialrule{0em}{2pt}{2pt}
    \multicolumn{3}{l}{$\sem{\LOCAL{\lst{q}}; S; \RELEASE{\lst{q}}}^* = (\bra{0}_{\lst{r}} \diamond \ket{0}_{\lst{r}}) \circ \sem{ S[\lst{r}/\lst{q}] }^* \circ (I_{\lst{r}} \diamond I_{\lst{r}})$} \\
    \specialrule{0em}{2pt}{2pt}
    \multicolumn{3}{l}{$\sem{\CALL\ \proc_i(\lst{a}_i)}^* = \bigsqcup_{n = 0}^{\infty} \sem{ S_i^{(n)}[\lst{a}_i / \lst{y}_i] }^*$}
  \end{tabular}
  \caption{The inductive definition of $\sem{S}^*$}
  \label{tab_dualOfSem}
\end{table}
\end{proof}

\section{Two counterexamples for no (R Subst)}\label{app_exams_counter}

\subsection{Counterexample for (Rp Rec)}\label{app_exam_counter_par}

{\bf Example \ref{exam_counter_par}.}
Let $q$ be a quantum variable with $\type(q) = \integer$.
We define the $(+i)$-operator $U_{+i}$ over the computational basis of $\Hh_q$ by
$$
  U_{+i}\colon \ket{x}\rightarrow\ket{x+i}
$$
and similarly for the $(-i)$-operator $U_{-i}$.
Declare the procedure $\toy$ by
  $$
  \recDec{\langle \toy \rangle}{\ifStat{\Box m\cdot M[q] = m \rightarrow S_m}}
  $$
with $M$ defined by
\begin{eqnarray*}
  M &\triangleq& \bigg\{ M_0 \triangleq \sum_{i \leq 0} \outprod{i}{i},\ M_1 \triangleq \sum_{i \geq 1} \outprod{i}{i} \bigg\}
\end{eqnarray*}
and $\{S_m\}_{m = 0, 1}$ defined by
$$
S_0 \triangleq \SKIP, \quad S_1 \triangleq q \starequal U_{-1};\ \CALL\ \toy;\ q \starequal U_{+1}
$$
Fix $n \geq 0$. We can derive the partial correctness formula
\begin{equation}\label{toy_par_cor}
  \apcor{\big}{\voutprod{n}{q}{n}}{\CALL\ \toy}{\voutprod{n}{q}{n}}
\end{equation}
by using (Rp Rec).
However, this is not the case if the use of (R Subst) is disallowed.

\begin{proof}
The proof of this lemma is divided into the following two parts:

\vspace{2mm}
\noindent {\bf Unprovability without (R Subst).}
Suppose for a contradiction that one can derive Hoare's triple (\ref{toy_par_cor}) without using (R Subst).
By (R Order), we have to show that there exist $\PQPT$s $P$ and $Q$ s.t.
\begin{gather}
  \voutprod{n}{q}{n} \sqsubseteq P \nonumber\\
  \pcor{P}{\CALL\ \toy}{Q} \label{exam_counter_par_21} \\
  Q \sqsubseteq \voutprod{n}{q}{n} \label{exam_counter_par_22}
\end{gather}
where Hoare's triple (\ref{exam_counter_par_21}) is derived by using (Rp Rec).
Then we have to show that
\begin{eqnarray*}
  \trueOrder, \pcor{P}{\CALL\ \toy}{Q} &\reVdash{\qBase^-}& \pcor{P}{\ifStat{\Box m\cdot M[q] = m \rightarrow S_m}}{Q}
\end{eqnarray*}
where the proof system $\qBase^-$ is defined by
\begin{eqnarray*}
  \qBase^- &\triangleq& \qBase - \mbox{(R Subst)}
\end{eqnarray*}
By (R Case), we have to prove that there exist $\PQPT$s $\{P_m\}_{m = 0, 1}$ s.t.
\begin{eqnarray}
  P &\sqsubseteq& \sum_{m} M_m^{\dag} P_m M_m \nonumber \\
  \trueOrder, \pcor{P}{\CALL\ \toy}{Q} &\reVdash{\qBase^-}& \pcor{P_0}{S_0}{Q} \nonumber \\
  \trueOrder, \pcor{P}{\CALL\ \toy}{Q} &\reVdash{\qBase^-}& \pcor{P_1}{S_1}{Q} \label{exam_counter_par_24}
\end{eqnarray}
To prove Ass. (\ref{exam_counter_par_24}), we have to show that
$$
\{ P_1 \}\  q \starequal U_{-1}\ \{ P \}\ \CALL\ \toy\ \{ Q \}\ q \starequal U_{+1}\ \{ Q \}
$$
By (A Unit, R Order), we have to show that
\begin{eqnarray}
  P_1 &\sqsubseteq& (U_{-1})^{\dag} P (U_{-1}) \nonumber \\
  Q &\sqsubseteq& (U_{+1})^{\dag} Q (U_{+1}) \label{exam_counter_par_25}
\end{eqnarray}
By Ass. (\ref{exam_counter_par_22}), it follows that
\begin{eqnarray}\label{exam_counter_par_26}
  Q &=& \delta \voutprod{n}{q}{n}, \quad \mbox{ with } \delta \in [0, 1]
\end{eqnarray}
This together with Ass. (\ref{exam_counter_par_25}) implies that
\begin{eqnarray*}
  \delta \voutprod{n}{q}{n} &\sqsubseteq& \delta \voutprod{n-1}{q}{n-1}
\end{eqnarray*}
A contradiction.

\vspace{2mm}
\noindent {\bf Provability with (R Subst).}
  We introduce quantum predicate variable $\calX_{q}$ for $q$ (abbr. $\calX$)
  with the $i$th main-diagonal element $\vinprod{i}{\calX}{i}$ abbreviated $\calX_i$.
  To prove Hoare's triple (\ref{toy_par_cor}),
  by (R Subst), together with the substitution $[(\voutprod{n}{q}{n})/\calX]$, it suffices to derive
  $$
   \apcor{\bigg}{\sum_{i \geq 0} \calX_i \outprod{i}{i}}{\CALL\ \toy}{\sum_{i \geq 0} \calX_i \outprod{i}{i}}
  $$
  By (Rp Rec), it suffices to show
  \begin{eqnarray*}
    \trueOrder, \avpcor{\big}{\sum_{i \geq 0} \calX_i \outprod{i}{i}}{\CALL\ \toy}{\sum_{i \geq 0} \calX_i \outprod{i}{i}}
    &\reVdash{\qBase}& \vpcor{\sum_{i \geq 0} \calX_i \outprod{i}{i}}{\ifStat{\Box m\cdot M[q] = m \rightarrow S_m}}{\sum_{i \geq 0} \calX_i \outprod{i}{i}}
  \end{eqnarray*}
  The routine verification of the above assertion is left to the reader,
  where (R Subst) will be applied with the substitution $[ ( \sum_{i \geq 0} \calX_{i+1} \outprod{i}{i} ) / \calX ]$
  to verifying the inner $\CALL\ \toy$.
\end{proof}
\begin{remark}
  Closer scrutiny of the above proof reveals that
  the postcondition $Q$ for the outer $\CALL\ \toy$ should have the form $\delta \voutprod{n}{q}{n}$
  \big(cf. Ass. (\ref{exam_counter_par_21}, \ref{exam_counter_par_22}, \ref{exam_counter_par_26})\big);
  in contrast, the postcondition $Q$ for the inner $\CALL\ \toy$ should have the form $\delta \voutprod{n-1}{q}{n-1}$
  \big(cf. Ass. (\ref{exam_counter_par_22}, \ref{exam_counter_par_25})\big).
  Unfortunately, the variation of $Q$ is beyond the expressibility of $\QPRED$.
  By choosing $P,Q$ as $\PQPT$s (containing parameters) and then applying (R Subst),
  one can achieve the transformation of $Q$ from $\delta \voutprod{n}{q}{n}$ to $\delta \voutprod{n-1}{q}{n-1}$
  as pointed out above.

\end{remark}

\subsection{Counterexample for (Rt Rec)}\label{app_exam_counter_tot}

{\bf Example \ref{exam_counter_tot}.}
Let the recursive procedure $\toy$ be as defined in Exm. \ref{exam_counter_par}.
Fix $n \geq 0$. We can derive the total correctness formula
\begin{equation}\label{toy_tot_cor}
  \atcor{\big}{\voutprod{n}{q}{n}}{\CALL\ \toy}{\voutprod{n}{q}{n}}
\end{equation}
by using (Rt Rec).
However, this is not the case if the use of (R Subst) is disallowed.

\begin{proof}
The proof of this lemma is divided into the following two parts:

\vspace{2mm}
\noindent {\bf Unprovability without (R Subst).}
Suppose for a contradiction that one can derive Hoare's triple (\ref{toy_tot_cor}) without using (R Subst).
By (R Order), we have to show that there exist $\PQPT$s $P$ and $Q$ s.t.
\begin{gather}
  \voutprod{n}{q}{n} \quad \sqsubseteq \quad P \nonumber\\
  \tcor{P}{\CALL\ \toy}{Q} \label{exam_counter_tot_21} \\
  Q \quad \sqsubseteq \quad \voutprod{n}{q}{n} \label{exam_counter_tot_22}
\end{gather}
where Hoare's triple (\ref{exam_counter_tot_21}) is derived by using (Rt Rec).
Then we have to show that there exists a sequence of $\PQPT$s $\{P_i\}_{i \geq 0}^{\sqsubseteq}$
with $P_0 = 0$ and $P \sqsubseteq \bigsqcup_{i = 0}^{\infty} P_i$ s.t.
\begin{eqnarray*}
  \trueOrder, \tcor{P_i}{\CALL\ \toy}{Q} &\reVdash{\qBase^-}& \tcor{P_{i+1}}{\ifStat{\Box m\cdot M[q] = m \rightarrow S_m}}{Q}
\end{eqnarray*}
for all $i \geq 0$ \big($\qBase^- \triangleq \qBase - \mbox{(R Subst)}$\big).
In particular (when $i = 0$), we need to show that
\begin{eqnarray*}
  \trueOrder, \tcor{0}{\CALL\ \toy}{Q} &\reVdash{\qBase^-}& \tcor{P_1}{\ifStat{\Box m\cdot M[q] = m \rightarrow S_m}}{Q}
\end{eqnarray*}
By (R Case), we have to show that there exist $\PQPT$s $\{R_m\}_{m = 0, 1}$ s.t.
\begin{eqnarray}
  P_1 &\sqsubseteq& \sum_{m = 0, 1} M_m^{\dag} R_m M_m \nonumber \\
  \trueOrder, \tcor{0}{\CALL\ \toy}{Q} &\reVdash{\qBase^-}& \tcor{R_0}{S_0}{Q} \nonumber \\
  \trueOrder, \tcor{0}{\CALL\ \toy}{Q} &\reVdash{\qBase^-}& \tcor{R_1}{S_1}{Q} \label{exam_counter_tot_24}
\end{eqnarray}
To prove Ass. (\ref{exam_counter_tot_24}), we have to show that
$$
\{ R_1 \}\ q \starequal U_{-1}\ \{ 0 \}\ \CALL\ \toy_1\ \{ Q \}\ q \starequal U_{+1}\ \{ Q \}
$$
By (A Unit, R Order), we have to show that
\begin{eqnarray}
  R_1 &=& 0 \nonumber \\
  Q &\sqsubseteq& (U_{+1})^{\dag} Q (U_{+1}) \label{exam_counter_tot_25}
\end{eqnarray}
By Ass. (\ref{exam_counter_tot_22}), it follows that
\begin{eqnarray*}
  Q &=& \delta \voutprod{n}{q}{n}, \quad \mbox{with } \delta \in [0, 1]
\end{eqnarray*}
This together with Ass. (\ref{exam_counter_tot_25}) implies that
\begin{eqnarray*}
  \delta \voutprod{n}{q}{n} &\sqsubseteq& \delta \voutprod{n-1}{q}{n-1}
\end{eqnarray*}
A contradiction.

\vspace{2mm}
\noindent {\bf Provability with (R Subst).}
We introduce quantum predicate variable $\calX_{q}$ for $q$ (abbr. $\calX$)
  with the $i$th main-diagonal element $\vinprod{i}{\calX}{i}$ (abbr. $\calX_i$).
  To prove Hoare's triple (\ref{toy_tot_cor}),
  by (R Subst), together with the substitution $[(\voutprod{n}{q}{n})/\calX]$, it suffices to derive
  $$
  \atcor{\bigg}{\sum_{j \geq 0} \calX_j \outprod{j}{j}}{\CALL\ \toy}{\sum_{j \geq 0} \calX_j \outprod{j}{j}}
  $$
  By (Rt Rec), it suffices to show, for all $i \geq 0$, that
  \begin{eqnarray*}
    \trueOrder, \avtcor{\big}{\sum_{0 \leq j < i} \calX_j \outprod{j}{j}}{\CALL\ \toy}{\sum_{j \geq 0} \calX_j \outprod{j}{j}} &\reVdash{\qBase}& \avtcor{\big}{\sum_{0 \leq j < i+1} \calX_j \outprod{j}{j}}{\ifStat{\Box m\cdot M[q] = m \rightarrow S_m}}{\sum_{j \geq 0} \calX_j \outprod{j}{j}}
  \end{eqnarray*}
  The routine verification of the above assertion is left to the reader,
  where (R Subst) will be applied with the substitution $[ ( \sum_{i > j \geq 0} \calX_{j+1} \outprod{j}{j} ) / \calX ]$
  to verifying the inner $\CALL\ \toy$.
\end{proof}

\begin{remark}
The reader might well wonder whether the contradiction of Ass. (\ref{exam_counter_tot_24})
is due to the fixed choice of postcondition (i.e. $Q$) of $\CALL\ \toy$ when using (Rt Rec).
In other words, one may derive Ass. (\ref{exam_counter_tot_21})
by using a variant of (Rt Rec) with increasing postconditions, e.g.
$$
\dfrac{\begin{array}{c}
    \exists\ \{P_n\}_{n \geq 0}^{\sqsubseteq}, \{Q_n\}_{n \geq 0}^{\sqsubseteq}  \mbox{ with } P_0 = 0 \mbox{ s.t.} \\
    \tcor{P_n}{\CALL\ \proc}{Q_n} \reVdash{\qBE} \tcor{P_{n+1}}{S}{Q_{n+1}} \mbox{ for all } n \geq 0, \\
    P \sqsubseteq\bigsqcup_{n=0}^{\infty} P_n \mbox{ and } \bigsqcup_{n=0}^{\infty} Q_n \sqsubseteq Q
    \end{array}
    }{ \tcor{P}{\CALL\ \proc}{Q} }
$$
However, we can refute it by redefining $S_1$ as
\begin{eqnarray*}
  S_1 &\triangleq& q \starequal U_{-1};\ \CALL\ \toy;\ q \starequal U_{-1};\ \CALL\ \toy;\ q \starequal U_{+2}
\end{eqnarray*}
To summarize, any variant of (Rt Rec) always has a counterexample for no (R Subst).
\end{remark}

\section{Soundness and completeness}\label{app_sec_sound_complete}

\begin{table}\footnotesize
  \centering

  \begin{minipage}[b]{\textwidth}
    \centering
    \begin{tabular}{rlcrl}
    (A Bot) & $\pcor{I}{\bottom}{P}$ \big(resp. $\tcor{0}{\bottom}{P}$\big) & &
    (A Skip) & $\pcor{P}{\SKIP}{P}$ \\
    \specialrule{0em}{4pt}{4pt}
    (A Unit) & $\dfrac{U U^{\dag} = U^{\dag} U = I_{\lst{q}}}{\pcor{U^\dag P U}{\lst{q} \starequal U}{ P }}$ & &
    (A Init) & $\dfrac{\sum_{i} \voutprod{i}{q}{i} = I_{q}}{\pcor{ \sum_{i} \ket{i}_q \bra{0} P \ket{0}_q \bra{i} }{ q \assnequal \ket{0} }{ P }}$ \\
    \specialrule{0em}{4pt}{4pt}
    (R Comp) & $\dfrac{\pcor{P}{S_1}{Q} \quad \pcor{Q}{S_2}{R}}{\pcor{P}{ S_1;S_2 }{R}}$ & &
    (R Case) & $\dfrac{\pcor{P_m}{S_m}{Q} \mbox{ \small for each } m}{\pcor{\sum_m M_m^\dag P_m M_m}{\IF}{Q}}$ \\
    \specialrule{0em}{4pt}{4pt}
    (R Order) & $\dfrac{P\sqsubseteq P' \quad \pcor{P'}{S}{Q'} \quad Q'\sqsubseteq Q}{\pcor{P}{S}{Q}}$ & &
    (R Subst) & $\dfrac{\pcor{P}{S}{Q}}{\pcor{P[R/\calX]}{S}{Q[R/\calX]}}$ \\
    \specialrule{0em}{4pt}{4pt}
    (R Loc) & $\dfrac{\pcor{P \otimes I_{\lst{r}}}{\lst{r} := \ket{0};S[\lst{r}/\lst{q}]}{Q \otimes I_{\lst{r}}}}{\pcor{P}{\LOCAL{\lst{q}}; S; \RELEASE{\lst{q}}}{Q}}$ & &
    (R Adap) & $\dfrac{\pcor{P}{S}{Q}}{\pcor{ P[\lst{p}/\lst{q}] }{S[\lst{p}/\lst{q}]}{Q[\lst{p}/\lst{q}] }}$
  \end{tabular}
  \subcaption{Extended base proof system $\qBE$.}
  \label{qBE}
  \end{minipage}

  \vfill

  \begin{minipage}[b]{\textwidth}
    \centering
    \begin{tabular}{rl}
    (Rp pRec) & $\dfrac{\big\{ \pcor{P_i}{\CALL\ \proc_i(\lst{y}_i)}{Q_i} \big\}_{1 \leq i \leq n} \reVdash{\qBE} \bigwedge_{1 \leq i \leq n} \pcor{P_i}{S_i}{Q_i}}{\bigwedge_{1 \leq i \leq n} \pcor{P_i}{\CALL\ \proc_i(\lst{y}_i)}{Q_i}}$ \\
    \specialrule{0em}{4pt}{4pt}
    (Rt pRec) & $\dfrac{\begin{array}{c}
    \mbox{for } 1 \leq i \leq n,\ \exists\ \{P_i^j\}_{j \geq 0}^{\sqsubseteq}  \mbox{ with } P_i^0 = 0 \mbox{ s.t.} \\
    \big\{ \tcor{P_i^j}{\CALL\ \proc_i(\lst{y}_i)}{Q_i} \big\}_{1 \leq i \leq n} \reVdash{\qBE} \bigwedge_{1 \leq i \leq n} \tcor{P_i^{j+1}}{S_i}{Q_i} \mbox{ for all } j \geq 0, \\
    P_i \sqsubseteq\bigsqcup_{j=0}^{\infty} P_i^j
    \end{array}
    }{\bigwedge_{1 \leq i \leq n} \tcor{P_i}{\CALL\ \proc_i(\lst{y}_i)}{Q_i} }$
    \end{tabular}
    \subcaption{Proof rules for parameterized procedures.}
    \label{proof_rules_pProc}
  \end{minipage}

  \caption{Proof systems $\qPP \triangleq \qBE + \mbox{(Rp pRec)}$ and $\qTP \triangleq \qBE + \mbox{(Rt pRec)}$.}
  \label{prf_sys_exp}
\end{table}

\paragraph{Notations and definitions.}
Quantum programming language $\EqPL$ is defined by:
\begin{displaymath}
\begin{array}{rcl}
  S & \triangleq & \bottom \mid \SKIP \mid q \assnequal \ket{0} \mid \lst{q} \starequal U \mid S_1;S_2 \mid \\
    & & \ifStat{\Box m\cdot M[\lst{q}] = m \rightarrow S_m} \mid \LOCAL{\lst{q}}; S; \RELEASE{\lst{q}}
\end{array}
\end{displaymath}
(an extension to the quantum base language $\qPL$).
Proof system $\qBE$ for both partial and total correctness of $\EqPL$
(an extension to $\qBase$ in Tab. \ref{tab_qBase}),
partial- and total-correctness proof systems $\qPP$ and $\qTP$ for $\EqRP$
are shown in Tab. \ref{prf_sys_exp} (also cf. Tab. \ref{tab_QHL_eRqPL}).

\paragraph{Organization and proof sketch}
We will discuss two versions of soundness and completeness
--- one for the general case and the other for the compact case.
Generally speaking,
the soundness and completeness of $\qPP$ and $\qTP$ will be reduced to that of $\qBE$
by simulating recursive procedures with their syntactic approximations
(To see this, we remark that
the inference rules for both partial and total correctness of recursive procedures are designed on the basis of $\qBE$).
In particular,
for the proof of the completeness,
we will take the notion of the most general partial (resp. total) correctness formula
$\apcor{\big}{\fwlp.(\CALL\ \proc(\lst{y})).\calX}{\CALL\ \proc(\lst{y})}{\calX}$
\big(resp. $\atcor{\big}{\fwp.(\CALL\ \proc(\lst{y})).\calX}{\CALL\ \proc(\lst{y})}{\calX}$\big)
whose original idea comes from the theory of classical Hoare logic \cite[Chap. 6]{Francez}.
Note that we call them the most general correctness formulas because any correct Hoare's triple
for the parameterized activation $\CALL\ \proc(\lst{a})$ can be deduced from them by using (R Adapt) and (R Subst).

\subsection{Soundness and completeness of $\qPP$}\label{app_subsec_qPP}

\begin{lemma}[Soundness and completeness of $\qBE$]\label{lem_sound_complete_EqPL}
For any $S\in \EqPL$ and any $\PQPT$s $P,Q$, we have that
\begin{equation}\label{ass_sound_complete_EqPL}
  \trueOrder \vdash_{\qBE} \pcor{P}{S}{Q} \iff \models_\mathbb{I} P \sqsubseteq \fwlp.S.Q
\end{equation}
in particular,
\begin{equation}\label{ass_cmpt_sound_complete_EqPL}
  \trueEquality \vdash_{\qBE} \pcor{P}{S}{Q} \iff \models_\mathbb{I} P = \fwlp.S.Q
\end{equation}
\end{lemma}
\begin{proof}
Note that the unique difference between
Ass. (\ref{ass_sound_complete_EqPL}, \ref{ass_cmpt_sound_complete_EqPL})
lies in applications of (R Order):
the former will take its original form;
while the latter can only take the form
$$
\dfrac{P = P' \quad \pcor{P'}{S}{Q'} \quad Q' = Q}{\pcor{P}{S}{Q}}
$$
So, in order to prove the lemma,
it suffices to show that every proof rule $\calR$ of $\qBE$ except for (R Order)
has the following  property:
\begin{quote}
  every Hoare's triple $\pcor{P}{S}{Q}$ in the antecedent of $\calR$ satisfies $\models_{\mathbb{I}} P = \fwlp.S.Q$

\end{quote}
if, and only if,
\begin{quote}
  every Hoare's triple $\pcor{P}{S}{Q}$ in the consequent of $\calR$ satisfies $\models_{\mathbb{I}} P = \fwlp.S.Q$.
\end{quote}
In other words, every axiom should have the form $\pcor{\fwlp.S.P}{S}{P}$
and every inference rule with exception of (R Order) should preserves this form bidirectionally
(That is that, the ``only if'' direction entails compact soundness of the rule $\calR$,
and the ``if'' direction entails compact completeness of $\calR$).
Consider the cases of (R Subst), (R Loc) and (R Adapt).
(For other cases, cf. the intuition behind $\qBase$ in Sec. \ref{sec_prf}.)

{\bf Case: (R Loc).}
It suffices to show the following two assertions
\begin{eqnarray}
  &\models_{\mathbb{I}}& P \otimes I_{\lst{r}} = \fwlp.\big( \lst{r} \assnequal \ket{0}; S[\lst{r}/\lst{q}] \big).(Q \otimes I_{\lst{r}}) \label{eq_cmpt_sound_qBE_4} \\
  &\models_{\mathbb{I}}& P = \fwlp.\big( \LOCAL{\lst{q}}; S; \RELEASE{\lst{q}} \big).Q \label{eq_cmpt_sound_qBE_5}
\end{eqnarray}
are equivalent. By definition of $\fwlp$ (cf. Tab. \ref{app_tab_wp_wlp}), we have that
\begin{align*}
  & \fwlp.\big( \lst{r} \assnequal \ket{0}; S[\lst{r}/\lst{q}] \big).(Q \otimes I_{\lst{r}}) \\
  & = \fwlp.(\lst{r} \assnequal \ket{0}).\big( \fwlp.S[\lst{r}/\lst{q}].(Q \otimes I_{\lst{r}}) \big) \\
  & = \sum_{i} \voutprod{i}{\lst{r}}{0} \big( \fwlp.S[\lst{r}/\lst{q}].(Q \otimes I_{\lst{r}}) \big) \voutprod{0}{\lst{r}}{i} \\
  & = \bra{0}_{\lst{r}} \big( \fwlp.S[\lst{r}/\lst{q}].(Q \otimes I_{\lst{r}}) \big) \ket{0}_{\lst{r}} \otimes I_{\lst{r}} \\
  & = \big( \fwlp.( \LOCAL{\lst{q}}; S; \RELEASE{\lst{q}}).Q \big) \otimes I_{\lst{r}}
\end{align*}
Then equivalence of Ass. (\ref{eq_cmpt_sound_qBE_4}, \ref{eq_cmpt_sound_qBE_5})
follows from the convention that $P \otimes I_{\lst{r}} = P$.

{\bf Case: (R Subst, R Adap).}
By Thm. \ref{thm_exp}.
\end{proof}

\begin{theorem}[Soundness and Completeness of $\qPP$]\label{thm_sound_complete_qPP}
$\qPP$ is both sound and complete.
That is, for any quantum program $S \in \EqRP$ and any $\PQPT$s $P,Q$, we have that
$$
\trueOrder \reVdash{\qPP} \pcor{P}{ S }{Q} \iff \models_\mathbb{I} \pcor{P}{ S }{Q}
$$
\end{theorem}
\begin{proof}
  The proof is divided into two parts: one for $\implies$ and the other for $\impliedby$.

  {\bf ($\implies$).}
The soundness of $\qPP$ can be reduced to that of $\qBE$
by simulating recursive procedures with their syntactic approximations.
To see this, suppose that
$$\pcor{P_k}{ \CALL\ \proc_k(\lst{y}_k)}{Q_k}$$
is deduced by using (Rp pRec).
Then we have to prove that
\begin{eqnarray}\label{eq_sound_qPP_5}
  &\models_{\mathbb{I}}& \pcor{P_k}{\CALL\ \proc_k(\lst{y}_k)}{Q_k}
\end{eqnarray}
By the supposition,
there are parameterized procedures $\proc_i(\lst{y}_i)$ with bodies $S_i$, $1 \leq i \leq n$ and $i \neq k$
(Here it is required that $1 \leq k \leq n$)
and a set of $\PQPT$s $\{P_i, Q_i\}_{1 \leq i \leq n}^{i \neq k}$ s.t.
\begin{eqnarray}\label{eq_sound_qPP_4}
  \trueOrder, \Big\{ \pcor{P_i}{\CALL\ \proc_i(\lst{y}_i)}{Q_i} \Big\}_{1 \leq i \leq n}
  &\reVdash{\qBE}& \bigwedge_{1 \leq i \leq n}\pcor{P_i}{S_i}{Q_i}
\end{eqnarray}

\begin{claim}\label{claim_sound_qPP}
  For every $j\geq 0$, it is the case that
  \begin{eqnarray}\label{eq_sound_qPP_7}
    \trueOrder &\reVdash{\qBE}& \bigwedge_{1 \leq i \leq n} \pcor{P_i}{S_i^{(j)}}{Q_i}
  \end{eqnarray}
\end{claim}
\paragraph{Proof of Claim \ref{claim_sound_qPP}.}
  By induction on $j$.

\noindent {\bf (Basis).} By (A Bot) and (R Order), together with the fact that $(P_i \sqsubseteq I) \in \trueOrder$.

\noindent {\bf (Induction).}
    Recalling the definition of $S_i^{(j+1)}$, i.e.
    \begin{eqnarray*}
      S_i^{(j+1)} &\triangleq& S_i \big[ \ldots, \big( \SKIP; S_l^{(j)}[\lst{a}_{l_m}/\lst{y}_l] \big) \big/ \CALL\ \proc_l(\lst{a}_{l_m}), \ldots \big]
    \end{eqnarray*}
    the inductive step can be done by simulating the proof of Ass. (\ref{eq_sound_qPP_4})
    with $\SKIP; S_l^{(j)}$ \big(resp. $\SKIP; S_l^{(j)}[\lst{a}_{l_m}/\lst{y}_l]$\big) in place of $\CALL\ \proc_l(\lst{y}_l)$ \big(resp. $\CALL\ \proc_l(\lst{a}_{l_m})$\big).

By Def. \ref{def_cor_qpr}, together with definition of $\fwlp$,
the proof of Ass. (\ref{eq_sound_qPP_5}) is reduced to proving
\begin{eqnarray*}
  &\models_{\mathbb{I}}& \pcor{P_k}{S_k^{(j)}}{Q_k}
\end{eqnarray*}
for all $j\geq 0$,
following from soundness of $\qBE$ (cf. Lem. \ref{lem_sound_complete_EqPL}) and Claim \ref{claim_sound_qPP}.

{\bf ($\impliedby$).}
By (R Order), together with Def. \ref{def_cor_qpr}, it suffices to show that
\begin{eqnarray*}
  \trueEquality &\reVdash{\qPP}& \pcor{\fwlp.S.Q}{S}{Q}
\end{eqnarray*}
for any $S \in \EqRP$ and any $\PQPT$s $Q$.
In the following, we only consider the case of parameterized activation,
i.e. $S \equiv \CALL\ \proc_k(\lst{a}_k)$ (Cf. Lem. \ref{lem_sound_complete_EqPL} for other cases).

By (R Adap), together with (R Subst), it suffices to show that
\begin{eqnarray*}
   \trueEquality &\reVdash{\qPP}& \apcor{\big}{\fwlp.\big( \CALL\ \proc_k(\lst{y}_k) \big).\calX_k}{\CALL\ \proc_k(\lst{y}_k)}{\calX_k}
\end{eqnarray*}
where $\calX_k$ is a quantum predicate variable covering program variables involved.
By (Rp pRec), it suffices to show that
\begin{equation}\label{eq_complete_qPP}
  \trueEquality, \bigg\{ \avpcor{\big}{\fwlp.\big(\CALL\ \proc_i(\lst{y}_i)\big).\calX_i}{\CALL\ \proc_i(\lst{y}_i)}{\calX_i} \bigg\}_{1 \leq i \leq n} \reVdash{\qBE} \bigwedge_{1 \leq i \leq n} \avpcor{\big}{\fwlp.\big(\CALL\ \proc_i(\lst{y}_i)\big).\calX_i}{S_i}{\calX_i}
\end{equation}
where $\proc_i(\lst{y}_i)$ with $1 \leq i \leq n$ and $i \neq k$ are parameterized procedures with bodies $S_i$,
(Here it is required that $1 \leq k \leq n$)
and $\{\calX_i\}_{1 \leq i \leq n}^{i \neq k}$ is a set of quantum predicate variables.
Observe that every Hoare's triple $\pcor{P'}{S'}{Q'}$ in the proof of Ass. (\ref{eq_complete_qPP})
should satisfy the condition that $ \models_\mathbb{I} P' = \fwlp.S'.Q'$:
for those $S' \equiv \CALL\ \proc_i(\lst{p}_i)$,
Hoare's triple $\pcor{P'}{S'}{Q'}$ can be deduced from
$\pcor{\fwlp.\big(\CALL\ \proc_i(\lst{y}_i)\big).\calX_i}{\CALL\ \proc_i(\lst{y}_i)}{\calX_i}$
by using (R Adapt) and (R Subst); for other cases of $S'$, resort to Lem. \ref{lem_sound_complete_EqPL}.
Note that the case analysis entails that such a proof indeed exists.
This completes the proof of the theorem.
\end{proof}

\begin{theorem}[Compact Soundness and Completeness of $\qPP$]\label{thm_cmpt_sound_complete_qPP}
$\qPP$ is compactly complete but not compactly sound.
That is, for any quantum program $S \in \EqRP$ and any $\PQPT$s $P,Q$, we have that
\begin{equation}\label{ass_cmpt_sound_complete_qPP_1}
  \trueEquality \reVdash{\qPP} \pcor{P}{ S }{Q} \impliedby \models_\mathbb{I} P = \fwlp.S.Q
\end{equation}
However, there are quantum program $S \in \EqRP$ and $\PQPT$s $P,Q$ s.t.
\begin{equation}\label{ass_cmpt_sound_complete_qPP_2}
  \trueEquality \reVdash{\qPP} \pcor{P}{ S }{Q} \mbox{ but } \models_\mathbb{I} P \sqsubset \fwlp.S.Q
\end{equation}
\end{theorem}
\begin{proof}
  Note that the proof of Ass. (\ref{ass_cmpt_sound_complete_qPP_1})
  can be adapted from that of Thm. \ref{thm_sound_complete_qPP} ($\impliedby$).
  To prove Ass. (\ref{ass_cmpt_sound_complete_qPP_2}),
  let $\CALL\ P_{\bottom}$ be as defined in Exam. \ref{exam_bottom},
  and $P,Q$ $\PQPT$s with $\models_\mathbb{I} P \sqsubset I$.
  It's trivial that
  \begin{eqnarray*}
    \pcor{P}{ \CALL\ P_{\bottom} }{Q} &\reVdash{\qBE}& \pcor{P}{ \CALL\ P_{\bottom} }{Q}
  \end{eqnarray*}
  By (Rp pRec), it follows that
  \begin{eqnarray*}
    \trueEquality &\reVdash{\qPP}& \pcor{P}{ \CALL\ P_{\bottom} }{Q}
  \end{eqnarray*}
  However, by definition of $\fwlp$ (cf. Tab. \ref{app_tab_wp_wlp}), we have that
  \begin{eqnarray*}
     &\models_\mathbb{I}& P \sqsubset I = \fwlp.(\CALL\ P_{\bottom}).Q
  \end{eqnarray*}
  This completes the proof.
\end{proof}

\subsection{Soundness and completeness of $\qTP$}\label{app_sound_complete_qTP}

\begin{lemma}[Soundness and completeness of $\qBE$]\label{lem_sound_complete_EqPL_tot}
For any $S\in \EqPL$ and any $\PQPT$s $P,Q$, we have that
\begin{equation}\label{ass_sound_complete_EqPL_tot}
  \trueOrder \vdash_{\qBE} \tcor{P}{S}{Q} \iff \models_\mathbb{I} P \sqsubseteq \fwp.S.Q
\end{equation}
in particular,
\begin{equation}\label{ass_cmpt_sound_complete_EqPL_tot}
  \trueEquality \vdash_{\qBE} \tcor{P}{S}{Q} \iff \models_\mathbb{I} P = \fwp.S.Q
\end{equation}
\end{lemma}
\begin{proof}
  Adapted from the proof of Lem. \ref{lem_sound_complete_EqPL} with $\fwp$ in place of $\fwlp$.
\end{proof}

\begin{theorem}[Soundness and completeness of $\qTP$]\label{thm_sound_complete_EqRP}
  For any quantum program $S\in \EqRP$ and any $\PQPT$s $P,Q$, we have that
  \begin{equation}\label{ass_sound_complete_qTP}
    \trueOrder \reVdash{\qTP} \tcor{P}{S}{Q} \iff \models_\mathbb{I} P \sqsubseteq \fwp.S.Q
  \end{equation}
  in particular,
  \begin{equation}\label{ass_cmpt_sound_complete_qTP}
    \trueEquality \vdash_{\qTP} \tcor{P}{S}{Q} \iff \models_\mathbb{I} P = \fwp.S.Q
  \end{equation}
\end{theorem}
\begin{proof}
Similar to the proof of Thm. \ref{thm_sound_complete_qPP}.
\end{proof}

\begin{theorem}[Reasoning about $\EqRP$ with probabilities]\label{thm_reas_EqPR_with_prob}
  For any quantum program $S\in \EqRP$, any $\PQPT$s $P,Q$ and any $\delta\in [0, 1]$, it is the case that
  \begin{equation}\label{ass_reas_EqPR_approx_prob}
    \trueOrder \vdash_{\qTP} \tcor{\delta P}{S}{Q} \mbox{ if and only if } \forall \rho.\ \tr(P\rho) = 1 \implies \tr \big( Q\sem{S}(\rho) \big) \geq \delta
  \end{equation}
  in particular,
  \begin{equation}\label{ass_reas_EqPR_exact_prob}
    \trueEquality \vdash_{\qTP} \tcor{\delta P}{S}{Q} \mbox{ if and only if } \forall \rho.\ \tr(P\rho) = 1 \implies \tr \big( Q\sem{S}(\rho) \big) = \delta
  \end{equation}
\end{theorem}
\begin{proof}

In the following, we only provide the proof for Ass. (\ref{ass_reas_EqPR_approx_prob}).
\big(The proof of Ass. (\ref{ass_reas_EqPR_exact_prob})
can be adapted from the proof of Ass. (\ref{ass_reas_EqPR_approx_prob}) with $\geq$ in place of $=$.\big)

By Thm. \ref{thm_sound_complete_EqRP}, it follows that
\begin{eqnarray*}
  \trueEquality &\vdash_{\qTP}& \tcor{\delta P}{S}{Q}
\end{eqnarray*}
if, and only if,
\begin{eqnarray*}
  &\models_{\mathbb{I}}& \delta P = \fwp.S.Q
\end{eqnarray*}
By Thm. \ref{app_thm_exp}, the last assertion is equivalent to saying that
\begin{eqnarray*}
  &\models_{\mathbb{I}}& \delta P = \sem{S}^*(Q)
\end{eqnarray*}
By Lem. \ref{lem_lowner_compr}, the last assertion is equivalent to saying that
$$
\forall \rho.\ \tr(\delta P \rho) = \tr\big(\sem{S}^*(Q) \rho \big)
$$
By definition of Schr\"{o}dinger-Heisenberg dual, the last assertion is equivalent to saying that
$$
\forall \rho.\ \tr(\delta P \rho) = \tr\big(Q\sem{S}(\rho)\big)
$$
By an easy transformation, the last assertion is equivalent to saying that
$$
\forall \rho.\ \tr(P\rho) = 1 \implies \tr\big( Q\sem{S}(\rho) \big) = \delta
$$
This completes the proof of the theorem.
\end{proof}

\section{Fixed-point Grover's search}\label{app_sec_qSearch}

\subsection{Basic idea of the algorithm}

Let $\ket{s}$ and $\ket{t}$ be the respective starting and target states in a Hilbert space,
where $\ket{s}$ is possibly superposed,
and $\ket{t}$ is a (not necessarily uniform) superposition of all the possible solutions.
The core of the algorithm is to design a search engine
--- a series of unitary operators $\{V_n\}_{n \geq 0}$ (given by an inductive definition)
\begin{equation}\label{def_search_engine}
  \begin{array}{c}
     V_0 \triangleq V, \quad V_{n+1} \triangleq V_n R_s V_n^{\dag} R_t V_n
  \end{array}
\end{equation}
where the $\frac{\pi}{3}$-phase shifts (i.e., unitary operators)
$R_s$ and $R_t$ for $\ket{s}$ and $\ket{t}$ are defined as
\begin{equation}\label{def_R_x}
  \begin{array}{cc}
     R_s \triangleq I - \big(1-\exp(i \frac{\pi}{3})\big)\outprod{s}{s}, & R_t \triangleq I - \big(1-\exp(i \frac{\pi}{3})\big)\outprod{t}{t}
  \end{array}
\end{equation}
such that the resulting state $V_n \ket{s}$ after applying $V_n$ to $\ket{s}$
converges monotonically to $\ket{t}$ as $n$ approaches infinity, that is to say,
\begin{eqnarray*}
  \lim\limits_{n \to \infty }{V_n \ket{s}} &=& \ket{t}
\end{eqnarray*}
As the last step, we fetch information of the solution $\ket{t}$ by a measurement on $V_n \ket{s}$.

Note that $V$ can be selected arbitrarily.
Suppose that $V$ drives the state vector from $s$ to $t$ with a probability of $(1-\epsilon)$, i.e.
\begin{eqnarray*}
  \leng{\bra{t}V\ket{s}}^2 &=& (1-\epsilon)
\end{eqnarray*}
Then it is straightforward but tedious to show that
the resulting state $V_n \ket{s}$ after applying $V_n$ deviates from $t$ with a probability of $\epsilon^{3^n}$, i.e.
\begin{eqnarray*}
  \leng{\bra{t}V_n\ket{s}}^2 &=& (1-\epsilon^{3^n})
\end{eqnarray*}
hence reducing the error probability from $\epsilon$ to $\epsilon^{3^n}$.

\subsection{Quantum programs of the algorithm}

\begin{table}
 \centering
 \begin{minipage}[b]{0.48\textwidth}
  \centering
  \begin{tabular}{rcl}
  \hline
  \multicolumn{3}{c}{$\recDec{\qSearch}{S}$} \\
  \hline
  $S$   & $\triangleq$ & $\ifStat{\Box m\cdot M[q_1] = m \rightarrow S_m}$ \\
  $S_0$ & $\triangleq$ & $q_2 \starequal V$ \\
  $S_1$ & $\triangleq$ & $q_1 \starequal U_{-1};$ \\
      & & $\CALL\ \qSearch;$ \\
      & & $q_2 \starequal R_t;$ \\
      & & $\CALL\ \qSearchDag;$ \\
      & & $q_2 \starequal R_s;$ \\
      & & $\CALL\ \qSearch;$ \\
      & & $q_1 \starequal U_{+1}$ \\
  \hline
\end{tabular}
  \subcaption{$\qSearch$}
  \label{app_qSearch}
  \end{minipage}
  \hfill
  \begin{minipage}[b]{0.48\textwidth}
  \centering
  \begin{tabular}{rcl}
  \hline
  \multicolumn{3}{c}{$\recDec{\qSearchDag}{S'}$}  \\
  \hline
  $S'$   & $\triangleq$ & $\ifStat{\Box m\cdot M[q_1] = m \rightarrow S_m'}$ \\
  $S_0'$ & $\triangleq$ & $q_2 \starequal V^\dag$ \\
  $S_1'$ & $\triangleq$ & $q_1 \starequal U_{-1};$ \\
      & & $\CALL\ \qSearchDag;$     \\
      & & $q_2 \starequal R_s^{\dag};$   \\
      & & $\CALL\ \qSearch;$ \\
      & & $q_2 \starequal R_t^{\dag};$  \\
      & & $\CALL\ \qSearchDag;$   \\
      & & $q_1 \starequal U_{+1}$ \\
  \hline
\end{tabular}
  \subcaption{$\qSearchDag$}
  \label{app_qSearchDag}
  \end{minipage}
  \vspace{-16pt}
\caption{
  $\qSearch$ and $\qSearchDag$ implement the search engine $V_n$ and its adjoint $V_n^{\dag}$.
  }
  \label{app_qSearch_examPrg}
  \vspace{-10pt}
\end{table}

To be precise, we define the involved Hilbert spaces carefully.
Define the search space $\Hh_s$ to be an $N$-dimensional Hilbert space
with orthonormal basis states $\big\{\ket{n}\colon 0 \leq n < N \big\}$,
for encoding a database with solutions represented as $\ket{t}$.
Define the counting space $\Hh_c$ to be a $2^m$-dimensional Hilbert space
with orthonormal basis states $\big\{ \ket{i} \colon 0 \leq i < 2^m \big\}$,
to encode an upper-bounded set of natural numbers.
Note that we use orthonormal basis states of $\Hh_c$
to encode the counter values of the search engine ($m$ should be large enough),
and therefore use a quantum variable of $\Hh_c$
to model the counter instead of a classical counter variable.
We then define $(+i)$-operator $U_{+i}$ of $\Hh_c$ by
$$
\begin{array}{c}
  U_{+i}\colon \ket{x}\rightarrow\ket{(x+i) \mod 2^m},
\end{array}
$$
to model the classical modular $(+i)$-operator,
and similarly for $(-i)$-operator $U_{-i}$.

Now the state space of the search algorithm is $\Hh_c \otimes \Hh_s$.
We set the initial state to be $\ket{n}\ket{s}$.
To achieve this,
we apply unitary operators $U_{+n}$ and $U_s$ to $\ket{0}_{\Hh_c}$ and $\ket{0}_{\Hh_s}$, respectively.
Here, the unitary operator $U_s$ is artificially devised to prepare the starting state $\ket{s}$.

In each step of the search procedure:

(1) Prepare the counting state $\ket{n}$ and starting state $\ket{s}$ by applying $U_{+n}\otimes U_s$ to $\ket{0}\ket{0}$.

(2) Apply the search engine $V_n$, as defined above,
to the starting state $\ket{s}$ with the counting state $\ket{n}$ to determine the recursion depth.
To do so, we first perform the measurement
\begin{eqnarray*}
  M &\triangleq& \bigg\{ M_0 \triangleq \ket{0}\bra{0},\; M_1 \triangleq \sum_{i=1}^{2^m-1} \ket{i}\bra{i} \bigg\}
\end{eqnarray*}
on the counting state $\ket{n}$; then execute the following depending on the measurement outcome.
\begin{itemize}
  \item if the outcome is $0$, the search procedure apply $V$;
  \item otherwise, the search procedure applies $V_{n-1} R_s V_{n-1}^{\dag} R_t V_{n-1}$
\end{itemize}

(3) Measure the resulting state $V_n \ket{s}$ to obtain the information of solution.
Here we can choose a standard (computational) basis measurement,
if elements of the database are encoded as a standard basis state.

Let $q_1$ and $q_2$ be respective quantum variables over $\Hh_c$ and $\Hh_s$.
Recursive quantum procedure $\qSearch$ for the search engine is designed in Tab. \ref{app_qSearch_examPrg}.

\subsection{Partial correctness}\label{app_qSearch_par}

\begin{table}[!htbp]\small

  \centering
  \begin{tabular}{llr}
    \hline
        & $\{ A_{0,0} \outprod{0}{0} \otimes B \}$ & \\
    (a) & $q_2 \starequal V$ $\{ A_{0,0} \outprod{0}{0} \otimes V_0 B V_0^\dag \}$ & (A Unit) \\
    (b) & $\{ \sum_{i=0}^{2^m-1} A_{i,i} \outprod{i}{i} \otimes V_i B V_i^\dag \}$ & $\trueOrder$ \\

    \hline
        & $\{ \sum_{i=1}^{2^m-1} A_{i,i} \outprod{i}{i} \otimes B \}$ \\
    (c) & $q_1 \starequal U_{-1};$ $\{ \sum_{i=0}^{2^m-2} A_{i,i}' \outprod{i}{i} \otimes B \}$ & (A Unit) \\
    (d) & $\CALL\ \qSearch;$ $\{ \sum_{i=0}^{2^m-2} A_{i,i}' \outprod{i}{i} \otimes V_i B V_i^\dag \}$
        & $\Prem_1$, (R Subst) \\
    (e) & $q_2 \starequal R_t;$ $\{ \sum_{i=0}^{2^m-2} A_{i,i}' \outprod{i}{i} \otimes R_t V_i B V_i^\dag R_t^\dag \}$
        & (A Unit) \\
    (f) & $\CALL\ \qSearchDag;$
          $\{ \sum_{i=0}^{2^m-2} A_{i,i}' \outprod{i}{i} \otimes V_i^\dag R_t V_i B V_i^\dag R_t^\dag V_i \}$
        & $\Prem_2$, (R Subst) \\
    (g) & $q_2 \starequal R_s;$
    $\{ \sum_{i=0}^{2^m-2} A_{i,i}' \outprod{i}{i} \otimes R_s V_i^\dag R_t V_i B V_i^\dag R_t^\dag V_i R_s^\dag \}$
        & (A Unit) \\
    (h) & $\CALL\ \qSearch$
    $\{ \sum_{i=0}^{2^m-2} A_{i,i}' \outprod{i}{i} \otimes V_i R_s V_i^\dag R_t V_i B V_i^\dag R_t^\dag V_i R_s^\dag V_i^\dag \}$
        & $\Prem_1$, (R Subst)\\
    (i) & $q_1 \starequal U_{+1}$ $\{ \sum_{i=1}^{2^m-1} A_{i,i} \outprod{i}{i} \otimes V_i B V_i^\dag \}$
         & (A Unit) \\
    (j) & $\{ \sum_{i=0}^{2^m-1} A_{i,i} \outprod{i}{i} \otimes V_i B V_i^\dag \}$
         & $\trueOrder$ \\

    \hline

    (k) & $\pcor{\sum_{i=0}^{2^m-1} A_{i,i} \outprod{i}{i} \otimes B}{S}{\sum_{i=0}^{2^m-1} A_{i,i} \outprod{i}{i} \otimes V_i B V_i^\dag}$
         & (a-b, c-j, R Case) \\
\hline
  \end{tabular}

  \caption{Proof of Ass. (\ref{eq_qSearch_pcor_3}).}
  \label{tab_proof_fgs_par} 
  \vspace{-12pt}
\end{table}

We claim that, on input $\ket{n}_{q_1}\otimes \ket{s}_{q_2}$,
quantum activation statement $\CALL\ \qSearch$ (cf. Tab. \ref{app_qSearch_examPrg} for $\qSearch$)
executes with output $\ket{n}_{q_1}\otimes V_n \ket{s}_{q_2}$ (if terminates).
Formally speaking, the claim can be expressed as a partially correct quantum Hoare's triple:
\begin{eqnarray*}
  &\models_\mathbb{I}& \apcor{\big}{\voutprod{n}{q_1}{n}\otimes \voutprod{s}{q_2}{s}}{\CALL\ \qSearch}{\voutprod{n}{q_1}{n}\otimes V_n \voutprod{s}{q_2}{s} V_n^\dag}
\end{eqnarray*}
By soundness and completeness of $\qPP$, it is to say that
\begin{eqnarray}\label{eq_qSearch_pcor_1}
  \trueOrder &\reVdash{\qPP}& \apcor{\big}{\voutprod{n}{q_1}{n}\otimes \voutprod{s}{q_2}{s}}{\CALL\ \qSearch}{\voutprod{n}{q_1}{n}\otimes V_n \voutprod{s}{q_2}{s} V_n^\dag}
\end{eqnarray}

Let $A$ be a quantum predicate variable on $\Hh_c$,
and $B$ a quantum predicate variable on $\Hh_s$.
The $(i,j)$-component $\mixprod{i}{A}{j}$ of $A$ is abbreviated as $A_{i,j}$,
so $A = \sum_{i,j} A_{i,j} \outprod{i}{j}$.
Following we shall use the main-diagonal elements of $A$ to encode classical information,
which is in accordance with the fact that quantum variable $q_1$ is used classically.
Define $\PQPT$ $A'$ by
\begin{eqnarray*}
  A' &\triangleq& \sum_{i=1}^{2^m-1} A_{i,i} \outprod{i-1}{i-1}
\end{eqnarray*}
To prove Ass. (\ref{eq_qSearch_pcor_1}), by (Subst Rule), together with the simultaneous substitution
$$
\big[\voutprod{n}{q_1}{n} / A,\ \voutprod{s}{q_2}{s} / B\big]
$$
it suffices to show that
\begin{eqnarray}\label{eq_qSearch_pcor_2}
  \trueOrder &\reVdash{\qPP}& \apcor{\bigg}{\sum_{i=0}^{2^m-1} A_{i,i} \outprod{i}{i} \otimes B}{\CALL\ \qSearch}{\sum_{i=0}^{2^m-1} A_{i,i} \outprod{i}{i} \otimes V_i B V_i^\dag}
\end{eqnarray}
Intuitively, the precondition (resp. postcondition) of the Hoare's triple in Ass. (\ref{eq_qSearch_pcor_2}) says that
the control flow arrives at each recursion depth $0 \leq i \leq 2^m-1$ (denoted by variable $q_1$) of procedure $\qSearch$ with probability $A_{i,i}$, and at depth $i$, the state of variable $q_2$ should satisfy
the predicate $B$ (resp. $V_i B V_i^\dag$).

Define a set of premises $\big \{ \Prem_i \big \}_{i = 1,2}$ by
\begin{eqnarray*}
  \Prem_1 &\triangleq& \apcor{\bigg}{\sum_{i=0}^{2^m-1} A_{i,i} \outprod{i}{i} \otimes B}{\CALL\ \qSearch}{\sum_{i=0}^{2^m-1} A_{i,i} \outprod{i}{i} \otimes V_i B V_i^\dag} \\
  \Prem_2 &\triangleq& \apcor{\bigg}{\sum_{i=0}^{2^m-1} A_{i,i} \outprod{i}{i} \otimes B}{\CALL\ \qSearchDag}{\sum_{i=0}^{2^m-1} A_{i,i} \outprod{i}{i} \otimes V_i^\dag B V_i}
\end{eqnarray*}
To prove Ass. (\ref{eq_qSearch_pcor_2}), by (Rp pRec), it suffices to show that
\begin{eqnarray}\label{eq_qSearch_pcor_3}
  \trueOrder, \big\{ \Prem_i \big\}_{i = 1,2} &\reVdash{\qBE}& \apcor{\bigg}{\sum_{i=0}^{2^m-1} A_{i,i} \outprod{i}{i} \otimes B}{S}{\sum_{i=0}^{2^m-1} A_{i,i} \outprod{i}{i} \otimes V_i B V_i^\dag}
\end{eqnarray}
illustrated in Tab. \ref{tab_proof_fgs_par}, and that
\begin{eqnarray}\label{eq_qSearch_pcor_4}
  \trueOrder, \big\{ \Prem_i \big\}_{i = 1,2} &\reVdash{\qBE}& \apcor{\bigg}{\sum_{i=0}^{2^m-1} A_{i,i} \outprod{i}{i} \otimes B}{S'}{\sum_{i=0}^{2^m-1} A_{i,i} \outprod{i}{i} \otimes V_i^\dag B V_i}
\end{eqnarray}
whose proof is similar to that of Ass. (\ref{eq_qSearch_pcor_3}) and is left as an exercise to the reader.

Note that for Hoare's triple (d) in Tab. \ref{tab_proof_fgs_par}, we use the substitution
\begin{equation*}
  \big[ A' \big/ A \big]
\end{equation*}
for (f), we use the substitution
\begin{equation*}
  \big[ A' \big/ A,\ \big( R_t V_i B V_i^\dag R_t^\dag \big) \big/ B \big]
\end{equation*}
for (h), we use the substitution
\begin{equation*}
  \big[ A' \big/ A,\ \big( R_s V_i^\dag R_t V_i B V_i^\dag R_t^\dag V_i R_s^\dag \big) \big/ B \big]
\end{equation*}

\subsection{Total correctness}\label{app_qSearch_tot}

We claim that quantum activation statement $\CALL\ \qSearch$ (cf. Tab. \ref{app_qSearch_examPrg} for $\qSearch$),
on input $\ket{n}_{q_1}\otimes \ket{s}_{q_2}$, always terminates with output $\ket{n}_{q_1}\otimes V_n \ket{s}_{q_2}$.
In a formal way, the claim can be expressed as a totally correct quantum Hoare's triple:
\begin{eqnarray*}
  &\models_\mathbb{I}& \atcor{\big}{\voutprod{n}{q_1}{n}\otimes \voutprod{s}{q_2}{s}}{\CALL\ \qSearch}{\voutprod{n}{q_1}{n}\otimes V_n \voutprod{s}{q_2}{s} V_n^\dag}
\end{eqnarray*}
By soundness and completeness of $\qTP$, it is to say that
\begin{eqnarray}\label{eq_qSearch_tcor_1}
  \trueOrder &\reVdash{\qTP}& \atcor{\big}{\voutprod{n}{q_1}{n}\otimes \voutprod{s}{q_2}{s}}{\CALL\ \qSearch}{\voutprod{n}{q_1}{n}\otimes V_n \voutprod{s}{q_2}{s} V_n^\dag}
\end{eqnarray}
Recall from Subsec. \ref{app_qSearch_par} that $A$ is a quantum predicate variable on $\Hh_c$,
and $B$ a quantum predicate variable on $\Hh_s$.
Note that $(i,j)$-component $\mixprod{i}{A}{j}$ of $A$ is abbreviated as $A_{i,j}$.
To prove Ass. (\ref{eq_qSearch_tcor_1}), by (R Subst), together with the simultaneous substitution
$$
\big[ \voutprod{n}{q_1}{n} / A,\ \voutprod{s}{q_2}{s} / B \big]
$$
it suffices to show that
\begin{eqnarray}\label{eq_qSearch_tcor_2}
  \trueOrder &\reVdash{\qTP}& \atcor{\bigg}{\sum_{i=0}^{2^m-1} A_{i,i} \outprod{i}{i} \otimes B}{\CALL\ \qSearch}{\sum_{i=0}^{2^m-1} A_{i,i} \outprod{i}{i} \otimes V_i B V_i^\dag}
\end{eqnarray}
Define a sequence of $\PQPT$s $\big\{ P_j[A,B] \big\}_{j \geq 0}^{\sqsubseteq}$ by
\begin{eqnarray}\label{eq_qSearch_tcor_3}
  P_j[A,B] &\triangleq&  \left\{
                      \begin{array}{ll}
                        \sum_{i=0}^{j} A_{i,i} \outprod{i}{i} \otimes B & \hbox{if $0 \leq j < 2^m$} \\
                        \sum_{i=0}^{2^m-1} A_{i,i} \outprod{i}{i} \otimes B & \hbox{if $j\geq 2^m$}
                      \end{array}
                    \right.
\end{eqnarray}
Then a set of premises $\big \{ \Prem_i^j \big \}^{j \geq 0}_{i = 1,2}$ is defined by
\begin{eqnarray*}
  \Prem_1^j &\triangleq& \atcor{\bigg}{P_j[A,B]}{\CALL\ \qSearch}{\sum_{i=0}^{2^m-1} A_{i,i} \outprod{i}{i} \otimes V_i B V_i^\dag} \\
  \Prem_2^j &\triangleq& \atcor{\bigg}{P_j[A,B]}{\CALL\ \qSearchDag}{\sum_{i=0}^{2^m-1} A_{i,i} \outprod{i}{i} \otimes V_i^\dag B V_i}
\end{eqnarray*}
By (Rt pRec), it suffices to show, for all $j \geq 0$, that
\begin{eqnarray}
  \trueOrder, \big \{Prem_i^j \big \}_{i = 1,2}
  &\reVdash{\qBE}& \atcor{\bigg}{P_{j+1}[A,B]}{S}{\sum_{i=0}^{2^m-1} A_{i,i} \outprod{i}{i} \otimes V_i B V_i^\dag} \label{eq_qSearch_tcor_4}\\
  \trueOrder, \big \{Prem_i^j \big \}_{i = 1,2}
  &\reVdash{\qBE}& \atcor{\bigg}{P_{j+1}[A,B]}{S'}{\sum_{i=0}^{2^m-1} A_{i,i} \outprod{i}{i} \otimes V_i^\dag B V_i} \label{eq_qSearch_tcor_5}
\end{eqnarray}
We remark that the proof of Ass. (\ref{eq_qSearch_tcor_4})
can be adapted from that of (\ref{eq_qSearch_pcor_3}) [cf. Tab. \ref{tab_proof_fgs_par}]
by replacing the superscript ($2^m - 1$) of some necessary but not all summation operators (including those in the definition of $A'$) with $(j+1)$.
Moreover, Ass. (\ref{eq_qSearch_tcor_5}) can be proved similarly to (\ref{eq_qSearch_tcor_4}).
We leave it as an exercise to the reader.

\section{Recursive quantum Fourier sampling}\label{app_sec_RQFS}

\subsection{Problem description}

Let us first briefly recall recursive quantum Fourier sampling,
following the literature \cite{mckague2012interactive}.
We begin by defining a type of tree.
Let $n,l$ be positive integers and consider a symmetric tree where each node,
except the leaves, has $2^n$ children, and the depth is $l$.
Let the root be labelled by $(\emptyset)$.
The root's children are labelled $(x_1)$ with $x_1\in \{0,1\}^n$.
Each child of $(x_1)$ is, in turn, labelled $(x_1,x_2)$ with $x_2 \in \{0,1\}^n$.
We continue until we have reached the leaves, which are labelled by $(x_1,\ldots,x_l)$.
Thus each node's label can be thought of as a path describing how to find the node from the root.

Now we add the Fourier component to the tree.
We begin by fixing an efficiently computable function $g\colon \{0,1\}^n\rightarrow \{0,1\}$.
With each node of the tree $(x_1,\ldots,x_k)$ we associate a ``secret'' string $s_{(x_1,\ldots,x_k)} \in \{0,1\}^n$.
These secrets are promised to obey
\begin{eqnarray*}
  g( s_{(x_1,\ldots,x_k)} ) &\triangleq&  s_{(x_1,\ldots,x_{k-1})} \cdot x_k \mod 2
\end{eqnarray*}
for $k \geq 1$. \big(Here we take $s_{(x_1,\ldots,x_{k-1})}$ to mean $s_{(\emptyset)}$ if $k = 1$.\big)
In this way, each node's secret encodes one bit of information about its parent's secret.
Suppose that we are given an oracle $o\colon (\{ 0, 1 \}^n)^l \rightarrow \{0,1\}$ which behaves as
\begin{eqnarray*}
  o(x_1,\ldots,x_l) &\triangleq& g\big(s_{(x_1,\ldots,x_l)}\big)
\end{eqnarray*}
Note that $o$ works for the leaves of the tree {\it only}.
Our goal is to find $g(s_{(\emptyset)})$.
This is the recursive Fourier sampling problem ($\RFS$).

\subsection{Quantum solution}

\begin{table}[!htbp]

\centering

  \begin{minipage}[b]{\textwidth}
  \centering
  \begin{tabular}{rcl}
  \hline
  \multicolumn{3}{c}{$\recDec{\RQFS(q,Y)}{\ifStat{\Box m\cdot M[q] = m \rightarrow S_m}}$} \\
  \hline
  $S_0$ & $\triangleq$ & $q \starequal U_{+1};$  \\
  &&$\LOCAL{\bbX[q],Y'};$ \\
  &&$\big( \bbX[q], Y' \big) \starequal H^{\otimes n}\otimes HX;$ \\
  &&$\CALL\ \RQFS(q,Y');$   \\
  &&$\bbX[q] \starequal H^{\otimes n};$ \\
  &&$\big( \bbX[q], Y \big) \starequal \calG;$  \\
  &&$\bbX[q] \starequal H^{\otimes n};$ \\
  &&$\CALL\ \RQFS(q,Y');$           \\
  &&$\big( \bbX[q], Y' \big) \starequal H^{\otimes n}\otimes XH;$ \\
  &&$\RELEASE{\bbX[q],Y'};$ \\
  &&$q \starequal U_{-1}$ \\
  \hline
  $S_1$ & $\triangleq$ & $\big( \bbX[1], \ldots, \bbX[l], Y \big) \starequal \calO$ \\
  \hline
  $S_2$ & $\triangleq$ & $\bottom$       \\
  \hline
  \end{tabular}
  \subcaption{Recursive procedure $\RQFS$.}
  \label{app_RQFS_proc}
  \end{minipage}
  \vfill
  \begin{minipage}[b]{\textwidth}
  \centering
  \begin{tabular}{rcl}
  \hline
  $\MainProgram$ & $\triangleq$ & $(p, Z) \assnequal \ket{0}$; \\
  && $\CALL\ \RQFS(p,Z)$; \\
  && $\ifStat{\Box m\cdot M'[Z] = m \rightarrow \SKIP}$ \\
  \hline
  \end{tabular}
  \subcaption{Main program $\MainProgram$.}
  \label{app_RQFS_main}
  \end{minipage}

\caption{Programming recursive quantum Fourier sampling.}

\label{tab_RQFS_prm} 

\end{table}

We now consider a quantum solution to $\RFS$.
Define the descendant space $\Hh_{d}$ to be the $2^n$-dimensional Hilbert space with orthonormal basis states
--- $\big\{\ket{i}\colon 0 \leq i < 2^n \big\}$ --- to index each of $2^n$ children for any parental node.
Define the counting space $\Hh_{c}$ to be the $2^m$-dimensional Hilbert space with orthonormal basis states
--- $\big\{\ket{i} \colon 0 \leq i < 2^m \big\}$ --- for indexing the depth of the tree, such that $l < 2^m$.
Define $(+i)$-operator $U_{+i}$ of $\Hh_c$ by
$$
U_{+i}\colon \ket{x}\rightarrow\ket{(x+i) \mod 2^m}
$$
and similarly for $(-i)$-operator $U_{-i}$.

Let $p,q$ be quantum (individual) variables over $\Hh_c$, $Y,Y',Z$ quantum variables over $\Hh_2$,
and $\bbX$ (resp., $\bbY$) a quantum array-like variable over $\Hh_{d}$ (resp., $\Hh_2$)
with one argument, say $q$, indexing each component of the array,
s.t. each component $\bbX[q]$ (resp., $\bbY[q]$) acts like a quantum variable over $\Hh_d$ (resp., $\Hh_2$).
We shall treat $\bbX[\ket{i}]$ (resp., $\bbY[\ket{i}]$) as $\bbX[i]$ (resp., $\bbY[i]$) for simplicity.
The starting state space is $\Hh_{p,Z} = \Hh_{c} \otimes \Hh_2$,
and the initial state is $(p,Z) = \ket{0}\otimes \ket{0}$.
The quantum solution is calling the recursive quantum procedure
$$
\RQFS(p/q, Z/Y)
$$
followed by a measurement $M'$ on the resulting qubit of $Z$ \big(viz. $\ket{g(s_{(\emptyset)})}$\big), with
\begin{eqnarray*}
  M' &\triangleq& \big \{ M_0' \triangleq \outprod{0}{0},\; M_1' \triangleq \outprod{1}{1} \big\}
\end{eqnarray*}

In each step of recursive quantum procedure $\RQFS(q, Y)$:

(1) Perform measurement $M$ with
\begin{eqnarray*}
  M &\triangleq& \bigg\{ M_0 \triangleq \sum_{0 \leq i < l} \ket{i}\bra{i},\; M_1 \triangleq \ket{l}\bra{l},\; M_2 \triangleq \sum_{l < i < 2^m} \ket{i}\bra{i} \bigg\}
\end{eqnarray*}
on the counting state $q$; then execute steps (2-4) according to the measurement outcome.

(2) If the outcome is 0, perform steps (21-29).
\begin{description}
  \item[(21)] Increment the value of $q$ by 1. That is, apply $(+1)$-operator $U_{+1}$ to $q$.
  \item[(22)] Introduce ancillas $\bbX[q]$, $\itY'$ in the state $\ket{0} \otimes \ket{0}$.
  \item[(23)] Prepare $\big( \bbX[q], \itY' \big)$ to $\frac{1}{\sqrt{2^n}} \sum_{x = 0}^{2^n - 1} \ket{x} \otimes \frac{1}{\sqrt{2}} ( \ket{0} - \ket{1} )$ by applying $H^{\otimes n}\otimes HX$.
  \item[(24)] Call $\RQFS(q/q,Y'/Y)$.
  \item[(25)] Apply $H^{\otimes n}$ on register $\bbX[q]$.
  \item[(26)] Apply quantum oracle $\calG$ to $\bbX[q]$, $\itY$ with $\calG$ calculating $g$ as
  \begin{eqnarray*}
    \calG\ \ket{s}\ket{y} &\triangleq& \ket{s}\ket{y\oplus g(s)}
  \end{eqnarray*}
  \item[(27)] Return $\bbX[q]$, $\itY'$ to their original state by reversing steps (23-26).
  \item[(28)] Release ancillas $\bbX[q]$, $\itY'$.
  \item[(29)] Return $q$ to its original state. That is, apply $(-1)$-operator $U_{-1}$ to $q$.
\end{description}

(3) If the outcome is 1, apply quantum oracle $\calO$ to $\bbX[1],\ldots,\bbX[l]$, $\itY$, with $\calO$ defined as
\begin{eqnarray*}
  \calO\ \ket{x_1}\ldots\ket{x_l}\ket{y} &\triangleq& \ket{x_1}\ldots\ket{x_l}\ket{y\oplus g(s_{(x_1,\ldots,x_l)})}
\end{eqnarray*}

(4) If the outcome is 2, the procedure collapses (implemented by $\bottom$).

We refine recursive quantum procedure $\RQFS$ and main program $\MainProgram$ by Tab. \ref{tab_RQFS_prm}.

\paragraph{Notations and Definitions.}
Following notations and definitions will be used in the subsequent two subsections.
Let $K$ be a quantum predicate variable over $\Hh_{c}$.
Let $\big\{\ket{i}\big\}_i$ be the computational basis of $\Hh_{q}$.
For notational convenience, the $(i,i)$-component $\mixprod{i}{K}{i}$ of $K$ is abbreviated as $K_{i}$
\big($K = \sum_{i,j} \mixprod{i}{K}{j} \outprod{i}{j}$\big).
We shall use the main-diagonal elements of $K$ to encode classical information.

Define $C(i)$ and $D(i)$ by
\begin{eqnarray*}
  C(i) &\triangleq& \bigotimes_{j=0}^i \Big ( \sum_{k = 0}^{ 2^n - 1 }\voutprod{k}{x_j}{k} \otimes \alpha_j \Big ) \\
  D(i) &\triangleq& \bigotimes_{j=0}^i \Big ( \sum_{k = 0}^{ 2^n - 1 }\voutprod{k}{x_j}{k} \otimes \beta_j \Big )
\end{eqnarray*}
where $\alpha_j$ and $\beta_j$ are defined by
\begin{eqnarray*}
  \alpha_j &\triangleq& \left\{
                      \begin{array}{ll}
                        \voutprod{0}{y_0}{0}, & \hbox{if $j = 0$} \\
                        \voutprod{-}{y_j}{-}, & \hbox{if $1 \leq j \leq l$}
                      \end{array}
                    \right. \\
  \beta_j &\triangleq& \left\{
                      \begin{array}{ll}
                        \voutprod{g(s_{(\emptyset)})}{y_0}{g(s_{(\emptyset)})}, & \hbox{if $j = 0$} \\
                        \voutprod{-}{y_j}{-}, & \hbox{if $1 \leq j \leq l$}
                      \end{array}
                    \right.
\end{eqnarray*}

Define $P(K)$ and $Q(K)$ by
\begin{eqnarray*}
  P(K) &\triangleq& \sum_{i=0}^l  K_i \voutprod{i}{q}{i} \otimes C(i) \\
  Q(K) &\triangleq& \sum_{i=0}^l  K_i \voutprod{i}{q}{i} \otimes D(i)
\end{eqnarray*}

Intuitively, $P(K)$ \big(resp. $Q(K)$\big) says that the control flow arrives
at each recursion depth $0 \leq i \leq l$ (denoted by variable $q$) of algorithm $\RQFS$ with probability $K_i$
(where $i$ corresponds to each level of the tree, in particular, $i=0$ points to the root and $i=l$ to the leaves),
and at depth $i$, variable $\bbY[0]$ lies in the state $\ket{0}$ \big(resp. $\ket{g(s_{(\emptyset)})}$\big),
$\bbY[j]$ with $1 \leq j \leq i$ in $\ket{-}$,
and $\bbX[j]$ with $0 \leq j \leq i$ can lie in any state (by the predicate $I_{x_j}$).

\subsection{Partial correctness}\label{app_RQFS_par}

\begin{table}[!htbp]

\centering

  \begin{minipage}[b]{\textwidth}
  \centering

  \begin{tabular}{llr}
  \hline

  (a) & $\pcor{P(K)}{\CALL\ \RQFS\big(q, \bbY[q]\big)}{Q(K)}$ & $\Prem$ \\

  \hline
      & $\{ \sum_{i = 0}^{l-1} K_i \voutprod{i+1}{q}{i+1} \otimes C(i) \otimes I_{x_{i+1}} \otimes I_{y_{i+1}} \}$ &  \\
  (b) & $\big( \bbX[q], \bbY[q] \big) \assnequal \ket{0}$; & (A Init) \\
      & $\{ \sum_{i = 0}^{l-1} K_i \voutprod{i+1}{q}{i+1} \otimes C(i) \otimes 0_{x_{i+1}} \otimes 0_{y_{i+1}} \}$ & \\
  (c) & $\big( \bbX[q], \bbY[q] \big) \starequal H^{\otimes n}\otimes HX;$ & (A Unit) \\
      & $\{ \sum_{i = 0}^{l-1} K_i \voutprod{i+1}{q}{i+1} \otimes C(i) \otimes I_{x_{i+1}} \otimes \voutprod{-}{y_{i+1}}{-} \}$ & \\
  (d) & $\{ P\big(\sum_{i = 1}^{l} K_{i-1}\outprod{i}{i}\big) \}$ & $\trueOrder$ \\
  (e) & $\CALL\ \RQFS\big( q, \bbY[q] \big)$; $\{ Q\big(\sum_{i = 1}^{l} K_{i-1} \outprod{i}{i}\big) \}$ & (a, R Subst) \\
  (f) & $\bbX[q] \starequal H^{\otimes n};$ & (A Unit) \\
      & $\{ \sum_{i = 1}^{l} K_{i-1} \voutprod{i}{q}{i} \otimes C(i-1)
  \otimes H^{\otimes n} I_{x_i} H^{\otimes n} \otimes \voutprod{-}{y_i}{-} \}$ & \\
  (g) & $\{ \sum_{i = 1}^{l} K_{i-1} \voutprod{i}{q}{i} \otimes \sum_{x_1,\ldots,x_i=0}^{2^n - 1} \bigotimes_{k = 1}^{i-1} ( \outprod{x_k}{x_k} \otimes \alpha_k )$ & $\trueEquality$ \\
      & $\otimes \: \outprod{s_{(x_1,\ldots,x_{i-1})}}{s_{(x_1,\ldots,x_{i-1})}} \otimes \voutprod{-}{y_i}{-} \}$ & \\
  (h) & $\big( \bbX[q],Y \big) \starequal \calG;$ & (A Unit) \\
  & $\{ \sum_{i = 1}^{l} K_{i-1} \voutprod{i}{q}{i} \otimes \sum_{x_1,\ldots,x_i=0}^{2^n - 1} \bigotimes_{k = 1}^{i-1} ( \outprod{x_k}{x_k} \otimes \alpha_k )$ & \\
  & $\otimes \: \outprod{s_{(x_1,\ldots,x_{i-1})}}{s_{(x_1,\ldots,x_{i-1})}} \otimes \voutprod{-}{y_i}{-} \}$ & \\
  (i) & $\bbX[q] \starequal H^{\otimes n};$ & (A Unit) \\
  & $\{ \sum_{i = 1}^{l} K_{i-1} \voutprod{i}{q}{i} \otimes \sum_{x_1,\ldots,x_i=0}^{2^n - 1} \bigotimes_{k = 1}^{i-1} ( \outprod{x_k}{x_k} \otimes \alpha_k )$ & \\
  & $\otimes \: H^{\otimes n} \outprod{s_{(x_1,\ldots,x_{i-1})}}{s_{(x_1,\ldots,x_{i-1})}} H^{\otimes n} \otimes \voutprod{-}{y_i}{-} \}$ & \\
  (j) & $\{ P\big(\sum_{i = 1}^{l} K_{i-1}\outprod{i}{i}\big) \}$ & $\trueEquality$ \\
  (k) & $\CALL\ \RQFS\big(q, \bbY[q]\big)$; $\{ Q\big(\sum_{i = 1}^{l} K_{i-1} \outprod{i}{i}\big) \}$ & (a, R Subst) \\
  (l) & $\big( \bbX[q],Y' \big) \starequal H^{\otimes n}\otimes XH;$ & (A Unit) \\
  & $\{ \sum_{i = 1}^{l} K_{i-1} \voutprod{i}{q}{i} \otimes D(i - 1) \otimes \voutprod{0}{x_i}{0} \otimes \voutprod{0}{y_i}{0} \}$ & \\
  (m) & $\{ \sum_{i = 0}^{l - 1} K_{i} \voutprod{i+1}{q}{i+1} \otimes D(i)\otimes I_{x_{i+1}} \otimes I_{y_{i+1}} \}$ & $\trueOrder$ \\
  \hline
  \end{tabular}

  \subcaption{Partial correctness of the body of variable localization.}
  \label{tab_loc_pcor}
  \end{minipage}

  \vfill

  \begin{minipage}[b]{\textwidth}
  \centering
  \begin{tabular}{llr}
  \hline
      & $\{ \sum_{i = 0}^{l-1} K_i \voutprod{i}{q}{i} \otimes C(i) \}$ &  \\
  (n) & $q \starequal U_{+1};$ $\{ \sum_{i = 0}^{l-1} K_i \voutprod{i+1}{q}{i+1} \otimes C(i) \}$ & (A Unit) \\
  (o) & $\LOCAL{\bbX[p],Y'};\ldots;\RELEASE{\bbX[p],Y'};$ & (b-m, R Loc) \\
      & $\{ \sum_{i = 0}^{l - 1} K_{i} \voutprod{i+1}{q}{i+1} \otimes D(i) \}$ & \\
  (p) & $q \starequal U_{-1}$ $\{ \sum_{i = 0}^{l - 1} K_{i} \voutprod{i}{q}{i} \otimes D(i) \}$ & (A Unit) \\
  (q) & $\{ Q(K) \}$ & $\trueOrder$ \\
  \hline
      & $\{ K_{l} \voutprod{l}{q}{l} \otimes C(l) \}$ &  \\
  (r) & $\big( \bbX[1],\ldots,\bbX[l],Y \big) \starequal \calO$ $\{ K_{l} \voutprod{l}{q}{l} \otimes D(l) \}$ & (A Unit) \\
  (s) & $\{ Q(K) \}$ & $\trueOrder$ \\
  \hline
  (t) & $\pcor{0}{\bottom}{Q(K)}$ & (A Bot, R Order) \\
  \hline
  (u) & $\pcor{P(K)}{\ifStat{\Box m\cdot M[q] = m \rightarrow S_m}}{Q(K)}$ & (n-t, R Case) \\
  \hline
  (v) & $\pcor{P(K)}{\CALL\ \RQFS(q, \bbY[q])}{Q(K)}$ & (a, u, Rp pRec) \\
  \hline
  \end{tabular}
  \subcaption{Partial correctness of recursive procedure $\RQFS$.}
  \label{tab_proc_pcor}
  \end{minipage}

  \vfill

  \begin{minipage}[b]{\textwidth}
  \centering
  \begin{tabular}{llr}
  \hline
      & $\{ I_{p}\otimes I_{Z} \}$ &  \\
  (w) & $(p, Z) \assnequal \ket{0};$ $\{ \voutprod{0}{p}{0}\otimes \voutprod{0}{Z}{0} \}$ & (A Init) \\
  (x) & $\CALL\ \RQFS(p,Z);$ $\{ \voutprod{0}{p}{0}\otimes \voutprod{g(s_{(\emptyset)})}{Z}{g(s_{(\emptyset)})} \}$ & TBA \\
  (y) & $\ifStat{\Box m\cdot M'[Z] = m \rightarrow \SKIP}$ $\{ \voutprod{0}{p}{0}\otimes \voutprod{g(s_{(\emptyset)})}{Z}{g(s_{(\emptyset)})} \}$ & (A Skip, R Case) \\
  \hline
  \end{tabular}
  \subcaption{Partial correctness of main program $\MainProgram$.}
  \label{tab_main_pcor}
  \end{minipage}

\caption{Partial correctness of recursive quantum Fourier sampling.}

\label{tab_RQFS_pcor} 

\end{table}

We claim that, on any input,
the main program $\MainProgram$ executes with output $\ket{0}_p \otimes \ket{g(s_{(\emptyset)})}_Z$ (if it terminates).
In a formal fashion, it is claimed that
\begin{proposition}\label{prop_main_pcor}
It is the case that
\begin{eqnarray}\label{ass_main_pcor}
  \trueOrder &\vdash_{\qPP}& \apcor{\big}{I_{p}\otimes I_{Z}}{\MainProgram}{\voutprod{0}{p}{0}\otimes \voutprod{g(s_{(\emptyset)})}{Z}{g(s_{(\emptyset)})}}
\end{eqnarray}
\end{proposition}
\begin{proof}
The proof of Ass. (\ref{ass_main_pcor}) is shown in Tab. \ref{tab_main_pcor}, where ``TBA'' means ``To Be Announced''.
To prove Hoare's triple (x), by (R Adap), it is sufficient to show that
\begin{eqnarray*}
  \trueOrder &\vdash_{\qPP}& \apcor{\big}{\voutprod{0}{q}{0}\otimes \voutprod{0}{y_0}{0}}{\CALL\ \RQFS\big(q,\bbY[q]\big)}{\voutprod{0}{q}{0}\otimes \voutprod{g(s_{(\emptyset)})}{y_0}{g(s_{(\emptyset)})}}
\end{eqnarray*}
By definition of $P(K)$ and $Q(K)$, it suffices to show that
\begin{eqnarray*}
  \trueOrder &\vdash_{\qPP}& \apcor{\big}{P(\voutprod{0}{q}{0})}{\CALL\ \RQFS\big(q, \bbY[q]\big)}{Q(\voutprod{0}{q}{0})}
\end{eqnarray*}
By (R Subst), together with the substitution $[\voutprod{0}{q}{0}/K]$, it suffices to show that
\begin{eqnarray}\label{ass_RQFS_par}
  \trueOrder &\vdash_{\qPP}& \pcor{P(K)}{\CALL\ \RQFS\big(q, \bbY[q]\big)}{Q(K)}
\end{eqnarray}
The proof of Ass. (\ref{ass_RQFS_par}) is shown in Tab. \ref{tab_proc_pcor},
where partial correctness of the body of variable localization,
i.e. proof of Hoare's triple (o), is shown in Tab. \ref{tab_loc_pcor}.

To see Ass. (g), by definition of $C(i-1)$ and the fact that
\begin{eqnarray*}
  I_{x_i} &=& \sum_{k=0}^{2^n-1} \voutprod{k}{x_i}{k}
\end{eqnarray*}
we remark that
$$
\sum_{i = 1}^{l} K_{i-1} \voutprod{i}{q}{i} \otimes C(i-1) \otimes H^{\otimes n} I_{x_i} H^{\otimes n} \otimes \voutprod{-}{y_i}{-}
$$
is equivalent to
$$
\begin{array}{c}
\sum_{i = 1}^{l} K_{i-1} \voutprod{i}{q}{i} \otimes \sum_{x_1,\ldots,x_i=0}^{2^n - 1} \bigotimes_{k = 1}^{i-1} ( \outprod{x_k}{x_k} \otimes \alpha_k ) \\
\otimes \: H^{\otimes n} (-1)^{g(s_{(x_1,\ldots,x_i)})} \outprod{x_i}{x_i} (-1)^{g(s_{(x_1,\ldots,x_i)})} H^{\otimes n}
\otimes \voutprod{-}{y_i}{-}
\end{array}
$$
By Eq.
\begin{eqnarray*}
  g(s_{(x_1,\ldots,x_i)}) &=& s_{(x_1,\ldots,x_{i-1})} \cdot x_i
\end{eqnarray*}
it is to say
$$
\begin{array}{c}
\sum_{i = 1}^{l} K_{i-1} \voutprod{i}{q}{i} \otimes \sum_{x_1,\ldots,x_i=0}^{2^n - 1} \bigotimes_{k = 1}^{i-1} ( \outprod{x_k}{x_k} \otimes \alpha_k ) \\
\otimes \: H^{\otimes n} (-1)^{s_{(x_1,\ldots,x_{i-1})} \cdot x_i} \outprod{x_i}{x_i} (-1)^{s_{(x_1,\ldots,x_{i-1})} \cdot x_i} H^{\otimes n} \otimes \voutprod{-}{y_i}{-}
\end{array}
$$
By the fact that
\begin{eqnarray*}
   \sum_x H^{\otimes n} (-1)^{x\cdot y} \ket{x} &=& \ket{y}
\end{eqnarray*}
it is equivalent to saying that
$$
\begin{array}{c}
\sum_{i = 1}^{l} K_{i-1} \voutprod{i}{q}{i} \otimes \sum_{x_1,\ldots,x_i=0}^{2^n - 1} \bigotimes_{k = 1}^{i-1} ( \outprod{x_k}{x_k} \otimes \alpha_k ) \\
\otimes \: \outprod{s_{(x_1,\ldots,x_{i-1})}}{s_{(x_1,\ldots,x_{i-1})}} \otimes \voutprod{-}{y_i}{-}
\end{array}
$$

To see Ass. (j), by the fact
\begin{eqnarray*}
  H^{\otimes n} \ket{y} &=& \sum_x (-1)^{x\cdot y} \ket{x}
\end{eqnarray*}
we remark that
$$
\begin{array}{c}
\sum_{i = 1}^{l} K_{i-1} \voutprod{i}{q}{i} \otimes \sum_{x_1,\ldots,x_i=0}^{2^n - 1} \bigotimes_{k = 1}^{i-1} ( \outprod{x_k}{x_k} \otimes \alpha_k ) \\
\otimes \: H^{\otimes n} \outprod{s_{(x_1,\ldots,x_{i-1})}}{s_{(x_1,\ldots,x_{i-1})}} H^{\otimes n} \otimes \voutprod{-}{y_i}{-}
\end{array}
$$
is equivalent to
$$
\begin{array}{c}
\sum_{i = 1}^{l} K_{i-1} \voutprod{i}{q}{i} \otimes \sum_{x_1,\ldots,x_i=0}^{2^n - 1} \bigotimes_{k = 1}^{i-1} ( \outprod{x_k}{x_k} \otimes \alpha_k ) \\
\otimes \: (-1)^{s_{(x_1,\ldots,x_{i-1})} \cdot x_i} \outprod{x_i}{x_i} (-1)^{s_{(x_1,\ldots,x_{i-1})} \cdot x_i} \otimes \voutprod{-}{y_i}{-}
\end{array}
$$
An easy calculation yields
$$
\begin{array}{c}
\sum_{i = 1}^{l} K_{i-1} \voutprod{i}{q}{i} \otimes \sum_{x_1,\ldots,x_i=0}^{2^n - 1} \bigotimes_{k = 1}^{i-1} ( \outprod{x_k}{x_k} \otimes \alpha_k ) \otimes \outprod{x_i}{x_i} \otimes \voutprod{-}{y_i}{-}
\end{array}
$$
By definition of $C( i )$, it is equivalent to saying that
$$
\sum_{i = 1}^{l} K_{i-1} \voutprod{i}{q}{i} \otimes C(i)
$$
By definition of $P(K)$, it is exactly $P\big(\sum_{i = 1}^{l} K_{i-1}\outprod{i}{i}\big)$.

This completes the proof.
\end{proof}

\subsection{Total correctness}\label{app_RQFS_tot}

\begin{table}[!htbp]

\centering

  \begin{minipage}[b]{\textwidth}
  \centering

  \begin{tabular}{llr}
  \hline

  (a) & $\tcor{P_h(K)}{\CALL\ \RQFS\big(q, \bbY[q]\big)}{Q(K)}$ & $\Prem_h$ \\

  \hline
      & $\langle \sum_{i = l-h}^{l-1} K_i \voutprod{i+1}{q}{i+1} \otimes C(i) \otimes I_{x_{i+1}} \otimes I_{y_{i+1}} \rangle$ &  \\
  (b) & $\big( \bbX[q],\bbY[q] \big) \assnequal \ket{0}$; & (A Init) \\
      & $\langle \sum_{i = l-h}^{l-1} K_i \voutprod{i+1}{q}{i+1} \otimes C(i) \otimes 0_{x_{i+1}} \otimes 0_{y_{i+1}} \rangle$ & \\
  (c) & $\big(\bbX[q], \bbY[q]\big) \starequal H^{\otimes n}\otimes HX;$ & (A Unit) \\
      & $\langle \sum_{i = l-h}^{l-1} K_i \voutprod{i+1}{q}{i+1} \otimes C(i) \otimes I_{x_{i+1}} \otimes \voutprod{-}{y_{i+1}}{-} \rangle$ & \\
  (d) & $\langle P_h\big(\sum_{i = l+1-h}^{l} K_{i-1}\outprod{i}{i}\big) \rangle$ & $\trueOrder$ \\
  (e) & $\CALL\ \RQFS\big(q,\bbY[q]\big)$; $\langle Q\big(\sum_{i = l+1-h}^{l} K_{i-1} \outprod{i}{i}\big) \rangle$ & (a, R Subst) \\
  (f) & $\bbX[q] \starequal H^{\otimes n};$ & (A Unit) \\
      & $\langle \sum_{i = l+1-h}^{l} K_{i-1} \voutprod{i}{q}{i} \otimes C(i-1)
  \otimes H^{\otimes n} I_{x_i} H^{\otimes n} \otimes \voutprod{-}{y_i}{-} \rangle$ & \\
  (g) & $\langle \sum_{i = l+1-h}^{l} K_{i-1} \voutprod{i}{q}{i} \otimes \sum_{x_1,\ldots,x_i=0}^{2^n - 1} \bigotimes_{k = 1}^{i-1} ( \outprod{x_k}{x_k} \otimes \alpha_k )$ & $\trueEquality$ \\
      & $\otimes \: \outprod{s_{(x_1,\ldots,x_{i-1})}}{s_{(x_1,\ldots,x_{i-1})}} \otimes \voutprod{-}{y_i}{-} \rangle$ & \\
  (h) & $\big( \bbX[q],Y \big) \starequal \calG;$ & (A Unit) \\
  & $\langle \sum_{i = l+1-h}^{l} K_{i-1} \voutprod{i}{q}{i} \otimes \sum_{x_1,\ldots,x_i=0}^{2^n - 1} \bigotimes_{k = 1}^{i-1} ( \outprod{x_k}{x_k} \otimes \alpha_k )$ & \\
  & $\otimes \: \outprod{s_{(x_1,\ldots,x_{i-1})}}{s_{(x_1,\ldots,x_{i-1})}} \otimes \voutprod{-}{y_i}{-} \rangle$ & \\
  (i) & $\bbX[q] \starequal H^{\otimes n};$ & (A Unit) \\
  & $\langle \sum_{i = l+1-h}^{l} K_{i-1} \voutprod{i}{q}{i} \otimes \sum_{x_1,\ldots,x_i=0}^{2^n - 1} \bigotimes_{k = 1}^{i-1} ( \outprod{x_k}{x_k} \otimes \alpha_k )$ & \\
  & $\otimes \: H^{\otimes n} \outprod{s_{(x_1,\ldots,x_{i-1})}}{s_{(x_1,\ldots,x_{i-1})}} H^{\otimes n} \otimes \voutprod{-}{y_i}{-} \rangle$ & \\
  (j) & $\langle P_h\big(\sum_{i = l+1-h}^{l} K_{i-1}\outprod{i}{i}\big) \rangle$ & $\trueEquality$ \\
  (k) & $\CALL\ \RQFS\big(q, \bbY[q]\big)$; $\langle Q\big(\sum_{i = l+1-h}^{l} K_{i-1} \outprod{i}{i}\big) \rangle$ & (a, R Subst) \\
  (l) & $\big( \bbX[q],Y' \big) \starequal H^{\otimes n}\otimes XH;$ & (A Unit) \\
  & $\langle \sum_{i = l+1-h}^{l} K_{i-1} \voutprod{i}{q}{i} \otimes D(i - 1) \otimes \voutprod{0}{x_i}{0} \otimes \voutprod{0}{y_i}{0} \rangle$ & \\
  (m) & $\langle \sum_{i = l-h}^{l - 1} K_{i} \voutprod{i+1}{q}{i+1} \otimes D(i)\otimes I_{x_{i+1}} \otimes I_{y_{i+1}} \rangle$ & $\trueOrder$ \\
  \hline
  \end{tabular}

  \subcaption{Total correctness of the body of variable localization.}
  \label{tab_loc_tcor}
  \end{minipage}

  \vfill

  \begin{minipage}[b]{\textwidth}
  \centering
  \begin{tabular}{llr}
  \hline
      & $\langle \sum_{i = l-h}^{l-1} K_i \voutprod{i}{q}{i} \otimes C(i) \rangle$ &  \\
  (n) & $q \starequal U_{+1};$ $\langle \sum_{i = l-h}^{l-1} K_i \voutprod{i+1}{q}{i+1} \otimes C(i) \rangle$ & (A Unit) \\
  (o) & $\LOCAL{\bbX[p],Y'};\ldots;\RELEASE{\bbX[p],Y'};$ & (b-m, R Loc) \\
      & $\langle \sum_{i = l-h}^{l - 1} K_{i} \voutprod{i+1}{q}{i+1} \otimes D(i) \rangle$ & \\
  (p) & $q \starequal U_{-1}$ $\langle \sum_{i = l-h}^{l - 1} K_{i} \voutprod{i}{q}{i} \otimes D(i) \rangle$ & (A Unit) \\
  (q) & $\langle Q(K) \rangle$ & $\trueOrder$ \\
  \hline
  \end{tabular}
  \subcaption{Proof of Ass. (\ref{ass_rqfs_tcor_3}).}
  \label{tab_S_0_tcor}
  \end{minipage}

\caption{Total correctness of recursive quantum Fourier sampling.}

\label{tab_RQFS_tcor} 

\end{table}

We claim that, on any input, the main program $main$ always terminates with output
$\ket{0}_p \otimes \ket{g(s_{(\emptyset)})}_Z$.
In a formal fashion, it is claimed that
\begin{proposition}\label{prop_main_tcor}
It is the case that
\begin{eqnarray*}
  \trueOrder &\vdash_{\qTP}& \atcor{\big}{I_{p}\otimes I_{Z}}{main}{\voutprod{0}{p}{0}\otimes \voutprod{g(s_{(\emptyset)})}{Z}{g(s_{(\emptyset)})}}
\end{eqnarray*}
\end{proposition}
\begin{proof}
The proof is as for Prop. \ref{prop_main_pcor}, with exception of the following assertion
\begin{eqnarray}\label{ass_rqfs_tcor_1}
  \trueOrder &\vdash_{\qTP}& \tcor{P(K)}{\CALL\ \RQFS\big(q, \bbY[q]\big)}{Q(K)}
\end{eqnarray}

Define a sequence of $\PQPT$s $\big\{ P_h(K) \big\}_{h \geq 0}$ by
\begin{eqnarray*}
  P_h(K) &\triangleq& \left\{
                    \begin{array}{ll}
                      \sum_{i=l+1-h}^l K_i \voutprod{i}{q}{i} \otimes C(i) & \hbox{$0 \leq h < l$} \\
                      P(K), & \hbox{otherwise}
                    \end{array}
                  \right.
\end{eqnarray*}
It's easy to see that
\begin{eqnarray*}
  &\models_\mathbb{I}& P_0(K) = 0 \\
  &\models_\mathbb{I}& P_h(K) \sqsubseteq P_{h+1}(K), \quad \forall h \geq 0
\end{eqnarray*}
For the sake of space savings, define a set of premises $\big\{ \Prem_h \big\}_{h \geq 0}$ by
\begin{eqnarray*}
  \Prem_h &\triangleq& \atcor{\big}{P_h(K)}{\CALL\ \RQFS(q, \bbY[q])}{Q(K)}, \quad \forall h \geq 0
\end{eqnarray*}
To prove Ass. (\ref{ass_rqfs_tcor_1}), by (Rt pRec), it suffices to show that
\begin{eqnarray}\label{ass_rqfs_tcor_2}
  \trueOrder, \Prem_h &\vdash_{\qBE}& \atcor{\big}{P_{h+1}(K)}{\ifStat{\Box m\cdot M[q] = m \rightarrow S_m}}{Q(K)}
\end{eqnarray}
for all $h \geq 0$.
The case of $h \geq l$ has been shown in the proof of Prop. \ref{prop_main_pcor},
while the remaining cases --- $0 \leq h < l$ --- can be uniformly dealt with as follows.

Fix $0 \leq h < l$.
To prove Ass. (\ref{ass_rqfs_tcor_2}), by (R Case), it suffices to show that
\begin{eqnarray}
  \trueOrder, \mathit{Prem_h} & \reVdash{\qBE} & \atcor{\bigg}{\sum_{i=l-h}^{l-1} K_i \voutprod{i}{q}{i} \otimes C(i)}{S_0}{Q(K)} \label{ass_rqfs_tcor_3} \\
  \trueOrder &\vdash_{\qBE}& \atcor{\big}{K_l \voutprod{l}{q}{l} \otimes C(l)}{S_1}{Q(K)} \label{ass_rqfs_tcor_4} \\
  \trueOrder &\vdash_{\qBE}& \atcor{\big}{0}{S_2}{Q(K)} \label{ass_rqfs_tcor_5}
\end{eqnarray}
where Ass. (\ref{ass_rqfs_tcor_5}) follows from (A Bot),
the proof of (\ref{ass_rqfs_tcor_4}) is as in proof of Prop. \ref{prop_main_pcor},
and, finally, the proof of (\ref{ass_rqfs_tcor_3}) can be adapted from that of Prop. \ref{prop_main_pcor}
by replacing $\sum_{i = 0}^{l-1}$ and $\sum_{i = 1}^{l}$ with $\sum_{i = l-h}^{l-1}$ and $\sum_{i = l+1-h}^{l}$, respectively (cf. Tab. \ref{tab_S_0_tcor}).
\end{proof}

\end{document}
\endinput